\documentclass[11pt]{article}
\usepackage[margin=.8in,left=.8in]{geometry}
\usepackage{amsmath}
\usepackage{amsfonts}
\usepackage{amssymb}
\usepackage{amsthm}
\usepackage{mathtools}
\usepackage{subcaption}
\usepackage{graphicx}
\usepackage{color}
\usepackage{xcolor}
\usepackage{xfrac}
\usepackage{stmaryrd}
\usepackage[utf8]{inputenc}
\usepackage{rotating}
\usepackage{hyperref}
\usepackage[all]{xy}

\usepackage{tikz}
\usepackage{tikz-cd}
\usepackage{lipsum}
\usepackage{adjustbox}

\usepackage{multirow}

\usepackage{stmaryrd}    

\usepackage{enumerate} 

\usepackage{helvet}   

\usepackage{mathptmx}
\usepackage{amsmath}
\usepackage{graphicx}

\usepackage{floatflt}  
\usepackage{array}     
\newcolumntype{L}[1]{>{\raggedright\let\newline\\\arraybackslash\hspace{0pt}}m{#1}}
\newcolumntype{C}[1]{>{\centering\let\newline\\\arraybackslash\hspace{0pt}}m{#1}}
\newcolumntype{R}[1]{>{\raggedleft\let\newline\\\arraybackslash\hspace{0pt}}m{#1}}

\newcommand{\mapsup}{\mbox{$\;\,$\begin{rotate}{90}$\!\!\!\mapsto$\end{rotate}}}

\makeatletter
\newcommand{\raisemath}[1]{\mathpalette{\raisem@th{#1}}}
\newcommand{\raisem@th}[3]{\raisebox{#1}{$#2#3$}}
\makeatother

\usepackage{tabularx}   

\usepackage[new]{old-arrows}   

\newdir{> }{{}*!/10pt/@{>}}

\makeatletter
\newif\if@sup
\newtoks\@sups
\def\append@sup#1{\edef\act{\noexpand\@sups={\the\@sups #1}}\act}%
\def\reset@sup{\@supfalse\@sups={}}%
\def\mk@scripts#1#2{\if #2/ \if@sup ^{\the\@sups}\fi \else%
  \ifx #1_ \if@sup ^{\the\@sups}\reset@sup \fi {}_{#2}%
  \else \append@sup#2 \@suptrue \fi%
  \expandafter\mk@scripts\fi}
\def\tensor#1#2{\reset@sup#1\mk@scripts#2_/}
\def\multiscripts#1#2#3{\reset@sup{}\mk@scripts#1_/#2%
  \reset@sup\mk@scripts#3_/}
\makeatother

\makeatletter
\newbox\slashbox \setbox\slashbox=\hbox{$/$}
\def\itex@pslash#1{\setbox\@tempboxa=\hbox{$#1$}
  \@tempdima=0.5\wd\slashbox \advance\@tempdima 0.5\wd\@tempboxa
  \copy\slashbox \kern-\@tempdima \box\@tempboxa}
\def\slash{\protect\itex@pslash}
\makeatother

\def\clap#1{\hbox to 0pt{\hss#1\hss}}
\def\mathllap{\mathpalette\mathllapinternal}
\def\mathrlap{\mathpalette\mathrlapinternal}
\def\mathclap{\mathpalette\mathclapinternal}
\def\mathllapinternal#1#2{\llap{$\mathsurround=0pt#1{#2}$}}
\def\mathrlapinternal#1#2{\rlap{$\mathsurround=0pt#1{#2}$}}
\def\mathclapinternal#1#2{\clap{$\mathsurround=0pt#1{#2}$}}

\let\oldroot\root
\def\root#1#2{\oldroot #1 \of{#2}}
\renewcommand{\sqrt}[2][]{\oldroot #1 \of{#2}}

\DeclareSymbolFont{symbolsC}{U}{txsyc}{m}{n}
\SetSymbolFont{symbolsC}{bold}{U}{txsyc}{bx}{n}
\DeclareFontSubstitution{U}{txsyc}{m}{n}

\DeclareSymbolFont{stmry}{U}{stmry}{m}{n}
\SetSymbolFont{stmry}{bold}{U}{stmry}{b}{n}

\DeclareFontFamily{OMX}{MnSymbolE}{}
\DeclareSymbolFont{mnomx}{OMX}{MnSymbolE}{m}{n}
\SetSymbolFont{mnomx}{bold}{OMX}{MnSymbolE}{b}{n}
\DeclareFontShape{OMX}{MnSymbolE}{m}{n}{
    <-6>  MnSymbolE5
   <6-7>  MnSymbolE6
   <7-8>  MnSymbolE7
   <8-9>  MnSymbolE8
   <9-10> MnSymbolE9
  <10-12> MnSymbolE10
  <12->   MnSymbolE12}{}

\makeatletter
\def\Decl@Mn@Delim#1#2#3#4{%
  \if\relax\noexpand#1%
    \let#1\undefined
  \fi
  \DeclareMathDelimiter{#1}{#2}{#3}{#4}{#3}{#4}}
\def\Decl@Mn@Open#1#2#3{\Decl@Mn@Delim{#1}{\mathopen}{#2}{#3}}
\def\Decl@Mn@Close#1#2#3{\Decl@Mn@Delim{#1}{\mathclose}{#2}{#3}}
\Decl@Mn@Open{\llangle}{mnomx}{'164}
\Decl@Mn@Close{\rrangle}{mnomx}{'171}
\Decl@Mn@Open{\lmoustache}{mnomx}{'245}
\Decl@Mn@Close{\rmoustache}{mnomx}{'244}
\makeatother

\makeatletter
\DeclareRobustCommand\widecheck[1]{{\mathpalette\@widecheck{#1}}}
\def\@widecheck#1#2{%
    \setbox\z@\hbox{\m@th$#1#2$}%
    \setbox\tw@\hbox{\m@th$#1%
       \widehat{%
          \vrule\@width\z@\@height\ht\z@
          \vrule\@height\z@\@width\wd\z@}$}%
    \dp\tw@-\ht\z@
    \@tempdima\ht\z@ \advance\@tempdima2\ht\tw@ \divide\@tempdima\thr@@
    \setbox\tw@\hbox{%
       \raise\@tempdima\hbox{\scalebox{1}[-1]{\lower\@tempdima\box
\tw@}}}%
    {\ooalign{\box\tw@ \cr \box\z@}}}
\makeatother

\makeatletter
\def\udots{\mathinner{\mkern2mu\raise\p@\hbox{.}
\mkern2mu\raise4\p@\hbox{.}\mkern1mu
\raise7\p@\vbox{\kern7\p@\hbox{.}}\mkern1mu}}
\makeatother

\newcommand{\underoverset}[3]{\underset{#1}{\overset{#2}{#3}}}

\newcommand{\gt}{>}
\newcommand{\lt}{<}

\renewcommand{\(}{\begin{equation}}
\renewcommand{\)}{\end{equation}}
\newcommand{\bea}{\begin{eqnarray*}}
\newcommand{\eea}{\end{eqnarray*}}

\usepackage{cleveref}

\crefformat{section}{\S#2#1#3} 
\crefformat{subsection}{\S#2#1#3}
\crefformat{subsubsection}{\S#2#1#3}

\theoremstyle{italics}
\newtheorem{theorem}{Theorem}[section]
\newtheorem{lemma}[theorem]{Lemma}
\newtheorem{prop}[theorem]{Proposition}
\newtheorem{cor}[theorem]{Corollary}

\theoremstyle{definition}
\newtheorem{defn}[theorem]{Definition}

\newtheorem{example}[theorem]{Example}

\newtheorem{remark}[theorem]{Remark}
\newtheorem{note[theorem]}{Note}

\usepackage{amsfonts}

\begin{document}

\title{
Super-exceptional geometry: origin of heterotic M-theory
\\
and super-exceptional embedding construction of M5
}

\author{Domenico Fiorenza, \; Hisham Sati, \; Urs Schreiber}
%

\maketitle

\begin{abstract}
In the quest for the mathematical formulation of M-theory,
we consider three major open problems:
a first-principles construction of the  single (abelian) M5-brane Lagrangian density,
the origin of the gauge field in heterotic M-theory,
and
the supersymmetric
enhancement of exceptional M-geometry.
By combining techniques from homotopy theory and from supergeometry
to what we call super-exceptional geometry within super-homotopy theory, we
present an elegant joint solution to all three problems.
This leads to a unified description of the Nambu-Goto,
Perry-Schwarz, and topological Yang-Mills Lagrangians in the topologically
nontrivial setting. After explaining how
charge quantization of the C-field in Cohomotopy reveals D'Auria-Fr{\'e}'s
``hidden supergroup'' of 11d supergravity as the super-exceptional target space, in the sense of Bandos,
for M5-brane sigma-models,
we prove, in exceptional generalization
of the doubly-supersymmetric super-embedding formalism, that a Perry-Schwarz-type Lagrangian
for single (abelian) $\mathcal{N} = (1,0)$ M5-branes emerges as the
super-exceptional trivialization of the M5-brane cocycle along the
super-exceptional embedding of the ``half'' M5-brane locus,
super-exceptionally compactified on the Ho{\v r}ava-Witten
circle fiber.
From inspection of the resulting 5d super Yang-Mills Lagrangian
we find that the extra fermion field
appearing in super-exceptional M-geometry,
whose physical interpretation had remained open,
is the M-theoretic avatar of the gaugino field.
\end{abstract}


\medskip

\vspace{-.6cm}

\tableofcontents

\vspace{-0cm}

\vfill

\section{Introduction}

An actual formulation of \emph{M-theory} remains a fundamental
 open problem from physical and mathematical points of view
 (see \cite[Sec. 12]{Moore14}\cite[Sec. 2]{HSS18}).
We had initiated in \cite{FSS13} a program
of attacking this problem,
based on universal constructions in
\emph{super-homotopy theory} (see \cite{FSS19a} for review),
and used this to find first-priniciples derivations of various
aspects expected
of M-theory (see \cite{FSS16b}\cite{FSS19b}\cite{FSS19c}).
In this paper we look to carry this further and
consider the following three major sub-problems:
\begin{enumerate}[{\bf 1.}]
  \vspace{-2mm}
  \item \hyperlink{ThirdOpenProblem}{Provide a systematic construction of M5-brane Lagrangians}.
  \vspace{-2mm}
  \item \hyperlink{FirstOpenProb}{Identify M-theory avatar
   degrees of freedom of the gauge field and gaugino field appearing on MO9-planes}.
  \vspace{-2mm}
  \item \hyperlink{SecondOpenProb}{Extend exceptional M-geometry to
    the supergeometric setting in a natural and constructive manner}.
\end{enumerate}
 \vspace{-2mm}
 We present a unified approach which leads to an elegant \hyperlink{JointSolution}{joint solution} to all three at once,
 using the two principles, \hyperlink{FirstPrin}{super-geometry} and
\hyperlink{SecondPrin}{super-homotopy theory}, explained below.


\medskip

\hypertarget{ThirdOpenProblem}{}
\noindent {\bf First open problem: M5-brane Lagrangians.}
A widely recognized open sub-problem
is the identification of the 6d superconformal field theory
(see \cite{Moore12})
on coincident M5-branes (see \cite[Sec. 3]{Lambert19}),
whose dimensional reduction to four dimensions is expected to elucidate
deep aspects of non-perturbative 4-dimensional Yang-Mills theory;
and not only those of theoretical interest such as
$\mathcal{N}= 1$ Montonen-Olive duality (see \cite{Witten07}),
but also of profound interest in phenomenology,
such as for the prediction of hadron spectra in
confined quantum chromodynamics
(\cite[Sec. 4]{Witten98b}\cite{SakaiSugimoto04}\cite{SakaiSugimoto05},
see
\cite{Rebhan14}\cite{Guijosa16}\cite{Sugimoto16}).

\begin{itemize}
\vspace{-3mm}
\item Typically, it is asserted that this is an open problem only for $N \geq 2$
coincident M5-branes, while the special case of a single M5-brane is well-known.
Indeed, there is a non-covariant Lagrangian formulation
\cite{PerrySchwarz97} \cite{Schwarz97}\cite{APPS97}
adapted to M5-s wrapped on the M-theory circle fiber,
as well as a covariant version at the cost of introducing an auxiliary field
\cite{PastiSorokinTonin97}\cite{BLNPST97}.
Both of these involved some ingenuity in their
construction which makes them look somewhat baroque.
Indeed, their double dimensional reduction
reproduces the D4-brane Lagrangian,
and hence the 5d super Yang-Mills + topological Yang-Mills Lagrangian,
only up to an intricate field redefinition
\cite[Sec. 6]{APPS97}\cite[Sec. 6 \& App. A]{APPS97b}.
\vspace{-2mm}
\item Such complications, already in the formulation of the
base case of a theory whose expected generalization
remains elusive, may indicate that the natural
perspective on the problem has not been identified yet.
What has been missing is a derivation of the
M5-brane Lagrangian systematically from first principles
of M-theory, with manifest dimensional reduction
to the D4-brane.
\end{itemize}

\vspace{-2mm}
\hypertarget{FirstOpenProb}{}
\noindent {\bf Second open problem: Heterotic gauge enhancement.}
The non-perturbative completion of heterotic string
theory has famously been argued \cite{HoravaWitten95}\cite{HoravaWitten96}
to be given by the M-theoretic completion of 11-dimensional supergravity
KK-compactified on a $\mathbb{Z}_2$-orbifolded circle fiber,
where the $\mathbb{Z}_2$-action on the circle has two fixed points,
hence two fixed planes as an action on spacetime: the \emph{MO9-planes}.
\begin{itemize}
\vspace{-2mm}
\item With an actual formulation of M-theory lacking,
the argument for this is necessarily indirect, and it goes as follows.
Plain $11d$ supergravity turns out to have a
gravitational anomaly when considered on such MO9 boundaries,
hence to be inconsistent in itself. Thus, if the putative M-theory
completion indeed exists and hence is consistent, it must
\emph{somehow}
introduce a further contribution to the total anomaly such as to cancel
it. The form of that further anomaly contribution inferred this way
is the same as that of a would-be field theory of charged chiral fermions
on the MO9-planes, just as found in heterotic string theory.

\vspace{-2mm}
\item This suggests that if M-theory actually exists, it must
include avatars of these super gauge field theory degrees of freedom appearing on MO9-branes.
While many consistency checks for this assumption
have been found, it remained open what the
M-theoretic avatar of the heterotic gauge field actually is.
In \cite{HoravaWitten95}\cite{HoravaWitten96}
the 10d SYM action on the MO9s is just added by hand to that
of 11d supergravity.
\end{itemize}

\vspace{0mm}
\hypertarget{SecondOpenProb}{}
\noindent {\bf Third open problem: Super-exceptional M-geometry.}
The Kaluza-Klein (KK) compactifications of 11d supergravity on $n$-tori have a rich space
of scalar moduli fields invariant under ever larger exceptional Lie groups as $n$ increases \cite{CJ},
reflecting just the expected duality symmetries acting on the corresponding
string theories \cite{HT}.
This led to the proposal
\cite{Hull07}
(see also \cite{KNS00}\cite{West03}\cite{PW}\cite{West11}\cite{Bandos17})
that M-theory is an enhancement of $D=11$ supergravity
to a theory of ``exceptional geometry'' with
a ``generalized tangent bundle'' of the form
\begin{equation}
  \label{EGTangent}
  \underset{
    \mathclap{
    \mbox{
      \tiny
      \color{blue}
      exceptional tangent bundle
    }
    }
  }{
    \underbrace{
      T_{\mathrm{ex}} X^{n}
    }
  }
    \;\coloneqq\;
  T X^{n}
    \;\oplus\;
  \underset{
    \mathclap{
    \mbox{
      \tiny
      \color{blue}
      M2 wrapping modes
    }
    \;\;
    }
  }{
    \underbrace{
      \wedge_X^2 T^\ast X^{n}
    }
  }
    \;\oplus\;
  \underset{
    \mathclap{
    \;\;
    \mbox{
      \tiny
      \color{blue}
      M5 wrapping modes
    }
    }
  }{
    \underbrace{
      \wedge_X^5 T^\ast X^{n}
    }
  }
    \;\oplus\;
  \cdots
\end{equation}
locally encoding wrapping modes of the M2- and the M5-brane
already before KK-compactification.
\begin{itemize}
\vspace{-2mm}
\item While the exceptional generalized
geometry enhancements of the bosonic sector of 11d supergravity
is well studied (see, e.g., \cite{Fe}
and references therein),
the inclusion of fermionic exceptional coordinates,
hence a unification of supergeometry with exceptional generalized
geometry to
``super-exceptional generalized geometry'', had remained an
open problem \cite[p. 39]{Cederwall14}\cite[pp. 4, 7]{CederwallEdlundKarlsson13}.
Arguments  were given in \cite{Vaula07}\cite{Bandos17}\cite{FSS18}\cite{SS18}
that the super-exceptional geometry for maximal $n=11$
is to be identified with what was called the
``hidden supergroup of 11d supergravity'' in
\cite{DF}\cite{BAIPV04}\cite{ADR16},
but open questions remained.
In particular, the physical meaning of

\begin{enumerate}[{\bf (a)}]
\vspace{-4mm}
  \item the extra fermion field $\eta$ on super-exceptional spacetime
   (see Def. \ref{ExceptionalTangentSuperSpacetime} below),

\vspace{-1mm}
  \item the parameter $s \in \mathbb{R} \setminus \{0\}$
   for decompositions of the C-field
   (see Prop. \ref{TransgressionElementForM2Cocycle} below)
\end{enumerate}
\vspace{-4mm}
had remained open.

\vspace{-2mm}
\item It may seem that supersymmetrization is but an afterthought once the bosonic
sector of exceptional geometry is understood, (e.g. \cite{BSS18} for $n = 7$).
But most aspects of M-theory are controlled by -- and are emergent from -- its local
supersymmetry structure (see, e.g., \cite{Townsend97}\cite{FSS19a}), with the bosonic
sector being implied by the spin geometry, instead of the other way around. The lift of this
\emph{supersymmetry first principle} to exceptional generalized geometry had remained open.
\end{itemize}

\hypertarget{JointSolution}{}
\noindent {\bf The joint solution.}
In \cite[4.6]{FSS18}\cite{SS18} we had already observed that a supersymmetric enhancement of
$n=11$ exceptional M-geometry is provided by what \cite{DF} called the ``hidden supergroup''
of 11d supergravity. With \cite[Prop. 4.31]{FSS19b}\cite[Prop. 4.4]{FSS19c},
it follows that this must be the correct target space for M5-brane sigma-models, as
we explain in \cref{SuperExceptionalMGeometry}. Accordingly, in  \cref{SuperExceptionalM5Locus}
we consider super-exceptional 5-brane embeddings and find in \cref{SuperExceptionalLagrangian}
that this induces the Perry-Schwarz Lagrangian (reviewed in \cref{PSLagrangianForM5OnS1})
and, after super-exceptional equivariantization along the M-theory circle fiber
introduced in \cref{SuperExceptionalReductionOfMTheoryCircle}, the full
super-exceptional M5-brane Lagrangian, in \cref{SuperExceptionalM5Lagrangian}.
The resulting D4-brane Lagrangian with its 5d SYM+tYM Lagrangian is manifest
(Remark \ref{D4WZ})
and identifies the super-exceptional fermion as the M-theoretic avatar of the heterotic
gauge field (Remark \ref{Super2FormGaugeFieldStrengthAndGauginos}).

\medskip
\noindent Before giving more detail in the
\hyperlink{Outline}{Outline of results},
we recall the two foundational principles of our development:

\medskip

\hypertarget{FirstPrin}{}
\noindent {\bf Principle 1: Super-geometry.}
Despite the evident relevance of super-geometry for the foundations of M-theory,
many constructions in the literature start out with the bosonic data
(e.g. \cite[Sec. 2]{APPS97}) and relegate super-geometrization
to an afterthought (e.g. \cite[Sec. 3]{APPS97}). Countering this tendency,
the ``doubly supersymmetric'' approach of \cite{BPSTV95}\cite{HS97},
reviewed under the name ``super-embedding approach'' in \cite{Sorokin99}\cite{Sorokin01}, shows that
seemingly mysterious, or at least convoluted-looking, aspects of traditional constructions
find their natural meaning and more elegant formulation when strictly everything
is systematically internalized into super-geometry. In particular, the
all-important ``$\kappa$-symmetry'' of super $p$-brane sigma-models,
which, following \cite{GS84}, is traditionally imposed by hand onto the action
principle, is revealed by the superembedding approach to be
(\cite{STV89}, see \cite[Sec. 4.3]{Sorokin99}\cite[Sec. 4.3]{HoweSezgin05})
nothing but the super-odd-graded component of the super-worldvolume
super-diffeomorphism symmetry -- hence a consequence of
the fundamental principle of general covariance internal to super-geometry.

\vspace{-.5cm}

$$
  \hspace{-.4cm}
  \xymatrix@C=2.5pt@R=-.1em{
    \mbox{\footnotesize
            \begin{tabular}{c}
        $p$-brane
        \\
        sigma-models
      \end{tabular}
    }
    &
    &&
    \fbox{
      \small
      \begin{tabular}{c}
        \multirow{2}{*}{
          NSR-type
        }
        \\
        \\
        \cite{NeveuSchwarz71}\cite{Ramond71}
      \end{tabular}
    }
    &&
    \fbox{
      \small
      \begin{tabular}{c}
        \multirow{2}{*}{
          GS-type
        }
        \\
        \\
        \cite{GS84}\cite{BST87}
      \end{tabular}
    }
    &&
    \fbox{
      \small
      \begin{tabular}{c}
        super
        \\
        embedding
        \\
        \cite{BPSTV95}
      \end{tabular}
    }
    &&&
    \fbox{
      \small
      \begin{tabular}{c}
        super-exceptional
        \\
        embedding
        \\
        \cref{SuperExceptionalMGeometry} \cref{SuperExceptionalM5Locus}
      \end{tabular}
    }
    \\
    & &&
    \mbox{
      \tiny
      bosonic
    }
    &&
    \mbox{
      \tiny
      {\color{blue}super} geometric
    }
    &&
    \mbox{
      \tiny
      {\color{blue}super-}geometric
    }
    &&&
    \mbox{
      \tiny
      {\color{blue}super-exceptionally} geometric
    }
    \\
    \mbox{
      \tiny
      spacetime
    }
    &
    X
    &&
    X^{d,1}
    &&
    X^{d,1\vert \mathbf{N}}
    &&
    X^{d,1\vert \mathbf{N}}
    &&&
    X^{d,1\vert \mathbf{N}}_{\mathrm{ex}_s}
    \\
    \\
    \\
    \\
    \mbox{
      \tiny
      worldvolume
    }
    \ar[uuuu]^{
      \mbox{
        \tiny
        \begin{tabular}{c}
          sigma-model
          \\
          field
        \end{tabular}
      }
    }
    &
    \Sigma
    \ar[uuuu]
    &&
    \Sigma^{p,1\vert \mathbf{N}/k}
    \ar[uuuu]
    &&
    \Sigma^{p,1}
    \ar[uuuu]
    &&
    \Sigma^{p,1\vert \mathbf{N}/k}
    \ar[uuuu]
    &&&
    \Sigma^{p,1\vert \mathbf{N}/k}_{\mathrm{ex}_s}
    \ar[uuuu]
    \\
    &
    &&
    \mathclap{
      \mbox{
        \tiny
        \begin{tabular}{c}
          {\color{blue}super-}geometric
       \end{tabular}
      }
    }
    &&
    \mathclap{
      \mbox{
        \tiny
        \begin{tabular}{c}
          bosonic
       \end{tabular}
      }
    }
    &&
    \mathclap{
      \mbox{
        \tiny
        \begin{tabular}{c}
          {\color{blue} super} geometric
       \end{tabular}
      }
    }
    &&&
    \mathclap{
      \mbox{
        \tiny
        \begin{tabular}{c}
          {\color{blue} super-exceptionally} geometric
       \end{tabular}
      }
    }
  }
$$
Indeed, all of the following has been
systematically obtained from the superembedding approach:
The equations of motion of the superstring \cite[Sec. 4]{BPSTV95}
of the M2-brane \cite[Sec. 3]{BPSTV95}
and of the M5-brane \cite{HS97} \cite{HSW97}\cite[5.2]{Sorokin99},
as well as the Lagrangian density of
the superstring and of the M2-brane \cite{BSV95}\cite{HoweSezgin05}.
But an analogous derivation of the M5-brane's Lagrangian density
had remained open.
Notice that it
is the Lagrangian density which gives the
crucial instanton contributions for these branes
\cite{BeckerBeckerStrominger95}\cite{HarveyMoore99}.

\medskip

 \hypertarget{SecondPrin}{}
\noindent {\bf Principle 2: Homotopy theory.}
  The \emph{gauge principle} of physics
  --
  read as saying that
  no two things (e.g. field configuratons) are ever equal or not,
  but that we have to ask for gauge transformations
  between these, and higher order gauge-of-gauge transformations
  between those
  --
  is mathematically embodied in \emph{homotopy theory},
  these days increasingly referred to as ``higher structures''
  (see \cite[Sec. 2]{BSS18} for a lightning introduction and pointers
  to details, and
  see \cite{JSSW19} for a gentle invitation).
Combining this with super-geometry yields
\emph{super-homotopy theory} where \emph{super-geometric $\infty$-groupoids}
(super-$\infty$-stacks) unify super moduli spaces for higher super gauge fields with
super-orbifolds appearing as super-spacetimes.
\begin{center}
\begin{tabular}{lrll}
  & {\it Physics} & \phantom{AA} & {\it Mathematics}
  \\
  \hline
  & {Gauge principle} & & {Homotopy theory}
  \\
  \& &
  {Pauli exclusion principle}
  & &
  {Super-geometry}
  \\
  \hline
  \hline
  = & & &
    {\it Super-homotopy theory}
\end{tabular}
\end{center}
  Homotopy theory, and more so super-homotopy theory, is extremely rich.
  But if, for the time being, we ignore torsion cohomology groups,
  homotopy theory simplifies to \emph{rational homotopy theory}
  \cite{Quillen69}\cite{Sullivan77}
 (see \cite{Hess06}\cite{GrMo13},
 and see \cite{FSS19a}\cite[Sec. 2]{BSS18} for review in our context).
 The main result here is that topological spaces,
 regarded up to rational weak homotopy equivalence,
 are encoded by their differential graded-commutative algebra
 of Sullivan differential forms, regarded up to
 quasi-isomorphism.
  If we suppress some technical fine-print
 (see \cite[(8)]{BSS18} for the precise statement),
 we may schematically write this as follows:
 $$
   \xymatrix{
   \mathrm{Spaces}_{
     \Big/
     \!\!\!\!\!
     \mbox{
       \tiny
       \begin{tabular}{l}
         rational
         \\
         weak homotopy
         \\
         equivalence
       \end{tabular}
     }
   }
   \ar[rr]^-{ \mathrm{CE}(\mathfrak{l} - ) }_-{\simeq}
   &&
    \mathrm{dgcAlgebras}^{\mathrm{op}}_{
     \Big/\!\!\!\!\!
     \mbox{
       \tiny
       \begin{tabular}{l}
         quasi-
         \\
         isomorphism
       \end{tabular}
     }
   }
  }
 $$
  For super-homotopy theory this yields rational superspaces in
\emph{rational super-homotopy theory}
\cite[Sec. 2]{HSS18} (see \cite{FSS19a} for review)
The following table shows the notation which we use,
exemplified for key examples of rational super spaces:

\medskip
\medskip
{\small
\hspace{-5mm}
\def\arraystretch{1.2}
\begin{tabular}{|c||c|c|c|}
  \hline
  &
  \begin{tabular}{c}
    {\bf Rational }
    \\
    {\bf super space }
  \end{tabular}
  &
  \begin{tabular}{c}
    {\bf Loop }
    \\
    {\bf super $L_\infty$-algebra }
  \end{tabular}
  &
  \begin{tabular}{c}
    {\bf Chevalley-Eilenberg }
    \\
    {\bf super dgc-algebras}
    \\
    (``FDA''s)
  \end{tabular}
  \\
  \hline
  \hline
  \mbox{General}
    &
  $X$
    &
  $\mathfrak{l}X$
    &
  $\mathrm{CE}\big( \mathfrak{l}X \big)$
  \\
  \hline
  \begin{tabular}{c}
    Super
    \\
    spacetime
  \end{tabular}
    &
  $\mathbb{T}^{d,1\vert \mathbf{N}}$
    &
  $\mathbb{R}^{d,1\vert \mathbf{N}}$
    &
  $
    \mathbb{R}\big[ \{\psi^\alpha\}_{\alpha = 1}^N, \{e^a\}_{a = 0}^d \big]
    \Big/
    \left(
      \!\!\!
      \begin{array}{lcl}
        d\,\psi^\alpha & \!\!\!\!\!\!\ = \!\!\!\!\!\!\ & 0
        \\
        d\,e^a & \!\!\!\!\!\!\ = \!\!\!\!\!\!\ &
          \overline{\psi}\,\Gamma^a \psi
      \end{array}
      \!\!\!
    \right)
  $
  \\
  \hline
  \begin{tabular}{c}
    Eilenberg-MacLane
    \\
    space
  \end{tabular}
  &
  $
  \begin{aligned}
    & K(\mathbb{R},p+2)
    \\
    \simeq_{{}_{\mathbb{R}}} & B^{p+1} S^1
  \end{aligned}
  $
  &
  $\mathbb{R}[p+1]$
  &
  $
    \mathbb{R}\big[ c_{p+2} \big]
    \Big/
    \left(
      \!\!\!
      \begin{array}{lcl}
        d\,c_{p+2} & \!\!\!\!\!\!\ = \!\!\!\!\!\!\ & 0
      \end{array}
      \!\!\!
    \right)
  $
  \\
  \hline
  \begin{tabular}{c}
    Odd dimensional
    \\
    sphere
  \end{tabular}
  &
  $S^{2k+1}$
  &
  $\mathfrak{l}(S^{2k+1})$
  &
  $
  \mathbb{R}\big[\omega_{2k+1}\big]\Big/
  \left(
    \!\!\!
    \begin{array}{lcl}
      d\,\omega_{2k+1} & \!\!\!\!\!\! = \!\!\!\!\!\! & 0
    \end{array}
    \!\!\!
  \right)
  $
  \\
  \hline
  \begin{tabular}{c}
    Even dimensional
    \\
    sphere
  \end{tabular}
  &
  $S^{2k}$
  &
  $\mathfrak{l}(S^{2k})$
  &
  $
  \mathbb{R}\big[\omega_{2k}, \omega_{4k-1}\big]\Big/
  \left(
    \!\!\!
    \begin{array}{lcl}
      d\,\omega_{2k} & \!\!\!\!\!\! = \!\!\!\!\!\! & 0
      \\
      d\,\omega_{4k-1} & \!\!\!\!\!\! = \!\!\!\!\!\! &
       -\omega_{2k} \wedge \omega_{2k}
    \end{array}
    \!\!\!
  \right)
  $
  \\
  \hline
  \begin{tabular}{c}
    M2-extended
    \\
    super spacetime
  \end{tabular}
  &
  $\widehat{ \mathbb{T}^{10,1\vert \mathbf{32}} }$
  &
  $\mathfrak{m}2\mathfrak{brane}$
  &
  $
    \mathbb{R}
    \big[
      \{\psi^\alpha\}_{\alpha = 1}^{32},
      \{e^a\}_{a = 0}^{10},
      h_3
    \big]
    \Big/
    \left(
      \!\!\!
      \begin{array}{lcl}
        d\,\psi^\alpha & \!\!\!\!\!\!\!\!\! = \!\!\!\!\!\!\ & 0
        \\
        d\,e^a & \!\!\!\!\!\!\!\!\! = \!\!\!\!\!\!\ &
          \overline{\psi}\,\Gamma^a \psi
        \\
        d\,h_3 & \!\!\!\!\!\!\!\!\!= \!\!\!\!\!\!\ &
         \mu_{{}_{\rm M2}}
      \end{array}
      \!\!\!
    \right)
  $
  \\
  \hline
\end{tabular}
}

\medskip

One finds that a considerable amount of structures
expected in M-theory emerge naturally in rational super-homotopy theory:

\begin{itemize}
\vspace{2mm}
\item On super-geometric $\infty$-groupoids,  the Sullivan construction
of rational homotopy theory (see, e.g., \cite{Hess06} \cite{GrMo13}) unifies
with higher super Lie integration \cite[Sec. 3.1]{BM18}
to exhibit super $L_\infty$-algebroids as models for rational
super-homotopy theory. Their Chevalley-Eilenberg algebras are
the ``FDA''s as known in the supergravity literature
\cite{vanNieuwenhuizen82}\cite{DF}\cite{CDF91}.

\vspace{-2mm}
\item  Using super-homotopy theory, we had shown \cite{FSS13}\cite{FSS16a} that
the completion of the ``old brane scan'' to the full ``brane bouquet'' emerges from the
superpoint $\mathbb{R}^{0\vert 1}$ as the classification of iterated universal invariant
higher central extensions.

\vspace{-2mm}
\item This process culminates \cite[p. 12]{FSS19a} in the $D = 11$,
$\mathcal{N} =1$ (hence $N = \mathbf{32}$) super-Minkowski spacetime,
carrying
the super M2-brane cocycle $\mu_{{}_{\rm M2}}$ and,
on the corresponding higher extension,
the super M5-brane cocycle \cite{FSS13} which is the curvature of the M5 Wess-Zumino (WZ) term
\cite[(8)]{BLNPST97}\cite{FSS15} (see \cite{FSS19a} for review):

\vspace{-1.2cm}

\begin{equation}
  \label{TheMBraneCocycles}
  \hspace{-.7cm}
  \xymatrix@C=1.9em@R=13pt{
    &
    &&
    \\
    \mbox{
      \tiny
      \color{blue}
      \begin{tabular}{c}
        extended
        \\
        super-spacetime
      \end{tabular}
    }
    &
    \widehat{ \mathbb{T}^{10,1\vert \mathbf{32}} }
    \ar[rr]^-{
      \mathclap{\phantom{A \atop A}}
      \mbox{
        \tiny
        $
          2
          \big(
          \overset{
            \mathllap{
              \mbox{
                \tiny
                \color{blue}
                super M5-brane cocycle
              }
            }
            \;\;
            \eqqcolon
            \,
            {\color{blue}
              \mathbf{dL}^{\!\!\mathrm{WZ}}
            }
          }{
          \overbrace{
            \tfrac{1}{2}
            h_3
            \wedge
            \pi^\ast \mu_{{}_{\rm M2}}
            +
            \pi^\ast\mu_{{}_{\rm M5} }
          }
          }
          \big)
        $
      }
    }_<<<<<<{\ }="s"
    \ar[dd]_-{
      \mbox{
        \tiny
        $
          \begin{aligned}
            \pi & \simeq \mathrm{hofib}(\mu_{{}_{\rm M2}})
            \\ & \simeq
               (\mu_{{{}_{\rm M2}}}, 2 \mu_{{}_{\rm M5}})^\ast
               \big( (h_{\mathbb{H}})_{\mathbb{R}}\big)
          \end{aligned}
        $
      }
    }
    ^<<<<<<{\ }="t"
    &&
    S^7_{\mathbb{R}}
    \ar[dd]^-{ (h_{\mathbb{H}})_{{}_{\mathbb{R}}} }
    &
          &
    \left(
      \!\!\!\!\!\!
      \mbox{
        \tiny
        $
        \begin{array}{lcl}
          d\,\psi^\alpha & \hspace{-3mm}= &\hspace{-2mm}  0
          \\
          d\,e^a & \hspace{-3mm}= & \hspace{-2mm}\overline{\psi}\;\Gamma^a\psi
          \\
          d\,h_3 & \hspace{-3mm}= & \hspace{-2mm}
            \mu_{{}_{\rm M2}}
        \end{array}
        $
      }
      \!\!\!\!\!\!
    \right)
    \ar@{<-}[rrrr]^{
      \mbox{
        \tiny
        $
        \begin{array}{ccc}
          h_3 & \!\!\!\!\!\! \mapsfrom \!\!\!\!\!\! & h_3
          \\
          \mu_{{}_{\rm M2}} & \!\!\!\!\!\! \mapsfrom \!\!\!\!\!\! & \omega_4
          \\
          \mu_{{}_{\rm M5}}
          & \!\!\!\!\!\! \mapsfrom \!\!\!\!\!\! &
          \omega_7
        \end{array}
        $
      }
    }
    \ar@{<-}[dd]^-{
      \mbox{
        \tiny
        $
        \begin{array}{ccc}
          \psi^\alpha & e^a
          \\
          \mapsup & \mapsup
          \\
          \psi^\alpha & e^a
        \end{array}
        $
      }
    }
    \ar@{}[ddrrrr]|-{ \mbox{ \tiny (po) } }
    &&&&
    \left(
      \!\!\!\!\!\!
      \mbox{
        \tiny
        $
        \begin{array}{lcl}
          d\,h_3 & \hspace{-3mm}= & \hspace{-2mm}\omega_4
          \\
          d\,\omega_4 & \hspace{-3mm}= & \hspace{-2mm} 0
          \\
          d\,\big(h_3 \wedge \omega_4 + 2\omega_7\big)
          & \hspace{-3mm}= & \hspace{-2mm} 0
        \end{array}
        $
      }
      \!\!\!\!\!\!
    \right)
    \ar@{<-^{)}}[dd]^-{
      \mbox{
        \tiny
        $
        \begin{array}{ccc}
          \omega_4 & \omega_7
          \\
          \mapsup & \mapsup
          \\
          \omega_4 & \omega_7
        \end{array}
        $
      }
    }
    \\
    \\
    \mbox{
      \tiny
      \color{blue}
      super-spacetime
    }
    &
    \mathbb{T}^{10,1\vert \mathbf{32}}
    \ar[rr]^-{
      \overset{
        \mbox{
          \tiny
          \color{blue}
          $
          \begin{tabular}{c}
            joint M2/M5-brane cocycle
            \\
            in rational Cohomotopy
          \end{tabular}
          $
        }
      }{
        \mbox{\tiny
          $
          \mu_{{}_{{\rm M2}/{\rm M5}}}
          \coloneqq
          (\mu_{{}_{\rm M2}}, 2 \mu_{{}_{\rm M5}})
          $
        }
      }
    }
    \ar[dr]_-{
      \mathllap{
        \mbox{
        \tiny
        \color{blue}
        \begin{tabular}{c}
          M2-brane
          \\
          cocycle
        \end{tabular}
      }
      \;\;\;
      }
      \mbox{\tiny
        $\mu_{ {}_{\rm M2}}$
      }
    }
    &&
    S^4_{\mathbb{R}}
    \ar[dl]^-{  }
    &
    &
    \left(
      \!\!\!\!\!\!
      \mbox{
        \tiny
        $
        \begin{array}{lcl}
          d\,\psi^\alpha & \hspace{-3mm}= & \hspace{-2mm} 0
          \\
          d\,e^a & \hspace{-3mm}= & \hspace{-2mm}\overline{\psi}\;\Gamma^a\psi
        \end{array}
        $
      }
      \!\!\!\!\!\!
    \right)
    \ar@{<-}[rrrr]_-{
      \mbox{
        \tiny
        $
        \begin{array}{lcc}
          \overset{
            \eqqcolon \;
            { \color{blue} \mu_{{}_{\rm M2}} }
          }{
          \overbrace{
            \big(\tfrac{i}{2}\overline{\psi}\;{{\Gamma}}_{a_1 a_2}\psi\big)
            \wedge e^{a_1} \wedge e^{a_2}
          }
          }
          & \!\!\!\!\!\! \mapsfrom \!\!\!\!\!\! \!\!\!\!\!   & \omega_4
          \\
          \\
          \underset{
            \eqqcolon \;
            { \color{blue} \mu_{{}_{\rm M5}} }
          }{
          \mathllap{
            2 \,
          }
          \underbrace{
          \tfrac{1}{5!}
          (\overline{\psi}\;{{\Gamma}}_{a_1 \cdots a_5} \psi)
          \wedge e^{a_1} \wedge \cdots \wedge e^{a_5}
          }
          }
          & \!\!\!\!\!\! \mapsfrom \!\!\!\!\!\!\!\!\!\!\!    &
          \omega_7
        \end{array}
        $
      }
    }
    &&&&
    \left(
      \!\!\!\!\!\!
      \mbox{
        \tiny
        $
        \begin{array}{lcl}
          d\,\omega_4 & \hspace{-3mm}=& \hspace{-2mm} 0
          \\
          d\,\omega_7 & \hspace{-3mm}= & \hspace{-2mm} - \omega_4 \wedge \omega_4
        \end{array}
        $
      }
      \!\!\!\!\!\!
    \right)
    \\
    &
    &
    B^4 \mathbb{R}
    \ar@{=>}|{h_3}^{ \mbox{ \tiny (pb) } } "s"; "t"
  }
\end{equation}

\vspace{-2mm}
\item
  The M2- and M5-brane cocycles unify
  \cite[2.5]{Sati13}\cite{FSS15} into a single
  non-abelian cocycle $\mu_{{}_{{\rm M2}/{\rm M5}}}$ \eqref{TheMBraneCocycles}
  with coefficients in the rational 4-sphere,
  hence in rational \emph{Cohomotopy} cohomology theory in degree 4.

\vspace{-2mm}
\item
  The Bredon-equivariant enhancement of the joint $\mu_{{}_{{\rm M2}/{\rm M5}}}$ cocycle
  to rational ADE-equivariant Cohomotopy,
  amounts \cite{HSS18}, via Elmendorf's theorem,
  to relative trivializations
  along super-embeddings of fixed/singular super-spacetimes:
\begin{equation}
  \label{RelativeTrivialization}
  \hspace{-1cm}
  \raisebox{44pt}{
  \xymatrix@R=12pt@C=30pt{
    &
    \ar@{}[rrrr]|-{
      \mbox{
        \tiny
        \color{blue}
        \begin{tabular}{c}
          unified M2/M5-brane cocycle
          \\
          in rational Cohomotopy
        \end{tabular}
      }
    }
    &&&&
    \\
    \mbox{
      \tiny
      \color{blue}
      \begin{tabular}{c}
        super
        \\
        spacetime
      \end{tabular}
    }
    \ar@(ul,ur)@[white]^{
      \mbox{
         \tiny
         \color{blue}
         \begin{tabular}{c}
           super
           \\
           ADE/HW-action
         \end{tabular}
      }
    }
    &
    \mathbb{T}^{10,1\vert \mathbf{32}}
    \ar[rrrr]^-{
      \mu_{{}_{\rm M2/M5}}
    }_<<<{\ }="s"
    \ar@(ul,ur)^{ G_{\mathrm{ADE/HW}} }
    &&&&
    S^4_{\mathbb{R}}
    \ar@(ul,ur)^{ G_{\mathrm{ADE/HW}} }
    \\
    \\
    \mbox{
      \tiny
      \color{blue}
      \begin{tabular}{c}
        super BPS
        \\
        M-brane spacetime
      \end{tabular}
    }
    \ar@{}[uu]|-{
      \mbox{
        \tiny
        \color{blue}
        \begin{tabular}{c}
          super
          \\
          embedding
        \end{tabular}
      }
    }
    &
    \mathbb{T}^{
      d,1
      \vert
      \mathbf{N}
    }
    \ar[rrrr]^>>>{\ }="t"
    \ar@{^{(}->}[uu]^-{ i }
    &&&&
    S^{d \lt 4}_{\mathbb{R}}
    \ar[uu]
    \ar@{=>}|{
      \mbox{
        \tiny
        \color{blue}
        \begin{tabular}{c}
          relative trivialization
          \\
          along super embedding
        \end{tabular}
      }
    \;}
    "s"; "t"
  }
  }
\end{equation}
\end{itemize}

These trivializations relative to $\sfrac{1}{2}$-BPS super-embeddings
constitute the corresponding super $p$-brane Green-Schwarz-type
sigma model Lagrangian, at least for branes without gauge fields
on their worldvolume \cite[Prop. 6.10]{HSS18}.
We recall how this works in the case of the M2-brane:

\medskip
\hypertarget{M2Recalled}{}
\noindent {\bf The M2-brane in super-homotopy theory.}
\cite[Prop. 6.10]{HSS18}
The $\kappa$-symmetric Green-Schwarz-type
Lagrangian density for the M2-brane \cite{BST87}
looks intricate when written out in the traditional component formulation
(see \cite[(2.1)]{deWitHoppeNicolai88}\cite[(3)]{DasguptaNicolaiPlefka03}),
but attains a highly elegant form in a fully supergeometric formulation.
Indeed, promoting the M2 worldvolume itself to a super-manifold
embedded into target super-spacetime,
locally of the form shown on the left of the following diagram
\begin{equation}
  \label{M2LagrangianAsAHomotopy}
  \hspace{-1cm}
  \raisebox{44pt}{
  \xymatrix@R=12pt@C=20pt{
    &
    \ar@{}[rrrr]|-{
      \mbox{
        \tiny
        \color{blue}
        \begin{tabular}{c}
          super
          \\
          M2-brane cocycle
        \end{tabular}
      }
    }
    &&&&
    \\
    \mbox{
      \tiny
      \color{blue}
      \begin{tabular}{c}
        super
        \\
        spacetime
      \end{tabular}
    }
    \ar@(ul,ur)@[white]^{
      \mbox{
         \tiny
         \color{blue}
         \begin{tabular}{c}
           super
           \\
           ADE-action
         \end{tabular}
      }
    }
    &
    \mathbb{T}^{10,1\vert \mathbf{32}}
    \ar[rrrr]^-{
      \mathbf{dL}^{\!\!\mathrm{WZ}}
      \;\coloneqq\;
      \mu_{{}_{\rm M2}}
    }_<<<{\ }="s"
    \ar@(ul,ur)^{ G_{\mathrm{ADE}} }
    &&&&
    B^4 \mathbb{R}
    \\
    \\
    \mbox{
      \tiny
      \color{blue}
      \begin{tabular}{c}
        super
        \\
        M2 spacetime
      \end{tabular}
    }
    \ar@{}[uu]|-{
      \mbox{
        \tiny
        \color{blue}
        \begin{tabular}{c}
          super
          \\
          embedding
        \end{tabular}
      }
    }
    &
    \mathbb{T}^{
      2,1
      \vert
      8 \cdot \mathbf{2}
    }
    \ar[rrrr]^>>>{\ }="t"
    \ar@{^{(}->}[uu]^-{ i }
    &&&&
    \ast
    \ar[uu]
    \ar@{=>}|{
      \underset{
        \mbox{
          \tiny
          \color{blue}
          Lagrangian density
        }
      }{
        \mathbf{L}^{\!\!\mathrm{NG}}
          \coloneqq
        e^0 \wedge e^1 \wedge e^2
      }
    }
    "s"; "t"
  }
  }
  \phantom{AAA}
  \Longleftrightarrow
  \phantom{AAA}
  \underset{
    \mathclap{
      \mbox{
        \tiny
        \color{blue}
        \begin{tabular}{c}
          super
          \\
          embedding
        \end{tabular}
      }
    }
  }{
    \underbrace{
      i^\ast
    }
  }
  \!\!
  \overset{
    \mbox{
      \tiny
      \color{blue}
      \begin{tabular}{c}
        super
        \\
        M2 cocycle
      \end{tabular}
    }
  }{
    \overbrace{
      \mathclap{\phantom{A \atop A}}
      \mathbf{dL}^{\!\!\mathrm{WZ}}
    }
  }
  \;\;=\;\;
  d\,\;
  \overset{
    \mathclap{
    \mbox{
      \tiny
      \color{blue}
      \begin{tabular}{c}
        super
        Nambu-Goto Lagrangian
        \\
        = Green-Schwarz Lagrangian
      \end{tabular}
    }
    }
  }{
    \overbrace{
      \mathclap{\phantom{A \atop A}}
      \mathbf{L}^{\!\!\mathrm{NG}}
    }
  }
\end{equation}
the M2-brane's Lagrangian density $\mathbf{L}^{\!\!\mathrm{NG}}$
arises simply as the super-homotopy theoretic trivialization
of the M2-brane cocycle restricted along the super-embedding.
Concretely, this identifies the Lagrangian
with the \emph{super-volume form}
\begin{equation}
  \label{SuperVolumeFormIn3d}
  \mathbf{L}^{\!\!\mathrm{NG}}
  \;\coloneqq\;
  \mathrm{svol}_{2+1}
  \;\coloneqq\;
  e^0 \wedge e^1 \wedge e^2
  \;=\;
  \tfrac{1}{3!} \epsilon_{a_0 a_1 a_1}
  e^{a_0} \wedge e^{a_1} \wedge e^{a_2}
  \;\;
  \in
  \Omega^\bullet_{\mathrm{li}}\big( \mathbb{R}^{2,1\vert 8 \cdot \mathbf{2}} \big)
  \;\simeq\;
  \mathrm{CE}\big( \mathbb{R}^{2,1\vert 8 \cdot \mathbf{2}} \big)
\end{equation}
on the super M2 worldvolume.
Here
$
  e^a
    \;\coloneqq\;
  d x^a + \overline{\theta} \Gamma^a d \theta
$
denotes the vielbein 1-forms
which are \emph{left-invariant} ``li" with respect to the
translational supersymmetry action of $\mathbb{R}^{2,1\vert 8 \cdot \mathbf{2}}$ on itself
(see Remark \ref{SuperExceptionalSpacetimeAsManifold} below).
Consequently,  the super-volume form \eqref{SuperVolumeFormIn3d}
has as bosonic component the ordinary volume form
$\mathrm{vol}_{2+1} \in \Omega^3\big( \mathbb{R}^{2,1} \big)$,
to which the fermionic components are added, ensuring
overall supersymmetry-invariance of the super-volume form
$$
  \begin{aligned}
    \mathbf{L}^{\!\!\mathrm{NG}}
    \;\coloneqq\;
    \overset{
        \overset{
          \mbox{
            \tiny
            \color{blue}
            super-volume form
          }
        }{
          \mathrm{svol}_{2+1}
          \coloneqq
        }
    }{
      \overbrace{
        e^0 \wedge e^1 \wedge e^2
      }
    }
    & \;=\;
    \overset{
      \overset{
        \mbox{
          \tiny
          \color{blue}
          ordinary volume form
        }
      }{
        = \mathrm{vol}_{2+1}
      }
    }{
      \overbrace{
        d x^0 \wedge dx^1 \wedge dx^2
      }
    }
    +
    \overset{
      \mbox{
        \tiny
        \color{blue}
        \raisebox{16pt}{
          fermionic corrections
        }
      }
    }{
      \mathcal{O}(\theta \Gamma d\theta).
    }
  \end{aligned}
$$
These fermionic correction terms, systematically
obtained here simply by expanding out the super-volume form
in components, constitute the otherwise intricate-looking
components of the Green-Schwarz-type Lagrangian
for the M2-brane, which is thereby revealed simply as the
super-Nambu-Goto Lagrangian.

\medskip
What had been left open in \cite{HSS18} is the analogous result
for brane species with gauge fields on their worldvolume, notably
the case of the M5-brane, which is a much richer situation (see \cite{FSS13}).
This will be one of the main topics that we address in this paper.

 \medskip

\hypertarget{Outline}{}
\noindent {\bf Outline of results.} We establish the following:
\begin{enumerate}[{\bf (i)}]

\vspace{-2mm}
\item In \cref{PSLagrangianForM5OnS1} we
generalize the bosonic Perry-Schwarz Lagrangian
$\mathbf{L}^{\!\!\mathrm{PS}} = F \wedge \widetilde F$
to a coordinate-invariant
expression applicable to possibly non-trivial worldvolume circle bundles.

\vspace{-2mm}
\item
 In \cref{SuperExceptionalMGeometry} we
 recall super-exceptional M-geometry
 with the super-exceptional M5-brane cocycle
 and
 introduce super-exceptional embedding of M-brane spacetimes.

\vspace{-2mm}
\item
  In \cref{SuperExceptionalM5Locus} we introduce specifically the
  super-exceptional embedding of the
  $\tfrac{1}{2}\mathrm{M5} = \mathrm{MK6} \cap \mathrm{MO9}$
  brane configuration
  and find the super-exceptional lift of the
  isometry along the Ho{\v r}ava-Witten-circle $S^1_{\mathrm{HW}}$:

  \vspace{-2mm}

  \begin{center}
  \fbox{
    \xymatrix@R=2pt@C=0pt{
      &
      &
      \ar@{}[rr]_-{
        \mbox{
          {\bf super-exceptional lift of...}
        }
      }
      &
      &
      &
      {\phantom{AA}}
      \\
      {\phantom{AAA}}
      &
      \mathclap{
      \mbox{\color{blue} \footnotesize
        \begin{tabular}{c}
          isometry
          \\
          along $S^1_{\mathrm{HW}}$
        \end{tabular}
      }
      \;\;\;\;\;\;\;\;\;\;\;\;\;
      }
      &
      \mathclap{
      \mbox{\color{blue} \footnotesize
        \begin{tabular}{c}
          on $\tfrac{1}{2}\mathrm{M5}$
        \end{tabular}
      }
      }
      \ar@{}[rr]|-{
        \mbox{\color{blue} \footnotesize
          \begin{tabular}{c}
            embedded in
          \end{tabular}
        }
      }
      &
      {\phantom{AAAAAAAAAAAAAAAAA}}
      &
      \mathclap{
      \mbox{\color{blue} \footnotesize
        \begin{tabular}{c}
          M-theory spacetime
        \end{tabular}
      }
      }
      &
      \\
      &
      \ar@(dl,ul)^{ v_5^{\mathrm{ex}_s} }
      &
      \big(
        \mathbb{R}^{5,1\vert \mathbf{8}}
        \times
        \mathbb{R}^1
      \big)_{\mathrm{ex}_s}
      \ar@{^{(}->}[rr]^-{ i_{\mathrm{ex}_s} }
      &&
      \mathbb{R}^{10,1\vert \mathbf{32}}_{\mathrm{ex}_s}
      &&
      &
      \\
      &
      \mathclap{
        \mbox{
          Prop. \ref{LiftOfVectorField}
        }
      \;\;\;\;\;\;\;\;\;\;\;\;\;
      }
      &
      \mbox{
        Def. \ref{HalfM5LocusAndItsExceptionalTangentBundle}
      }
      &&
      \mbox{
        Def. \ref{ExceptionalTangentSuperSpacetime}
      }
      &
    }
  }
  \end{center}

\item
  In \cref{SuperExceptionalLagrangian}
  we find a natural super-exceptional lift
  of the bosonic gauge field strength with KK-modes,
  the  bosonic Perry-Schwarz Lagrangian as well as of the
  topological Yang-Mills Lagrangian:

\vspace{2mm}
 \hspace{-6mm}
 {\small
  \begin{tabular}{|c|c|c|c|c|}
    \hline
    \multicolumn{5}{|c|}{
      {\bf super-exceptional lift of...}
    }
    \\
    \hline
    \mbox{\color{blue}
    \begin{tabular}{c}
      M5-worldvolume
      \\
      higher gauge flux
    \end{tabular}
    }
    &
    \mbox{\color{blue}
    \begin{tabular}{c}
      M5-brane cocycle
      \\
      = WZW curvature:
    \end{tabular}
    }
    &
    \mbox{\color{blue}
    \begin{tabular}{c}
      (dual) gauge flux
      \\
      with KK-modes:
    \end{tabular}
    }
    &
    \mbox{\color{blue}
    \begin{tabular}{c}
      bosonic
      \\
      Perry-Schwarz
      \\
      Lagrangian:
          \end{tabular}
    }
    &
    \mbox{\color{blue}
    \begin{tabular}{c}
      topological
      \\
      Yang-Mills
      \\
      Lagrangian:
    \end{tabular}
    }
    \\
    &&&&
    \\
    $
      \underset{
        \mathclap{ \phantom{{A \atop A} \atop A} }
        \mbox{
          \small
          s.t.
          $
          d H_{\mathrm{ex}_s}
          =
          (\pi_{\mathrm{ex}_s})^\ast \mu_{{}_{\rm M2}}
          $
        }
      }{
        H_{\mathrm{ex}_s}
      }
    $
    &
    $
      \underset{
        \eqqcolon
        \mathbf{dL}^{\!\!\mathrm{WZ}}_{\mathrm{ex}_s}
      }{
      \underbrace{
      (\pi_{\mathrm{ex}_s})^\ast \mu_{{}_{\rm M5}}
      +
      \tfrac{1}{2}H_{\mathrm{ex}_s} \wedge d H_{\mathrm{ex}_s}
      }
      }
    $
    &
    $
      F_{\mathrm{ex}_s}
      \;\;\;\;
      \big(\,\widetilde F_{\mathrm{ex}_s}\,\big)
    $
    &
    $
      \underset{
        \eqqcolon
        \mathbf{L}^{\!\!\mathrm{PS}}_{\mathrm{ex}_s}
      }{
      \underbrace{
        -\tfrac{1}{2}
        F_{\mathrm{ex}_s}
          \wedge
        \widetilde F_{\mathrm{ex}_s}
      }
      }
    $
    &
    $
      \underset{
        \eqqcolon
        \mathbf{L}^{\!\!\mathrm{tYM}}_{\mathrm{ex}_s}
      }{
      \underbrace{
        -\tfrac{1}{2}
        F_{\mathrm{ex}_s}
          \wedge
        F_{\mathrm{ex}_s}
      }
      }
    $
    \\
    Prop. \ref{TransgressionElementForM2Cocycle}
    &
    Def. \ref{ExceptionalM5SuperCocycle}
    &
    Def. \ref{SuperExceptionalPSLagrangian} (i)
    &
    Def. \ref{SuperExceptionalPSLagrangian} (ii)
    &
    Def. \ref{SuperExceptionalTopologicalYM}
    \\
    \hline
  \end{tabular}
 }

  \vspace{2mm}

  In the course of this identification we find that

   {\bf i)} the putative parameter $s$ of the super-exceptional
    geometry is fixed to $s = -3$
    (Prop. \ref{ExceptionalPreimageOfPerrySchwarzLagrangian}
    \hyperlink{SuperExceptionalPSLagrangianForsMinusThree}{ii)})

    {\bf ii)}
    the super-exception fermion $\eta$ is
    the M-theory avatar of the heterotic gaugino field
    (Prop. \ref{calFexcDecomposed},
    Rem. \ref{Super2FormGaugeFieldStrengthAndGauginos}).

  Then we show (Prop. \ref{TrivializationOf7CocycleOnExceptionalHalfM5})
  that the super-exceptional Perry-Schwarz Lagrangian
 arises via a super-exceptional analog of the super-embedding
 mechanism as a trivialization
 of the $S^1_{\mathrm{HW}}$-compactified
 super-exceptional M5-brane cocycle after restriction along the
 super-exceptional embedding of the $\tfrac{1}{2}\mathrm{M5}$.
 $$
   \underset{
     \mathclap{
     \mbox{
       \tiny
       \color{blue}
       \begin{tabular}{c}
         super-exceptional
         \\
         $S^1_{\mathrm{HW}}$-compactification
       \end{tabular}
     }
     \;\;\;\;\;\;\;\;\;\;\;\;\;
     }
   }{
   \underbrace{
     \mathclap{\phantom{A \atop A}}
     \iota_{v_5}^{\mathrm{ex}_s}
   }
   }
   \;\;
   \underset{
     \mathclap{
     \;\;\;\;\;\;\;\;\;\;\;\;\;
     \mbox{
       \tiny
       \color{blue}
       \begin{tabular}{c}
         super-exceptional
         \\
         embedding
       \end{tabular}
     }
     }
   }{
     \underbrace{
       \mathclap{\phantom{A \atop A}}
       (i_{\mathrm{ex}_s})^\ast
     }
   }
   \;
   \overset{
     \mathclap{
       \mbox{
         \tiny
         \color{blue}
         \begin{tabular}{c}
           super-exceptional
           \\
           M5-brane cocycle
         \end{tabular}
       }
       \;\;\;\;
     }
   }{
     \overbrace{
       \mathclap{ \phantom{A \atop A} }
       \mathbf{dL}^{\!\!\mathrm{WZ}}_{\mathrm{ex}_s}
     }
   }
   \;\;=\;\;
   d\,\;
   \overset{
     \mathclap{
     \;\;\;\;
     \mbox{
       \tiny
       \color{blue}
       \begin{tabular}{c}
         super-exceptional
         \\
         Perry-Schwarz Lagrangian
       \end{tabular}
     }
     }
   }{
     \overbrace{
       \mathclap{\phantom{A \atop A}}
       \mathbf{L}^{\!\!\mathrm{PS}}_{\mathrm{ex}_s}
     }
   }.
 $$
 This is a partial analog for the M5-brane
 of the super-embedding construction of the M2-brane
 \eqref{M2LagrangianAsAHomotopy}. To get the full
 statement we need not just compactify but
 \emph{equivariantize}
 along $S^1_{\mathrm{HW}}$:

\vspace{-2mm}
\item
  In \cref{SuperExceptionalReductionOfMTheoryCircle}
  we show that an equivariant enhancement of the super-exceptional M5-cocycle
  with respect to super-exceptional $\Omega S^2_{\mathrm{HW}}$-action
  exists, where $\Omega S^2$ is the based loop space of the two-sphere.
  Furthermore, this unifies it with the super-exceptional Perry-Schwarz
  and the super-exceptional topological Yang-Mills Lagrangian
  (Theorem \ref{SuperExceptionalPSEquivariantEnhancement}).

  To put this in perspective,
  we also explain
  (by Prop. \ref{RationalFibrationOf6dSuperspacetimeOver5d})
  how $\Omega S^2_{\mathrm{HW}}
  \to S^1_{\mathrm{HW}}$ refines the naive circle
  action by taking the super-cocycle for the little-string
  in 6d into account.
This is a 6d analog
to capturing the form fields in 11d M-theory
via the Cohomotopical 4-sphere coefficient
\cite{Sati13}\cite{FSS15},
leading to a description of type IIA in ten dimensions using
a refined variant of the loop space of  $S^4$,
namely the cyclic loop space
\cite{FSS16a}\cite{FSS16b} (see \cite{FSS19a} for overview).

\vspace{-2mm}
\item
  In \cref{SuperExceptionalM5Lagrangian} we
  put all the pieces together and establish
  (Cor \ref{M5LagrangianIsRelativeTrivialization}) the
  full super-exceptional embedding construction of the M5-brane
  Lagrangian as a sum of the super-exceptional Nambu-Goto Lagrangian
  and the super-exceptional Perry-Schwarz Lagrangian,
  the analogue of the M2 brane construction \eqref{M2LagrangianAsAHomotopy}:

\vspace{-.6cm}

\begin{equation}
  \label{M5LagrangianAsAHomotopy}
  \hspace{-1cm}
  \raisebox{44pt}{
  \xymatrix@R=12pt@C=23pt{
    &
    \ar@{}[rrrr]|-{
      \mbox{
        \tiny
        \color{blue}
        \begin{tabular}{c}
          super-exceptional
          \\
          M5-brane cocycle
        \end{tabular}
      }
    }
    &&&&
    \\
    \mbox{
      \tiny
      \color{blue}
      \begin{tabular}{c}
        super-exceptional
        \\
        spacetime
      \end{tabular}
      \hspace{-1cm}
    }
    \ar@(ul,ur)@[white]^{
      \mbox{
         \tiny
         \color{blue}
         \begin{tabular}{c}
           super-exceptional
           \\
           ADE-action
         \end{tabular}
         \hspace{-1cm}
      }
    }
    &
    \big(
      \mathbb{T}^{9,1\vert \mathbf{16}}
      \times
      \mathbb{R}^1
    \big)_{\mathrm{ex}_s}
    \ar[rrrr]^-{
      \mathbf{dL}^{\!\!\mathrm{WZ}}_{\mathrm{ex}_s}
      \;\coloneqq\;
      (\pi_{\mathrm{ex}_s})^\ast
      \mu_{{}_{\rm M5}}
      +
      \frac{1}{2}
      H_{\mathrm{ex}_s} \wedge d H_{\mathrm{ex}_s}
    }_<<<{\ }="s"
    \ar@(ul,ur)^{ G_{\mathrm{ADE}} }
    &&&&
    S^7_{\mathbb{R}}
    \\
    \\
    \mbox{
      \tiny
      \color{blue}
      \begin{tabular}{c}
        super-exceptional
        \\
        $\tfrac{1}{2}\mathrm{M5}$ spacetime
      \end{tabular}
      \hspace{-1cm}
    }
    \ar@{}[uu]|-{
      \mbox{
        \tiny
        \color{blue}
        \begin{tabular}{c}
          super-exceptional
          \\
          embedding
        \end{tabular}
        \hspace{-1cm}
      }
    }
    &
    \big(
      \mathbb{T}^{
        5,1
        \vert
        \mathbf{8}
      }
      \times
      \mathbb{T}^1
    \big)_{\mathrm{ex}_s}
    \ar[rrrr]^>>>{\ }="t"
    \ar@{^{(}->}[uu]^-{
      i_{\mathrm{ex}_s}
    }
    &&&&
    \ast
    \ar[uu]
    \ar@{=>}|{
      \underset{
        \mbox{
          \tiny
          \color{blue}
          \begin{tabular}{c}
            relative trivialization
            \\
            along super-exceptional embedding
          \end{tabular}
        }
      }{
        \mathbf{L}^{\!\!\mathrm{NG}}_{\mathrm{ex}_s}
        +
        \mathbf{L}^{\!\!\mathrm{PS}}_{\mathrm{ex}_s}
        \wedge e^5
      }
    }
    "s"; "t"
  }
  }
  \phantom{A}
  \Leftrightarrow
  \phantom{A}
  \underset{
    \mathclap{
      \mbox{
        \tiny
        \color{blue}
        \begin{tabular}{c}
          super-exceptional
          \\
          embedding
        \end{tabular}
      }
    }
  }{
    \underbrace{
      (i_{\mathrm{ex}_s})^\ast
    }
  }
  \;
  \overset{
    \mathclap{
    \mbox{
      \tiny
      \color{blue}
      \begin{tabular}{c}
        super-exceptional
        \\
        M5 cocycle
      \end{tabular}
    }
    }
  }{
    \overbrace{
      \mathclap{\phantom{A \atop A}}
      \mathbf{dL}^{\!\!\mathrm{WZ}}_{\mathrm{ex}_s}
    }
  }
  \;\;\;=\;\;\;
  d\,\;
  \Big(
  \overset{
    \mathclap{
    \mbox{
      \tiny
      \color{blue}
      \begin{tabular}{c}
        super-exceptional
        \\
        Nambu-Goto Lagrangian
      \end{tabular}
    }
    \;\;\;\;\;\;\;\;\;
    }
  }{
    \overbrace{
      \mathclap{\phantom{A \atop A}}
      \mathbf{L}^{\!\!\mathrm{NG}}_{\mathrm{ex}_s}
    }
  }
  \;\;\;\;
  +
  \;\;\;\;
  \overset{
    \mathclap{
    \;\;\;\;\;\;\;
    \mbox{
      \tiny
      \color{blue}
      \begin{tabular}{c}
        super-exceptional
        \\
        Perry-Schwarz Lagrangian
      \end{tabular}
    }
    }
  }{
    \overbrace{
      \mathclap{\phantom{A \atop A}}
      \mathbf{L}^{\!\!\mathrm{PS}}_{\mathrm{ex}_s}
    }
    \wedge
    e^5
  }
  \Big)
\end{equation}

 Moreover, we show (Theorem \ref{TheTrivialization})
 that
 the $\Omega S^2_{\mathrm{HW}}$-equivariant enhancement of
 the super-exceptional M5-brane cocycle,
 hence the homotopy-theoretic KK-compactification on the
 $S^1_{\mathrm{HW}}$-fiber, makes this super-embedding
 construction pick up the manifest WZ-term of the D4-brane
 (Remark \ref{D4WZ}):

\vspace{3mm}
\fbox{
$\;
\begin{array}{c}
\small
      \Big(
        \underset{
          \mathclap{
          \mbox{
            \tiny
            \color{blue}
            \begin{tabular}{c}
              super-exceptional
              \\
              embedding
            \end{tabular}
          }
          \;\;\;\;\;\;
          }
        }{
          \underbrace{
            \mathclap{\phantom{A \atop A}}
            (i_{\mathrm{ex}_s})^\ast
          }
        }
        \overset{
          \mathclap{
          \mbox{
            \tiny
            \color{blue}
            \begin{tabular}{c}
              super-exceptional
              \\
              M5 cocycle
            \end{tabular}
          }
          }
        }{
          \overbrace{
            \mathclap{\phantom{A \atop A}}
            \mathbf{dL}^{\!\!\mathrm{WZ}}_{\mathrm{ex}_s}
          }
        }
      \Big)
      \!\!
      \underset{
        \;\;\;\;\;\;
        \mathclap{
        \mbox{
          \tiny
          \color{blue}
          \begin{tabular}{c}
            super-exceptional
            \\
            $\Omega S^2_{\mathrm{HW}}$-equivariant
            \\
            enhancement
          \end{tabular}
        }
        }
      }{
        \underbrace{
          {}_{\sslash \Omega S^2_{\mathrm{HW}}}
        }
      }
  \;=\;
  d
  \Big(
  \overset{
    \mathclap{
    \mbox{
       \tiny
       \color{blue}
       \begin{tabular}{c}
         super-exceptional
         \\
         Nambu-Goto Lagrangian
       \end{tabular}
    }
    }
  }{
  \overbrace{
    \mathclap{\phantom{A \atop A}}
    \mathrm{vol}^{5+1}_{\mathrm{ex}_s}
  }
  }
  \;+
  \!
  \underset{
    \mathclap{
    \mbox{
      \tiny
      \color{blue}
      \begin{tabular}{c}
        MC-form along $S^1_{\mathrm{HW}}$
      \end{tabular}
    }
    }
  }{
    \underbrace{
      \mathclap{\phantom{A \atop A}}
      e^5
    }
  }
  \!\!
  \wedge
  \overset{
    \mathclap{
    \mbox{
      \tiny
      \color{blue}
      \begin{tabular}{c}
        super-exceptional
        \\
        Perry-Schwarz Lagrangian
      \end{tabular}
    }
    }
  }{
  \overbrace{
  \mathclap{\phantom{A \atop a}}
    \tfrac{1}{2}
    F_{\mathrm{ex}_s}
    \wedge
    \widetilde F_{\mathrm{ex}_s}
  }
  }
  \Big)
  \;+\;
  e^5
  \wedge
  d\Big(
  \overset{
    \mathrlap{
      \mbox{
        \tiny
        \color{blue}
        \begin{tabular}{c}
          $
          \mathrlap{
          \!\!\!\!\!\!\!\!\!\!\!\!\!\!
          \mbox{
          super-exceptional
          topological Yang-Mills Lagrangian
          }}
          $
          \\
          D4 WZ Lagrangian
        \end{tabular}
      }
    }
  }
  {
    \overbrace{
    \underset{
      \mathclap{
      \mbox{
        \tiny
        \color{blue}
        \begin{tabular}{c}
          graviphoton
          \\
          RR potential
        \end{tabular}
      }
      }
    }{
    \underbrace{
      \mathclap{\phantom{A \atop A}}
      \!\!C_1
    }
    }
    \!\!
    \wedge
      \mathclap{ \phantom{A \atop A} }
        \tfrac{1}{2}
        F_{\mathrm{ex}_s}
        \wedge
        F_{\mathrm{ex}_s}
    }
  }
  \Big)
  \;+\;
  \underset{
    \mathclap{
      \mbox{
        \tiny
        \color{blue}
        \begin{tabular}{c}
          little string
          \\
          cochain in 5d
        \end{tabular}
      }
    }
  }{
    \underbrace{
      \mu^{5d}_{{}_{\rm L1}}
    }
  }
  \wedge
    \overbrace{
      \mathclap{ \phantom{A \atop A} }
        \tfrac{1}{2}
        F_{\mathrm{ex}_s}
        \wedge
        F_{\mathrm{ex}_s}
    }
  \\
  \;\; \in \;
  H^7
  \Big(
    \underset{
      \mathclap{
      \mbox{
        \tiny
        \color{blue}
        \begin{tabular}{c}
          super-exceptional homotopy-reduction
          \\
          on M-theory circle fiber
        \end{tabular}
      }
      }
    }{
    \underbrace{
    \overset{
      \mathclap{
      \mbox{
        \tiny
        \color{blue}
        \begin{tabular}{c}
          super-exceptional
          \\
          $\tfrac{1}{2}\mathrm{M5}$-spacetime
        \end{tabular}
      }
      }
    }{
    \overbrace{
    \mathclap{\phantom{A \atop A}}
    \big(
      \mathbb{R}^{5,1\vert \mathbf{8}}
      \times
      \mathbb{R}^1
    \big)_{\mathrm{ex}_s}
    }
    }
    \!\sslash\!
    \Omega S^2_{\mathrm{HW}}
    }
    }
  \Big)
\end{array}
$
}

Finally, we observe
(Remark \ref{EquivariantRelativeTrivialization})
that there are two extensions of the
compactified super-exceptional $\tfrac{1}{2}\mathrm{M5}$-spacetime
on which the D4 WZ-term becomes exact already before dimensional reduction:
one of these implements the heterotic Green-Schwarz mechanism and the
WZ-term of the heterotic NS5-brane (Remark \ref{Heterotic5Brane}).

With this, we may elegantly sum up the whole picture in the following
homotopy diagram:

\begin{equation}
  \label{M5LagrangianAsAHomotopy}
  \hspace{-1cm}
  \raisebox{125pt}{
  \xymatrix@R=22pt@C=45pt{
    &
    \ar@{}[rrrr]|-{
      \mbox{
        \tiny
        \begin{tabular}{c}
          {\color{blue} super-exceptional}
          \\
          { \color{blue} M5-brane cocycle }
          \\
          (Def. \ref{ExceptionalM5SuperCocycle})
        \end{tabular}
      }
    }
    &&&&
    \\
    \mbox{
      \tiny
      \begin{tabular}{c}
        {\color{blue} super-exceptional }
        \\
        {\color{blue} heterotic M-theory }
        \\
        {\color{blue} spacetime }
        \\
        (Def. \ref{SuperExceptionalMTheorySpacetime})
      \end{tabular}
      \hspace{-1cm}
    }
    \ar@(ul,ur)@[white]^{
      \mbox{
         \tiny
         \begin{tabular}{c}
           {\color{blue} super-exceptional }
           \\
           {\color{blue} ADE-action }
           \\
           (Lemma \ref{ReflectionAutomorphismOnExceptionalTangentSuperSpacetime})
         \end{tabular}
         \hspace{-1cm}
      }
    }
    &
    \big(
      \mathbb{T}^{9,1\vert \mathbf{16}}
      \times
      \mathbb{R}^1
    \big)_{\mathrm{ex}_s}
    \ar[rrrr]^-{
      \mathbf{dL}^{\!\!\mathrm{WZ}}_{\mathrm{ex}_s}
      \;\coloneqq\;
      (\pi_{\mathrm{ex}_s})^\ast
      \mu_{{}_{\rm M5}}
      +
      \frac{1}{2}
      H_{\mathrm{ex}_s} \wedge d H_{\mathrm{ex}_s}
    }_<<<<<<<<<<{\ }="s2"
    \ar@(ul,ur)^{ G_{\mathrm{ADE}} }
    &&&&
    S^7_{\mathbb{R}}
    \\
    \\
    \mbox{
      \tiny
      \begin{tabular}{c}
        {\color{blue} super-exceptional }
        \\
        {\color{blue} $\tfrac{1}{2}\mathrm{M5}$ spacetime }
        \\
        (Def. \ref{HalfM5LocusAndItsExceptionalTangentBundle})
      \end{tabular}
      \hspace{-1cm}
    }
    \ar@{}[uu]|-{
      \mbox{
        \tiny
        \begin{tabular}{c}
          {\color{blue} super-exceptional }
          \\
          {\color{blue} embedding }
          \\
          (Lemma \ref{SuperExceptionalEmbeddings})
        \end{tabular}
        \hspace{-1cm}
      }
    }
    \ar@{}[dd]|-{
      \mbox{
        \tiny
        \begin{tabular}{c}
          {\color{blue} super-exceptional }
          \\
          {\color{blue} KK-compactification }
          \\
          \eqref{SuperExceptionalHalfM5QuotientedByOmegaS2}
        \end{tabular}
        \hspace{-1cm}
      }
    }
    &
    \big(
      \mathbb{T}^{
        5,1
        \vert
        \mathbf{8}
      }
      \times
      \mathbb{T}^1
    \big)_{\mathrm{ex}_s}
    \ar@{^{(}->}[uu]^-{
      i_{\mathrm{ex}_s}
    }
    \ar[dd]_-{ q_{{}_{\Omega S^2_{\mathrm{HW}}}} }
    &&&&
    \\
    \\
    \mbox{
      \tiny
      \begin{tabular}{c}
        {\color{blue} super-exceptional }
        \\
        {\color{blue} heterotic $\tfrac{1}{2}\mathrm{M5}$ spacetime }
        \\
        {\color{blue} compactified on $S^1_{\mathrm{HW}}$ }
        \\
        (Def. \ref{HomotopyQuotientOfSuperExceptionalHalfM5Spacetime},
         Rem. \ref{Heterotic5Brane})
      \end{tabular}
      \hspace{-1cm}
    }
    &
    \Big(
      \big(
        \mathbb{T}^{5,1\vert \mathbf{8}}
        \times
        \mathbb{T}^1
      \big)_{\mathrm{ex}_s}^{\mathrm{het}}
    \Big)_{\sslash \Omega S^2_{\mathrm{HW}}}
    \ar@/^1pc/[uuuurrrr]|>>>>>>>>>>>>>>>>>>>>{\;\;\;\;\;\;\;\;\;\;\;\;\;\;\;
      \overset{
        \mbox{
          \tiny
          \begin{tabular}{c}
            {\color{blue} $\Omega S^2$-equivariant }
            \\
            {\color{blue} super-exceptional }
            \\
            {\color{blue} M5-brane cocycle }
            \\
            (Theorem \ref{SuperExceptionalPSEquivariantEnhancement})
            \\
            \phantom{a}
          \end{tabular}
        }
      }{
      (i_{\mathrm{ex}_s})^\ast\mathbf{dL}^{\!\!\mathrm{WZ}}_{\mathrm{ex}_s}
      -
      \omega^2 \wedge \mathbf{L}^{\!\!\mathrm{PS}}_{\mathrm{ex}_s}
      -
      \omega_3 \wedge \mathbf{L}^{\!\!\mathrm{tYM}}_{\mathrm{ex}_s}
      }
    }
    ^<<<<<<<<<<<<<<<<<<<<<<<<<<<<<{\ }="t2"
    _<<<<<<<<<<<<<<<<<<<<<<<<<<<<<{\ }="s"
    \ar[rrrr]^-{\ }="t"
    &&&&
    \ast
    \ar[uuuu]
    \ar@{=} "s2"; "t2"
    \ar@{=>}|-{
      \overset{
        \mbox{
          \tiny
          \begin{tabular}{c}
            {\color{blue} super-exceptional }
            \\
            {\color{blue} M5-brane Lagrangian }
            \\
            (Theorem \ref{TheTrivialization})
            \\
            \phantom{a}
          \end{tabular}
        }
      }{
        \mathbf{L}^{\!\!\mathrm{NG}}_{\mathrm{ex}_s}
        \,+\,
        \mathbf{L}^{\!\!\mathrm{PS}}_{\mathrm{ex}_s} \wedge e^5
        \,-\,
        \tfrac{1}{2}
        \mu_{{}_{L1}}
          \wedge
        H^{\mathrm{NS}}_{\mathrm{ex}_s}
      }
    }
    "s"; "t"+(5,0)
  }
  }
\end{equation}

\end{enumerate}

\newpage
\medskip

\section{Perry-Schwarz Lagrangian for $\mathrm{M5}$ on $S^1$ }
\label{PSLagrangianForM5OnS1}

For ease of reference and in order to introduce notation needed in later sections,
we review here the bosonic part of the Perry-Schwarz-Lagrangian from
\cite{PerrySchwarz97}, re-cast in coordinate-independent Cartan calculus
and generalized to possibly non-trivial circle fibrations.
We try to bring out the
logic that motivated the construction in \cite{PerrySchwarz97}, but below in
\cref{SuperExceptionalReductionOfMTheoryCircle}
and \cref{SuperExceptionalM5Lagrangian} we re-derive the Perry-Schwarz
Lagrangian systematically from first principles.
Readers familiar with this material may want to skip this section
and just follow pointers to it from the main text when needed.

\medskip

The formulation of a manifestly covariant Lagrangian for
the self-dual higher gauge field without further auxiliary fields
in 6 dimensions (and generally in $4k+2$-dimensions),
an hence in particular for the single M5-brane sigma-model,
is famously subtle, at best (see e.g. \cite{Moore12}\cite{HeckmanRudelius18}).
But if one considers breaking manifest Lorentz invariance
to  5 dimensions, as befits KK compactification of the theory on a circle fiber,
such as for double dimensional reduction of the M5 brane to the D4-brane,
then there is a Lagrangian formulation
due to Perry-Schwarz \cite{PerrySchwarz97}\cite{Schwarz97}\cite{APPS97},
following \cite{HenneauxTeitelboim88}.

\medskip
This ``non-covariant'' formulation of
self-dual higher gauge theory and specifically of the M5-brane sigma-model may be
covariantized by introducing an auxiliary scalar field \cite{PastiSorokinTonin96}
(whose gradient plays the role of the spacetime direction which gets singled out,
thus promoting this choice to a dynamical field) which yields the covariant formulation
of the M5-brane sigma-model \cite{PastiSorokinTonin97}\cite{BLNPST97}. This
comes with a corresponding auxiliary gauge symmetry that admits a gauge fixing
which recovers the non-covariant formulation, rendering the two formulations
equivalent, with each ``about as complicated'' as the other \cite[p. 3]{APPS97}.

\medskip

\noindent {\bf Worldvolume and self-duality.}
Let $(\Sigma^6, g)$ be a Lorentzian manifold of signature $(-,+,+,+,+,+)$,
to be called (the bosonic body of) the \emph{worldvolume}
of an M5-brane. In this dimension and with this signature, corresponding to the metric $g$,
the Hodge star operator on differential forms
$
  *
  \colon
  \Omega^\bullet( \Sigma^6)
  \to
  \Omega^{ 6-\bullet }( \Sigma^6)
$
squares to $+1$. This allows for considering on a differential 3-form
\begin{equation}
  H \;\in\; \Omega^3\big(\Sigma^6\big)
\end{equation}
the condition that it be self-dual
\begin{equation}
  \label{SelfDualityIn6d}
  H
  \;=\;
  *
  H
  \,.
\end{equation}
We will assume that $H$ is exact and pick a trivializing 2-form
\begin{equation}
  \label{BPotential}
  B \;\in\; \Omega^2\big( \Sigma^6 \big)
  \phantom{AAA}
  \mbox{such that}
  \phantom{AAA}
  H = d B
  \,.
\end{equation}

\medskip

\noindent {\bf Compactification on $S^1$.}
Consider then on the worldvolume $\Sigma^6$ the structure of an
$S^1 = U(1)$-principal bundle
\begin{equation}
  \label{FibrationWorldvolume}
  \xymatrix@R=1em{
    S^1 \ar[r]
    &
    \Sigma^6
    \ar[d]
    \\
    & \Sigma^5
  }
\end{equation}
We write
$
  v_5 \in \Gamma( T \Sigma^6 )
$
for the vector field which encodes the infinitesimal $S^1$-action,
hence the derivative of the circle action
$
  U(1) \times \Sigma^6 \overset{\rho}{\longrightarrow} \Sigma^6
$
at the neutral element, along a chosen basis element
$t \in T_e(U(1)) \simeq \mathfrak{u}(1) \simeq \mathbb{R}$:
\begin{equation}
  \label{FiberVectorField}
  v_5 \colon
  \xymatrix{
    \Sigma^6
     \simeq
    \{(e,t)\} \times \Sigma^6
    ~\ar@{^{(}->}[r]
    &
    T S^1 \oplus_{\Sigma^6} T \Sigma^6
    \simeq
    T \big( S^1 \times \Sigma^6\big)
    \ar[rr]^-{d \rho}
    &&
    T \Sigma^6.
  }
\end{equation}
Accordingly, we write
\begin{equation}
  \label{LieDerivativeAlongv5}
  \mathcal{L}_{v_5}
  \;\coloneqq\;
  \underset{
    \mathclap{
            d \circ \iota_{v_5}
      +
      \iota_{v_5} \circ d
    }
  }{
    \underbrace{
      \big[
        d, \iota_{v_5}
      \big]
    }
  }
  \;:\;
  \Omega^\bullet\big( \Sigma^6 \big)
  \longrightarrow
  \Omega^\bullet\big( \Sigma^6 \big)
\end{equation}
for the Lie derivative of differential forms along the
vector field \eqref{FiberVectorField}, where $d$ denotes the
de Rham differential and where under the brace we are using
Cartan's magic formula.

\medskip
Next, consider an Ehresmann connection on the $S^1$-bundle \eqref{FibrationWorldvolume},
hence a differential 1-form which satisfies the Ehresmann conditions in that it is
normalized and invariant:
\begin{equation}
  \label{EhresmannConditions}
  \theta^5
  \;\in\;
  \Omega^1\big( \Sigma^6 \big)
  \phantom{AAA}
  \mbox{such that}
  \phantom{AAA}
  \iota_{v_5} \theta^5 = 1
  \phantom{AA}
  \mbox{and}
  \phantom{AA}
  \mathcal{L}_{v_5} \theta^5 = 0
  \;.
\end{equation}
Here on the left we have the operation of contracting differential forms
with vector fields, and on the right we have the Lie derivative from
\eqref{LieDerivativeAlongv5}.
So, in particular, the composition
\begin{equation}
  \theta^5 \wedge \circ \iota_{v_5}
  :
  \Omega^\bullet\big( \Sigma^6\big)
  \longrightarrow
  \Omega^\bullet\big( \Sigma^6\big)
\end{equation}
is a projection operator:
$
  \theta^5 \wedge \iota_{v_5}
  \circ
  \theta^5 \wedge \iota_{v_5}
  \;=\;
  \theta^5 \wedge \iota_{v_5}
$.
The complementary projection is that onto horizontal differential forms
with respect to the bundle structure \eqref{FibrationWorldvolume}:
\begin{equation}
  \label{HorizontalProjection}
  (-)^{\mathrm{hor}}
  :=
  \big(\mathrm{id} - \theta^5 \wedge \circ \iota_{v_5})
  \;:\;
  \Omega^\bullet\big( \Sigma^6\big)
  \longrightarrow
  \Omega^\bullet\big( \Sigma^6\big).
\end{equation}
Observe that:
\begin{lemma}[Horizontal vs. vertical differential]
  \label{HorizontalDifferentialOfHorizontalComponent}
  If the Ehresmann connection \eqref{EhresmannConditions}
  on the $S^1$ bundle
  is flat, in that
  \begin{equation}
    \label{Flatness}
    d \theta^5 \;=\; 0
    \,,
  \end{equation}
  then for any differential form
  $\omega \in \Omega^\bullet\big( \Sigma^6 \big)$ we have
  that the vertical component of the differential of its
  horizontal component \eqref{HorizontalProjection} is the vertical component of its full differential:
  \begin{equation}
    \label{}
    \theta^5 \wedge d \big( \omega^{\mathrm{hor}} \big)
    \;=\;
    \theta^5 \wedge d \omega\;.
  \end{equation}
\end{lemma}
\begin{proof}
  By direct computation, we have:
  $$
    \begin{aligned}
      \theta^5 \wedge d \big( \omega^{\mathrm{hor}} \big)
      & =
      \theta^5 \wedge d \big( \omega - \theta^5 \wedge \iota_{v_5} \omega \big)
      \\
      & =
      \theta^5 \wedge d \omega
      -
      \theta^5 \wedge d
      \big(
        \theta^5 \wedge \iota_{v_5} \omega
      \big)
      \\
      & =
      \theta^5 \wedge d \omega
      -
      \underset{
        = 0
      }{
        \underbrace{
        \theta^5
        \wedge
        \big(
          d \theta^5
        \big)
        \wedge
        \iota_{v_5} \omega
        }
      }\;.
    \end{aligned}
  $$

  \vspace{-7mm}
\end{proof}

Finally, we require the vector field $v_5$ from \eqref{FiberVectorField} to be a
\emph{spacelike isometry}.
This means that it interacts with the Hodge star operator as
\begin{equation}
  \label{HodgeStarCommutingWithIsometryContraction}
  * \,\circ\, \iota_{v_5}
  =
  -
  \theta^5 \wedge \,\circ\, *
  \;:\;
  \Omega^3\big( \Sigma^6 \big)
  \longrightarrow
  \Omega^4\big( \Sigma^6 \big)
  \,.
\end{equation}

\medskip
\noindent {\bf Self-duality after $S^1$-compactification.}
We introduce notation for the contraction of the 3-form $H$ and its
Hodge dual with the vector field $v_5$ \eqref{FiberVectorField}
as follows (to be called the ``compactified fields'',
a notation that follows \cite[(5), (6)]{PerrySchwarz97}):
\begin{equation}
  \label{DefOfTildeH}
  \mathcal{F}
  \;\coloneqq\;
  \iota_{v_5} H
  \,,
  \phantom{AAAA}
  \tilde H
  \;\coloneqq\;
  \iota_{v_5} * H\;.
\end{equation}
With this, we get the following immediate but crucially important re-formulation
of the self-duality condition after $S^1$-compactification
(extending \cite[(8)]{PerrySchwarz97}):

\begin{lemma}[Self-duality after $S^1$-compactification.]
  Given an $S^1$-bundle structure \eqref{FibrationWorldvolume}
  on the worldvolume $\Sigma^6$
  and any choice of Ehresmann connection \eqref{EhresmannConditions},
  the self-duality condition \eqref{SelfDualityIn6d}
  is equivalently expressed in
  terms of the compactified fields \eqref{DefOfTildeH} as:
  \begin{equation}
    \label{SelfDualityIntermsOfcalF}
    H = \ast H
    \phantom{AAA}
    \Leftrightarrow
    \phantom{AAA}
    \mathcal{F}
    \;=\;
    \tilde H
    \;.
  \end{equation}
\end{lemma}
\begin{proof}
  We have the following chain of equivalences:
  $$
    \begin{aligned}
      H \;+\; \ast H
      & \;\;\;\Longleftrightarrow\;\;\;
      \left(
        \begin{aligned}
          \iota_{v_5} H & \;=\; \iota_{v_5} \ast H
          \\
          \mbox{and}\;\; \theta^5 \wedge H & \;=\; \theta^5 \wedge \ast H
        \end{aligned}
      \right)
      \\
      & \;\;\;\Longleftrightarrow\;\;\;
      \iota_{v_5} H \;=\; \iota_{v_5} \ast H
      \\
      & \;\;\;\Longleftrightarrow\;\;\;
      \mathcal{F} \;=\; \tilde H\;.
    \end{aligned}
  $$
  Here the first step is decomposition into horizontal and
  vertical components \eqref{HorizontalProjection},
  the second step uses
  the isometry property \eqref{HodgeStarCommutingWithIsometryContraction}
  to conclude that the two resulting component equations are
  equivalent to each other.
  The last step identifies the compactified fields
  \eqref{DefOfTildeH}.
\end{proof}

\medskip

\noindent {\bf The gauge field.}
The contraction of the vector field $v_5$ from \eqref{FiberVectorField}
with  the 2-form potential $B$ from \eqref{BPotential}
defines the 1-form potential
\begin{equation}
  \label{AField}
  A
  \;\coloneqq\;
  -
  \iota_{v_5} B\;.
\end{equation}
Hence we get a decomposition of the 2-form as
\begin{equation}
  B
  \;=\;
  A \wedge \theta^5
  +
  B^{\mathrm{hor}}
  \,,
\end{equation}
where on the right we have the horizontal component of $B$
according to \eqref{HorizontalProjection}.
We say that the \emph{2-flux density} encoded by $B$
is the horizontal component of the exterior differential of
this vector potential
\begin{equation}
  \label{FaradayTensor}
  F
  \;\coloneqq\;
  (d A)^{\mathrm{hor}}.
\end{equation}
We will find in a moment that this is the 5d field strength with all higher KK-modes
still included, but it is most convenient here (and in all of the following0
to just call it ``$F$'' already in the 6d compactification before
passing to KK zero-modes. With this we have (cf. \cite[(5)]{PerrySchwarz97}):

\begin{lemma}[Shifted 2-flux]
\label{RelatingFcalF}
The 2-flux density $F$ from \eqref{FaradayTensor}
equals the compactified field $\mathcal{F}$ from \eqref{DefOfTildeH}
up to the Lie derivative \eqref{LieDerivativeAlongv5}
of the horizontal component \eqref{HorizontalProjection} of the
2-form \eqref{BPotential}:
\begin{equation}
  \label{FcalF}
  F
  \;=\;
  \mathcal{F} - \mathcal{L}_{v_5}B^{\mathrm{hor}}.
\end{equation}
\end{lemma}
\begin{proof}
We compute as follows:
\begin{equation}
  \label{calF}
  \begin{aligned}
    \mathcal{F}
    &\coloneqq
    \iota_{v_5} H
    \\
    & =
    \iota_{v_5} d B
    \\
    & =
    - d \iota_{v_5} B + [\iota_{v_5}, d] B
    \\
    & =
    d A + \mathcal{L}_{v_5} B
    \\
    & =
    \underset{
      = F
    }{
      \underbrace{
        (d A)^{\mathrm{hor}}
      }
    }
    +
    \underset{
      = \theta^5 \wedge \mathcal{L}_{v_5} A
    }{
      \underbrace{
        \theta^5 \wedge \iota_{v_5} d A
      }
    }
    +
    \mathcal{L}_{v_5} B^{\mathrm{hor}}
    +
    \underset{
      -\theta^5 \wedge \mathcal{L}_{v_5} A
    }{
      \underbrace{
        \mathcal{L}_{v_5} \theta^5 \wedge \iota_{v_5} B
      }
    }
    \\
    &
    =
    F
    +
    \mathcal{L}_{v_5} B^{\mathrm{hor}}
    .
  \end{aligned}
\end{equation}
Here the first step is the definition \eqref{DefOfTildeH},
while the second step is \eqref{BPotential}.
The fourth step uses the definition \eqref{AField}
of the vector potential and identifies
the Lie derivative \eqref{LieDerivativeAlongv5}.
The fifth step applies vertical/horizontal decomposition
\eqref{HorizontalProjection} to both summands and
uses \eqref{FaradayTensor} under the first brace
and the expressions \eqref{EhresmannConditions} and \eqref{AField}
under the third brace; while under the second brace it uses Cartan's formula
\eqref{LieDerivativeAlongv5}, observing that $\iota_{v_5} A = 0$
by definition \eqref{AField} and by nilpotency of the
contraction operation. The last step notices that this makes the second
and fourth summands cancel each other.
\end{proof}

With Lemma \ref{RelatingFcalF}, the self-duality condition \eqref{SelfDualityIn6d}
in the equivalent form \eqref{SelfDualityIntermsOfcalF} after $S^1$-compactification
says that the combination
$
  \widetilde H
  -
  \mathcal{L}_{v_5} B^{\mathrm{hor}}
$
is horizontally exact:
\begin{equation}
  \label{SelfDualityAsHorizontalExactness}
  H \;=\;\ast H
  \phantom{AAA}
  \Leftrightarrow
  \phantom{AAA}
  \widetilde H
  -
  \mathcal{L}_{v_5} B^{\mathrm{hor}}
  \;=\;
  \underset{
    = F
  }{
    \underbrace{
      (d A)^{\mathrm{hor}}
    }
  }
  \,.
\end{equation}

\medskip

\noindent {\bf Weak self-duality and PS equations of motion.}
In summary, we have the following implication of self-duality after $S^1$-compactification
(extending \cite[(16)]{PerrySchwarz97}):
\begin{prop}[Self-duality for flat circle bundles]
  If the worldvolume $\Sigma^6$ is equipped with an $S^1$-principal
  bundle structure \eqref{FibrationWorldvolume} which is flat \eqref{Flatness},
  then the self-duality condition $H = \ast H$ from \eqref{SelfDualityIn6d},
  in its equivalent incarnation
  on compactified fields \eqref{SelfDualityAsHorizontalExactness}
  implies the following differential equation:
  \begin{equation}
    \label{WeakSelfDuality}
    \theta^5 \wedge d
    \big(
      \tilde H - \mathcal{L}_{v_5} B^{\mathrm{hor}}
    \big)
    \;=\;
    0\;.
  \end{equation}
\end{prop}
By Lemma \ref{HorizontalDifferentialOfHorizontalComponent},
equation \eqref{WeakSelfDuality} may be understood as expressing
``self-duality up to horizontally exact terms''.
The proposal of \cite{PerrySchwarz97}
is to regard \eqref{WeakSelfDuality},
which is second order as a differential equation for $B$,
as the defining \emph{equation of motion} for a self-dual field on $\Sigma^6$ compactified on $S^1$. Given this, one is led
to finding a Lagrangian density
whose Euler-Lagrange equation
is the self-duality equation.

\medskip

\noindent {\bf The PS Lagrangian density.}
As in \cite[(17)]{PerrySchwarz97}, we say:
\begin{defn}[Bosonic PS Lagrangian]
\label{BosonicPSLagrangian}
The \emph{bosonic Perry-Schwarz-Lagrangian}
for a 2-form field \eqref{BPotential}
on a worldvolume $\Sigma^6$ compactified on
a flat \eqref{Flatness} $S^1$-bundle \eqref{FibrationWorldvolume}
is:
\begin{equation}
  \label{CartanCalculusPerrySchwarzLagrangian}
  B
  \;\longmapsto\;
  \mathbf{L}^{\!\!\mathrm{PS}}
  \wedge \theta^5
  \;\coloneqq\;
  \tfrac{1}{2}
  \big(
  \mathcal{L}_{v_5} B^{\mathrm{hor}}
  -
  \tilde H
  \big)
  \wedge
  *
  \tilde H\;.
\end{equation}
\end{defn}
\begin{remark}[Equivalent descriptions]
{\bf (i)} Equivalently, using the isometry property
\eqref{HodgeStarCommutingWithIsometryContraction}
the bosonic PS Lagrangian \eqref{CartanCalculusPerrySchwarzLagrangian}
reads
\begin{equation}
  \label{BosonicPSLagrangianWithH}
  \begin{aligned}
  \mathbf{L}^{\!\!\mathrm{PS}} \wedge \theta^5
  & =
  -
  \tfrac{1}{2}
  \big(
  \tilde H
  -
  \mathcal{L}_{v_5} B^{\mathrm{hor}}
  \big)
  \wedge
  H
  \wedge \theta^5
  \\
  & =
  -
  \tfrac{1}{2}
  \big(
  \iota_{v_5} * H
  -
  \mathcal{L}_{v_5} B^{\mathrm{hor}}
  \big)
  \wedge
  H
  \wedge \theta^5\;,
  \end{aligned}
\end{equation}
where in the second line we inserted the expression for $\widetilde{H}$ from
\eqref{DefOfTildeH};
\item {\bf (ii)} and if the self-duality condition \eqref{SelfDualityIn6d}
is imposed, the bosonic PS Lagrangian \eqref{BosonicPSLagrangianWithH}
becomes
\begin{equation}
  \label{NonCovariantLagrangianInFlatSpacetime}
  \mathbf{L}^{\!\!\mathrm{PS}} \wedge \theta^5
  \;=\;
  -
  \tfrac{1}{2}
  \big(
  \iota_{v_5} H
  -
  \mathcal{L}_{v_5} B^{\mathrm{hor}}
  \big)
  \wedge
  H
  \wedge \theta^5
  \phantom{AAA}
  \mbox{if} \;\; H = * H\;.
\end{equation}
\end{remark}

At this point, as in \cite[(16)]{PerrySchwarz97}, the bosonic PS Lagrangian
of Def. \ref{BosonicPSLagrangian} is motivated as the evident choice that
makes the following Prop. \ref{EquationOfMotionOfPSLagrangian} true.
Below in \cref{SuperExceptionalLagrangian}, we find a deeper origin of the Lagrangian
(in the case that self-duality is imposed).

\begin{prop}[EOMs of the bosonic PS Lagrangian]
  \label{EquationOfMotionOfPSLagrangian}
  The Euler-Lagrange equation corresponding to the
  PS Lagrangian $\mathbf{L}^{\!\!\mathrm{PS}} \wedge \theta^5$ \eqref{CartanCalculusPerrySchwarzLagrangian}
  is the weak self-duality equation \eqref{WeakSelfDuality}.
\end{prop}
\begin{proof}
  We may evidently regard the PS Lagrangian as the quadratic part of the
  following bilinear form on
  differential 2-forms with values in differential 6-forms:
  $$
    (B,B')
    \;\longmapsto\;
  -
  \tfrac{1}{2}
  \big(
  \iota_{v_5} * H
  -
  \mathcal{L}_{v_5} B^{\mathrm{hor}}
  \big)
  \wedge
  d B'
  \wedge \theta^5.
  $$
  Observe then that this bilinear form is symmetric up to a total differential:
  the first summand is strictly symmetric, as it is the standard Hodge pairing,
  while for the second summand symmetry up to a total derivative is established
  by a local integration by parts.
  Together these imply that the variational Euler-Lagrange derivative of the
  PS Lagrangian is twice the result of varying just the second factor of
  $B$:
  $$
    \begin{aligned}
      \delta \mathbf{L}^{\!\!\mathrm{PS}} \wedge \theta^5
      & =
      -
      \big(
      \iota_{v_5} * H
      -
      \mathcal{L}_{v_5} B^{\mathrm{hor}}
      \big)
      \wedge
      d (\delta B)
      \wedge \theta^5
      \\
      & =
      -
      \Big(
      \theta^5
      \wedge
      d
      \big(
      \iota_{v_5} * H
      -
      \mathcal{L}_{v_5} B^{\mathrm{hor}}
      \big)
      \Big)
      \wedge
      (\delta B)\;.
    \end{aligned}
  $$
  Hence the vanishing of the variational derivative is
  equivalent to \eqref{WeakSelfDuality}.
\end{proof}

\medskip
\noindent {\bf Reduction to 5d Maxwell theory.}
Consider finally
the special case of Kaluza-Klein compactification/double dimensional
reduction, where
\begin{equation}
  \mathcal{L}_{v_5} B^{\mathrm{hor}} = 0
  \phantom{AAA}
  \mbox{ for KK-reduction to D4}.
\end{equation}
In this case, expression \eqref{calF} reduces to
\begin{equation}
  \label{FFromH}
  \iota_{v_5} H = F
  \phantom{AAA}
  \mbox{if}
  \phantom{AA}
  \begin{array}{l}
    \mathcal{L}_{v_5} B^{\mathrm{hor}} = 0\;.
  \end{array}
\end{equation}
Hence, with self-duality (\cite[above (16)]{PerrySchwarz97}), we have
\begin{equation}
  \label{KKH}
  H = F \wedge \theta^5 + *_5 F \;,
\end{equation}
where now
$
  *_5
  :
  \Omega^\bullet( \Sigma^5)
  \to
  \Omega^{5-\bullet}( \Sigma^5)
$
is the Hodge star operator on the base 5-manifold.
Consequently, we have extension of \cite[above (16)]{PerrySchwarz97} to the topologically nontrivial setting:
\begin{prop}[5d Abelian Yang-Mills from Perry-Schwarz]
\label{MaxwellFromPS}
The $S^1$-dimensional reduction of the PS Lagrangian \eqref{NonCovariantLagrangianInFlatSpacetime}
is the Lagrangian of
5d Maxwell theory for the vector potential $A$
from \eqref{AField}:
\begin{equation}
  \label{5dYMFromPS}
  A
  \;\longmapsto\;
  \mathbf{L}^{\!\!\mathrm{PS}} \wedge \theta^5
  \;=\;
  -
  \tfrac{1}{2}
  \big(
    H \wedge \iota_v H
  \big)
  \;=\;
  -
  \tfrac{1}{2}
  \big(
    F
    \wedge
    *_5 F
  \big)
  \wedge
  \theta^5
  \phantom{AAAA}
  \mbox{if}
  \phantom{A}
     H = * H
    \phantom{A} \mbox{and} \phantom{A}
    \mathcal{L}_{v_5} B^{\mathrm{hor}} = 0\;.
\end{equation}
\end{prop}

\medskip

\noindent {\bf 6d self-dual field as (abelian) 5d Yang-Mills with KK-modes.}
Due to Prop. \ref{MaxwellFromPS}, one may regard the general bosonic
Perry-Schwarz Lagrangian  (Def. \ref{BosonicPSLagrangian}) for the self-dual
field on $\Sigma^6$ compactified on $S^1$ as that of (abelian)
\emph{5d Yang-Mills theory with a tower of KK-modes included},
which is a perspective on the M5-brane theory later advanced in
\cite{Douglas10}\cite{LPS10} (see \cite[3.4.3]{Lambert19})).
Here for the abelian bosonic sector, this perspective may be fully brought
out by introducing the following notation:
\begin{equation}
  \label{FTilde}
    F
    \;\coloneqq\;
    \iota_{v_5} H
    -
    \mathcal{L}_{v_5} B^{\mathrm{hor}}
    \qquad
    \text{and}
    \qquad
    \widetilde F
    \;\coloneqq\;
    H - F \wedge \theta^5,
\end{equation}
where on the left we have \eqref{FcalF},
as before, while on the right we are introducing notation for
the remaining summand in $H$.
With this notation, the bosonic PS Lagrangian, assuming self-duality \eqref{SelfDualityIn6d}
as in \eqref{NonCovariantLagrangianInFlatSpacetime},
finds the following suggestive expression:
\begin{equation}
  \label{YMFormOfPS}
  \mathbf{L}^{\!\!\mathrm{PS}}
  \;=\;
  -
  \tfrac{1}{2}
  F \wedge \widetilde F
  \phantom{AAAA}
  \mbox{if $\;H = \ast H$}
  \,,
\end{equation}
so that KK-reduction to 5d YM theory,
as in Prop. \ref{MaxwellFromPS},
is now given syntactically
simply by replacing $\widetilde F \mapsto \ast_5 F$.
We find the form \eqref{YMFormOfPS} of the PS Lagrangian
to be reflected by its super-exceptionalization
in Prop. \ref{ExceptionalPreimageOfPerrySchwarzLagrangian} below;
see Def. \ref{SuperExceptionalPSLagrangian}
and Remark \ref{CocycleForM5OnS1AsD4PlusKK}.

\medskip

\noindent {\bf The topological 5d Yang-Mills Lagangian.}
After dimenional reduction
\cite[Sec. 6]{APPS97}\cite[Sec. 6 \& App. A]{APPS97b},
the Perry-Schwarz Lagrangian for the M5-brane
is accompanied by
a multiple of the \emph{topological Yang-Mills Lagrangian}
\cite{BaulieuSinger88}\cite{vanBaal90}
\begin{equation}
  \label{TopologicalYM}
  \mathbf{L}^{\!\!\mathrm{tYM}}
  \;\coloneqq\;
  F \wedge F \;.
\end{equation}
We find this arise from the super-exceptional
embedding construction below in Theorem \ref{TheTrivialization},
see Remark \ref{D4WZ}.

\medskip

\section{Super-exceptional M-geometry }
\label{SuperExceptionalMGeometry}

We recall
(in Def. \ref{ExceptionalTangentSuperSpacetime} and
Prop. \ref{TransgressionElementForM2Cocycle})
the ``hidden supergroup of $D= 11$ supergravity''
\cite{DF}\cite{BAIPV04}\cite{ADR16}
interpreted as super-exceptional M-theory spacetime
\cite[4.6]{FSS18}\cite{SS18}
and explain, in Remark \ref{SuperExceptionalTargetForM5},
how the results of \cite{FSS19b}\cite{FSS19c}
identify this as the correct local target for M5-brane sigma-models.

\medskip
For definiteness, we make explicit our spinor and Clifford algebra convention:
\begin{defn}[e.g. {\cite[Prop. A.3 (v)]{HSS18}}]
\label{CliffordAlgebraRepresentation}
We write $\mathbf{32}$ for the irreducible real representation
of $\mathrm{Pin}^+(10,1)$, via a Clifford-algebra representation
$\{\pmb{\Gamma_a}\}_{a = 0}^{10}$ whose generators
satisfy the following relations:
\begin{equation}
  \label{CMinusMajoranaRepresentationForPin}
\begin{aligned}
  &
  \pmb{\Gamma}_a \pmb{\Gamma}_b + \pmb{\Gamma}_b \pmb{\Gamma}_a
  =
  + 2\eta_{a b}
  :=
  2 \mathrm{diag}( -1, +1, +1, \cdots, +1 )_{a b}
  \;,
  \\
  &\left( \pmb{\Gamma}_0\right)^2 = -1\;,
  \phantom{AAAA}\;
  \left( \pmb{\Gamma}_a\right)^2 = +1\;,
 \\
  &( \pmb{\Gamma}_0)^\dagger = - \pmb{\Gamma}_0\;,
  \phantom{AAAA}
  ( \pmb{\Gamma}_a)^\dagger = +  \pmb{\Gamma}_a\;,
   \phantom{AA}
    \mbox{for $a \in \{1,\cdots, 10\}$}.
\end{aligned}
\end{equation}
\end{defn}
\begin{remark}[Gamma matrix conventions]
Def. \ref{CliffordAlgebraRepresentation}
relates to an alternative Clifford algebra convention
$\{\Gamma_a\}_{a=0}^{10}$ used
in much of the relevant literature (e.g. \cite{DF}\cite{BAIPV04})
via the relation
$$
  {\pmb{\Gamma}}_a \coloneqq  i \, \Gamma_a
$$
understood inside a complex Dirac representation
(see {\cite[Prop. A.3]{HSS18}}).
We will use ``\;$\Gamma_a$'' for exhibiting expressions
manifestly compatible with the literature, and
``\;$\pmb{\Gamma}_a$'' for emphasizing the
actual real $\mathrm{Pin}^+(10,1)$-action.
\end{remark}

The following, Def. \ref{ExceptionalTangentSuperSpacetime} and
Prop. \ref{TransgressionElementForM2Cocycle},
are a formulation in rational super-homotopy theory
due to  \cite[Sec. 4.5]{FSS18}, of the
classical supergravity results in
\cite[Sec. 6]{DF}\cite[Sec. 3]{BAIPV04} (see also \cite{ADR16}) .
The Definition \ref{ExceptionalTangentSuperSpacetime}
of super-exceptional spacetime involves a parameter $s$
\cref{s}, which arises mathematically in Prop. \ref{TransgressionElementForM2Cocycle}
from different possibilities of decomposing the $H_3$-flux
on super-exceptional spacetime \cite{BAIPV04}.
We discover the physical meaning of this parameter below in
\cref{SuperExceptionalLagrangian}.

\begin{defn}[Super-exceptional M-theory spacetime]
  \label{ExceptionalTangentSuperSpacetime}
For parameter
\begin{equation}
  \label{s}
  s \in \mathbb{R} \setminus \{0\}
\end{equation}
the
\emph{$D= 11$, $\mathcal{N} = 1$, $n = 11$
super-exceptional M-theory spacetime}
$(\mathbb{T}^{10,1\vert \mathbf{32}})_{\mathrm{ex}_s}$
over ordinary
\emph{$D= 11$, $\mathcal{N} = 1$ super Minkowski spacetime}
$\mathbb{T}^{10,1\vert \mathbf{32}}$
is the rational super space given dually by the following
super dgc-algebra:
\begin{equation}
  \label{FermionicExtensionOfExceptionalTangentSuperspacetime}
  \hspace{-6cm}
  \raisebox{90pt}{
  \xymatrix{
    \big(
      \mathbb{T}^{10,1\vert \mathbf{32}}
    \big)_{\mathrm{ex}_s}
    \ar[dd]_{ \pi_{\mathrm{ex}_s} }
    &&
  \mathbb{R}\left[
    {\begin{array}{l}
      \underset{
        \mathrm{deg} = (1,\mathrm{even})
       }{
        \underbrace{
          \big\{
            e^a
          \}_{ 0 \leq a \leq 10 }
        }
      }
      \\
      \underset{\mathrm{deg} = (1,\mathrm{even}) }{
        \underbrace{
        \big\{
          e_{a_1 a_2}
        \big\}_{ 0 \leq a_1 \lt a_2 \leq 10 }
        }
      }
      \\
      \underset{ \mathrm{deg} = (1,\mathrm{even}) }{
        \underbrace{
          \big\{
            e_{a_1 \cdots a_5}
          \big\}_{ 0 \leq a_1 \lt \cdots \lt a_5 \leq 10 }
        }
      }
      \\
      \underset{ \mathrm{deg} = (1,\mathrm{odd}) }{
        \underbrace{
          \big\{
            \psi^\alpha
          \big\}_{0 \leq \alpha \leq 32 }
        }
      }
      \\
      \underset{ (1,\mathrm{odd}) }{
        \underbrace{
          \big\{
            \eta^\alpha
          \big\}_{ 0 \leq \alpha \leq 32 }
         }
      }
      \end{array}}
    \right]
    \Big/
    \mathrlap{
  \left(
    {\begin{array}{lcl}
      d \,\psi^\alpha
        & \!\!\!\!\!\!\!\!\!  = & \!\!\!\!\!\!
      0,
      \\
      d e^a & \!\!\!\!\!\!\!\!\!   =&  \!\!\!\!\!\! \overline{\psi}\; {\Gamma}^a \psi,
      \\
      d \, e_{a_1 a_2}
        & \!\!\!\!\!\!\!\!\!   = & \!\!\!\!\!\!
      \tfrac{i}{2}\overline{\psi} \; {\Gamma}_{a_1 a_2} \psi,
      \\
      d \, e_{a_1 \cdots a_5}
        &\!\!\!\!\!\!\!\!\!  =& \!\!\!\!\!\!
      \tfrac{1}{5!}\overline{\psi} \; {\Gamma}_{a_1 \cdots a_5} \psi,
      \\
      d \,\eta
        &\!\!\!\!\!\!\!\!\!  =& \!\!\!\!\!\!
       (s+1)
       e^a \wedge {\pmb{\Gamma}}_a \psi
       \\
       & &
       +
       e_{a_1 a_2} \wedge {\pmb{\Gamma}}^{a_1 a_2} \psi
       \\
       & &
       +
       (1 + \tfrac{s}{6})
       e_{a_1 \cdots a_5} \wedge {\pmb{\Gamma}}^{a_1 \cdots a_5} \psi
    \end{array}}
  \!\!\right)
  }
  \ar@{<-}[dd]^-{
    \mbox{
      \tiny
      $
      \begin{array}{ccc}
        \psi^\alpha & e^a
        \\
        \mapsup & \mapsup
        \\
        \psi^\alpha & e^a
      \end{array}
      $
    }
  }
  \\
  \\
  \mathbb{T}^{10,1\vert \mathbf{32}}
  &&
  \mathbb{R}\left[
    {\begin{array}{l}
      \underset{
        \mathrm{deg} = (1,\mathrm{even})
       }{
        \underbrace{
          \big\{
            e^a
          \}_{ 0 \leq a \leq 10 }
        }
      }s
      \\
      \underset{ \mathrm{deg} = (1,\mathrm{odd}) }{
        \underbrace{
          \big\{
            \psi^\alpha
          \big\}_{0 \leq \alpha \leq 32 }
        }
      }
      \end{array}}
    \right]
    \Big/
    \mathrlap{
  \left(
    {\begin{array}{lcl}
      d \,\psi^\alpha
        & \!\!\!\!\!\!\!\!  = & \!\!\!\!\!\!
      0,
      \\
      d e^a
      & \!\!\!\!\!\!\!\!  = & \!\!\!\!\!\!
      \overline{\psi}\; {\Gamma}^a \psi
    \end{array}}
  \right)
  }
  }
  }
\end{equation}
Here the index $\alpha$ ranges over a linear basis
of the real $\mathrm{Pin}^+(10,1)$-representation
$\mathbf{32}$ and the Clifford generators $\Gamma_a$
acting on these are as in Def. \ref{CliffordAlgebraRepresentation};
and
we use Einstein summation convention
with the $e_{a_1 a_2}$ and $e_{a_1 \cdots a_5}$
understood as completely antisymmetrized in their indices.
\end{defn}

\begin{remark}[Super-exceptional M-theory spacetime as a supermanifold]
\label{SuperExceptionalSpacetimeAsManifold}
We may alternatively regard
$\big(\mathbb{T}^{10,1\vert \mathbf{32}}\big)_{\mathrm{ex}_s}$
from Def. \ref{ExceptionalTangentSuperSpacetime}
as a super-manifold with canonical global coordinate functions
$$
  C^\infty
  \Big(
    \big(
      \mathbb{R}^{10,1\vert \mathbf{32}}
    \big)_{\mathrm{ex}_s}
  \Big)
  \;=\;
  \Big\langle
  \underset{
    \mathrm{deg} = (0, \mathrm{even})
  }{
  \underbrace{
    (x^a)
  }},
  \;\;\;
  \underset{
    \mathrm{deg} = (0,\mathrm{even})
  }{
  \underbrace{
    (B_{a_1 a_2})
  }},
  \;\;\;
  \underset{
    \mathrm{deg} = (0,\mathrm{even})
  }{
  \underbrace{
    (B_{a_1 \cdots a_5})
  }},
  \;\;\;
  \underset{
    \mathrm{deg} = (0,\mathrm{odd})
  }{
  \underbrace{
    (\theta^\alpha)
  }},
  \;\;\;
  \underset{
    \mathrm{deg} = (0,\mathrm{odd})
  }{
  \underbrace{
    (\rho^\alpha)
  }}
  \Big\rangle
  \,,
$$
hence with bosonic part being the exceptional tangent bundle \eqref{EGTangent} for maximal $n =11$:
$$
  \Big(
    \big(
      \mathbb{R}^{10,1\vert \mathbf{32}}
    \big)_{\mathrm{ex}_s}
  \Big)_{\mathrm{bos}}
  \;\simeq\;
  T \mathbb{R}^{10,1}
    \;\oplus\;
  \wedge^2 T^\ast \mathbb{R}^{10,1}
    \;\oplus\;
  \wedge^5 T^\ast \mathbb{R}^{10,1}
  \,,
$$
and with super-group structure such that
the Chevalley-Eilenberg algebra in Def. \ref{ExceptionalTangentSuperSpacetime}
identifies with the super de Rham dgc-algebra of
left-invariant (hence translationally supersymmetric) super-differential forms:
\begin{equation}
  \label{LeftInvariantVielbeingInTermsOfCanonicalCoordinates}
  \xymatrix@R=-2pt{
    \mathrm{CE}
    \Big(
      \big(
        \mathbb{T}^{10,1\vert \mathbf{32}}
      \big)_{\mathrm{ex}_s}
    \Big)
    \ar[r]^-{\simeq}
    &
    \Omega^\bullet_{\mathrm{li}}
    \Big(
      \big(
        \mathbb{T}^{10,1\vert \mathbf{32}}
      \big)_{\mathrm{ex}_s}
    \Big)
    ~\ar@{^{(}->}[rrr]
    &&&
    \Omega^\bullet
    \Big(
      \big(
        \mathbb{R}^{10,1\vert \mathbf{32}}
      \big)_{\mathrm{ex}_s}
    \Big).
    \\
    \psi^\alpha
      \ar@{|->}[r]
      &
    \mathrlap{
      \!\!\!\!\!\!\!\!\!\!\!
      \phantom{\tfrac{1}{\alpha_0(s)}}
      d\theta^\alpha
    }
    &&
    \\
    \eta^\alpha
      \ar@{|->}[r]
      &
    \mathrlap{
      \!\!\!\!\!\!\!\!\!\!\!
      \phantom{\tfrac{1}{\alpha_0(s)}}
      d\rho^\alpha
    }
    \\
    e^a
      \ar@{|->}[r]
    &
    \mathrlap{
      \!\!\!\!\!\!\!\!\!\!\!
        \phantom{\tfrac{1}{\alpha_0(s)}}
        d x^a +
        \overline{\theta} \Gamma^a d \theta
    }
    \\
    e_{a_1 a_2}
      \ar@{|->}[r]
      &
      \mathrlap{
        \!\!\!\!\!\!\!\!\!\!\!\!
        \phantom{\tfrac{1}{\alpha_0(s)}}
        d (B_{a_1 a_2})
        +
        \tfrac{i}{2}\overline{\theta} \Gamma_{a_1 a_2} d \theta
      }
    \\
    e_{a_1 \cdots a_5}
     \ar@{|->}[r]
      &
    \mathrlap{
      \!
      \tfrac{1}{\alpha_0(s)}
      d (B_{a_1 \cdots a_5})
      +
      \tfrac{1}{5!}
      \overline{\theta} \Gamma_{a_1 \cdots a_5} d \theta\;.
    }
  }
\end{equation}
Beware the bracketing in the last two lines on the right,
in contrast to $(d B)_{a_1 \cdots a_3}$ etc.
The bosonic component of the generator $e_{a_1 a_2}$
is the de Rham differential
of the  bosonic component functions
of a 2-form $B \coloneqq B_{a_1 a_2} d x^a \wedge dx^b$
$$
  d(B_{a_1 a_2}) \;=\; d x^\mu \partial_\mu B_{a_1 a_2}
$$
(without antisymmetryization over all three indices, at this point),
instead of the component functions of the de Rham differential of a 2-form,
which is instead obtained by anti-symmetrizing over all three indices
$$
  H
    \;\coloneqq\;
  \big( d(B_{a_1 a_2}) \big) \wedge d e^{a_1} \wedge d e^{e_2}
$$
as in \eqref{HwithsMinusThree} below.
This has a crucial effect in the following discussion;
see Lemma \ref{ExceptionalPreimageOfPerrySchwarzLagrangian}.
\end{remark}

\begin{prop}[Transgression of M2-cocycle on super-exceptional spacetime]
  \label{TransgressionElementForM2Cocycle}
  For $s \in \mathbb{R} \setminus \{0\}$, the fermionic extension of exceptional tangent superspacetime
  $\mathbb{R}^{10,1\vert \mathbf{32}}_{\mathrm{ex}_s}$ (Def. \ref{FermionicExtensionOfExceptionalTangentSuperspacetime}),
  regarded as fibered over 11d super-Minkowski spacetime
  $\mathbb{R}^{10,1\vert \mathbf{32}}$
  carries a transgression element
  for the M2-brane 4-cocycle $\mu_{{}_{\rm M2}}$ \eqref{TheMBraneCocycles}:
  \begin{equation}
    \label{Hexs}
    H_{\mathrm{ex}_s}
      \;\in\;
    \mathrm{CE}
    \Big(
      \big(
        \mathbb{R}^{10,1\vert \mathbf{32}}
      \big)_{\mathrm{ex}_s}
    \Big)
    \phantom{AA}
    \mbox{\text such that}
    \;\;
    d H_{\mathrm{ex}_s}
    \;=\;
    (\pi_{\mathrm{ex}_s})^\ast
    \mu_{{}_{\mathrm{M2}}}
  \end{equation}
    given by
\begin{equation}
  \label{DecomposedCFieldAsSumOfBosonicAndFermionicContribution}
  H_{\mathrm{ex}_s}
  \;=\;
  \underset{
    = (d B_{a_1 a_2}) \wedge e_{a_2} \wedge e_{a_3}  =: H
  }{
  \underbrace{
    \alpha_0(s)
    \,
    e_{a_1 a_2} \wedge e^{a_1} \wedge e^{a_2}
  }
  }
   -
   \alpha_3(s)
   \epsilon_{a_1 \cdots a_5 b_1 \cdots b_5 c}
   e^{a_1 \cdots a_5}
     \wedge
   e^{b_1 \cdots b_5} \wedge e^c
  +
  \mathrm{H}^{\mathrm{fib}}_s
\end{equation}
with
\begin{equation}
  \label{DecomposedCFieldBosonicContribution}
\begin{aligned}
   (H^{\mathrm{fib}}_s)_{\mathrm{bos}}
   &  =
     \phantom{+}\,
     \alpha_1(s)
     e^{a_1}{}_{a_2}
       \wedge
     e^{a_2}{}_{a_3}
       \wedge
     e^{a_3}{}_{a_1}
    \\
    &
    \phantom{=}\;
    +
    \alpha_2(s)
    e_{b_1 a_1 \cdots a_4}
      \wedge
    e^{b_1}{}_{b_2}
      \wedge
    e^{b_2 a_1 \cdots a_4}
    \\
    &
    \phantom{=}\;
    +
    \alpha_4(s)
    \epsilon_{\alpha_1 \cdots \alpha_6 b_1 \cdots b_5}
    e^{a_1 a_2 a_3}{}_{c_1 c_2}
      \wedge
    e^{a_4 a_5 a_6 c_1 c_2}
      \wedge
    e^{b_1 \cdots b_5}
  \end{aligned}
\end{equation}
and
\begin{equation}
  \label{DecomposedCFieldFermionicContribution}
  (H^{\mathrm{fib}}_s)_{\mathrm{ferm}}
  =
  -
  \tfrac{1}{2}
  \overline{\eta}_\alpha \wedge \psi^{\beta}
  \wedge
  \Big(
    \beta_1(s) (\Gamma_a)^\alpha{}_\beta \; e^a
    +
    \beta_2(s)
    (\Gamma^{a_1 a_2})^\alpha{}_\beta \;
    e_{a_1 a_2}
    +
    \beta_3(s)
    (\Gamma^{a_1 \cdots a_5})^\alpha{}_{\beta} \;
    e_{a_1 \cdots a_5}
  \Big)\;,
\end{equation}
for analytic functions $\alpha_i, \beta_j$ of the parameter $s \in \mathbb{R} \setminus \{0\}$ with the following zeros
\begin{equation}
  \label{CoefficientZerosInDecomposedCField}
  \begin{array}{lcl}
    \alpha_0(s) \neq 0
    \\
    \alpha_1(s) = 0 &\Leftrightarrow& s = -3
\\
        \alpha_2(s) = 0 &\Leftrightarrow& s = -6
    \\
    \alpha_3(s) = 0 &\Leftrightarrow& s = -6
    \\
    \alpha_4(s) = 0 &\Leftrightarrow& s = -6
    \\
    \beta_1(s) = 0 &\Leftrightarrow& s = -3/2
    \\
    \beta_2(s) = 0 &\Leftrightarrow& s = -3
    \\
    \beta_3(s) = 0 &\Leftrightarrow& s = -6.
  \end{array}
\end{equation}
\end{prop}

\begin{example}[Special parameter-value for super-exceptional geometry]
  \label{ExceptionalSuperspacetimeAtSEqualsMinusOne}
    If the parameter $s$ \eqref{s}
    in Def. \ref{ExceptionalTangentSuperSpacetime}
    and Prop. \ref{TransgressionElementForM2Cocycle}
    takes the value
    $$
      s = - 3
    $$
    the transgression element
    in Prop. \ref{TransgressionElementForM2Cocycle}
    has the property, from \eqref{CoefficientZerosInDecomposedCField}, that
    up to terms proportional to the
    5-index tensor $e_{\alpha_1 \cdots \alpha_5}$,
    its only dependence on
    $d B_{a_1 a_2} := \alpha_0(s) e_{a_1 a_2}$
    is through the leading term
    $H := (d B_{a_2 a_3}) \wedge d x^{a_2} \wedge dx^{a_3}$;
    concretely:
    \begin{equation}
      \label{HwithsMinusThree}
      H_{\mathrm{ex}_{(s=-3)}}
      \;=\;
      \underset{
       H
      }{
        \underbrace{
          (d B_{a_2 a_3}) \wedge d x^{a_2} \wedge dx^{a_3}
        }
      }
      -
      \tfrac{1}{2}
      \beta_1(-3)
      \cdot
      \overline{\eta}\;\Gamma_a \psi \wedge e^a
      \;+\;
      \mathcal{O}\big( \{e_{a_1 \cdots a_5}\} \big)
      \,,
    \end{equation}
    where the last term $\mathcal{O}\big( \{e_{a_1 \cdots a_5}\} \big)$
    denotes summands that vanish when the 5-index generators
    $e_{a_1 \cdots a_5}$ are set to zero.

We will see below in
Prop. \ref{ExceptionalPreimageOfPerrySchwarzLagrangian}
\hyperlink{SuperExceptionalPSLagrangianForsMinusThree}{(ii)}.
that in super-exceptional M-geometry
at value $s = -3$
the
bosonic Perry-Schwarz Lagrangian appears naturally.
\end{example}

The point of the super-exceptional M-theory spacetime
from Def. \ref{ExceptionalTangentSuperSpacetime}
is that it is a super-manifold
(via Remark \ref{SuperExceptionalSpacetimeAsManifold})
which approximates the universal
super 3-stack
$\widehat{ \mathbb{T}^{10,1\vert \mathbf{32}} }$
classified by the super M2-brane cocycle \eqref{TheMBraneCocycles}.
We record this phenomenon (which can be traced back to the
``hidden supergroup of 11d supergravity'' in \cite{DF}):
\begin{lemma}[Super-exceptional M-spacetime
from M2-brane super-2-gerbe {\cite[4.6]{FSS18}}]
\label{Decomposition}
We have a commuting diagram of the form
\begin{equation}
  \label{DecompositionMap}
  \raisebox{10pt}{
  \xymatrix@R=.7em{
    &&
    \\
    \big(
      \mathbb{T}^{10,1\vert \mathbf{32}}
    \big)_{\mathrm{ex}_s}
    \ar[rr]_-{ \mathrm{comp} }
    ^{
      \mbox{
        \tiny
        $
        \begin{array}{rcl}
          \psi^\alpha
          &\!\!\! \mapsfrom \!\!\!&
          \psi^\alpha
          \\
          e^a
          &\!\!\! \mapsfrom \!\!\!&
          e^a
          \\
          H_{\mathrm{ex}_s}
          &\!\!\! \mapsfrom \!\!\!&
          h_3
        \end{array}
        $
      }
    }
    \ar[ddr]_{ \pi_{\mathrm{ex}_s} }
    &&
    \mathfrak{m}2\mathfrak{brane}
    \ar[ddl]^{ \mathrm{\pi} }
    \\
    \\
    &
    \mathbb{R}^{10,1\vert \mathbf{32}}
  }
  }
\end{equation}
mapping the super-exceptional spacetime from Def. \ref{ExceptionalTangentSuperSpacetime}
to the homotopy fiber of the M2-brane cocycle from \eqref{TheMBraneCocycles}, such that the defining degree-3 generator
$h_3$ on the right is pulled back to the transgression element
$H_{\mathrm{ex}_s}$ from Prop. \ref{TransgressionElementForM2Cocycle}.
\end{lemma}
Hence we say:
\begin{defn}[Super-exceptional M5-brane cocycle]
\label{ExceptionalM5SuperCocycle}
The \emph{super-exceptional M5-brane cocycle},
to be denoted $\mathbf{dL}^{\!\!\mathrm{WZ}}_{\mathrm{ex}_s}$,
is the pullback of the super M5-brane cocycle \eqref{TheMBraneCocycles}
to the super-exceptional M-theory spacetime (Def. \ref{ExceptionalTangentSuperSpacetime})
along the decomposition morphism \eqref{DecompositionMap}:
\begin{equation}
  \label{SuperExceptionalM5Cocycle}
  \begin{aligned}
  \mathbf{dL}^{\!\!\mathrm{WZ}}_{\mathrm{ex}_s}
  & \coloneqq
  \mathrm{comp}^\ast
  \big(
    \tfrac{1}{2} h_3
    \;\wedge\;
    \overset{
      = d h_3
    }{
      \overbrace{
        \pi^\ast \mu_{{}_{\rm M2}}
      }
    }
    +
    \pi^\ast \mu_{{}_{\rm M5}}
  \big)
  \\
  & =
    \tfrac{1}{2} H_{\mathrm{ex}_s}
    \;\wedge\;
    (\pi_{\mathrm{ex}_s})^\ast \mu_{{}_{\rm M2}}
    +
    (\pi_{\mathrm{ex}_s})^\ast \mu_{{}_{\rm M5}}
  \\
  & =
    \tfrac{1}{2}
    H_{\mathrm{ex}_s}
    \;\wedge\;
    d H_{\mathrm{ex}_s}
    +
    \tfrac{1}{5!}
    (
      \overline{\psi}\;
      \Gamma_{a_1 \cdots a_5}
      \psi
    )
    \wedge
    e^{a_1} \wedge \cdots e^{a_5}
    \phantom{AA}
    \in\;
    \mathrm{CE}
    \Big(
      \big(
        \mathbb{R}^{10,1\vert \mathbf{32}}
      \big)_{\mathrm{ex}_s}
    \Big)
    \,.
  \end{aligned}
\end{equation}
\end{defn}

We close by commenting on the role and meaning of these
constructions:

\begin{remark}[Super-exceptional spacetime as target for M5-brane sigma-models]
\label{SuperExceptionalTargetForM5}
\item {\bf (i)}
The super-exceptional spacetime (Def. \ref{ExceptionalTangentSuperSpacetime}),
when regarded as a super-group of super-translations along itself, is what \cite[Sec. 6]{DF}
called the ``hidden supergroup of 11-dimensional supergravity''; motivated there by the
purely algebraic desideratum of finding a super Lie 1-algebra over which the would-be
M2-brane cocycle trivializes (Prop. \ref{TransgressionElementForM2Cocycle}),
as opposed to the super Lie 3-algebra $\mathfrak{m}2\mathfrak{brane}$ \cite{FSS13}
on which it does so universally.
\item {\bf (ii)} However, the actual meaning or role of this ``hidden supergroup''
had remained open; as it was not used for the
re-derivation of $D=11$ supergravity in \cite{DF}, but discussed
as an afterthought. In particular, the meaning or role of the
extra fermion field required by super-exceptional spacetime
had remained open \cite[p. 3]{ADR16}
and the relation of the bosonic fields to exceptional M-geometry
had remained unnoticed except in \cite{Vaula07}
and then more recently in \cite[p. 6]{Bandos17}\cite{FSS18}\cite{SS18}.

\item {\bf (iii)} But in \cite[Prop. 4.31]{FSS19b}\cite[Prop. 4.4]{FSS19c},
we found that the M5-brane sigma model in a given C-field background is
characterized as making the outermost square of the following
diagram homotopy-commute, here now displayed for flat super-spacetimes
instead of curved topological spacetimes:
$$
  \xymatrix@C=39pt@R=2em{
    \mathclap{
    \mbox{
      \tiny
      \color{blue}
      \begin{tabular}{c}
        M5 worldvolume
        \\
        super manifold
      \end{tabular}
    }
    }
    &
    \mathclap{
    \mbox{
      \tiny
      \color{blue}
      \begin{tabular}{c}
        super-exceptional
        spacetime
        \\
        $\sim$
        \\
        universal super manifold
        \\
        classifying M5-brane fields
      \end{tabular}
    }
    }
    &&
    \mbox{
      \tiny
      \color{blue}
      \begin{tabular}{c}
        extended super
        spacetime
        \\
        =
        \\
        universal super 3-stack
        \\
        classifying M5-brane fields
      \end{tabular}
    }
    &&
    \mathclap{
    \mbox{
      \tiny
      \color{blue}
      \begin{tabular}{c}
        spherical coefficients
        \\
        for Cohomotopy theory
      \end{tabular}
    }
    }
    \\
    \\
    \Sigma
    \ar@/^3pc/[rrrrr]|-{
      H_3 \wedge G_4 + 2 G_7
    }
    \ar@{-->}[r]
    \ar@/_1pc/[ddrrr]
    &
    \big(
      \mathbb{T}^{10,1\vert \mathbf{32}}
    \big)_{\mathrm{ex}_s}
    \ar@/^1.6pc/[rrrr]|-{
      H_{\mathrm{ex}_s}\,\wedge\, \pi^\ast_{\mathrm{ex}_s} \mu_{{}_{\rm M2}}
      + \pi^\ast_{\mathrm{ex}_s} 2 \mu_{{}_{\rm M5}}
    }
    \ar[ddrr]_-{ \pi_{\mathrm{ex}_s} }
    \ar[rr]|-{\;\mathrm{comp}\;}
    &&
    \widehat{
      \mathbb{T}^{10,1\vert \mathbf{32}}
    }
    \ar[dd]^>>>>>>{\ }="t"_-{ \pi }
    \ar[rr]|{\;
      h_3
      \,\wedge\,
      \pi^\ast \mu_{{}_{\rm M2}}
      +
      2
      \pi^\ast \mu_{{}_{\rm M5}}
    }_>>>>>>{\ }="s"
    \ar@{}[ddrr]|-{ \mbox{ \tiny (pb) } }
    &&
    S^7_{\mathbb{R}}
    \ar[dd]^-{ h_{\mathbb{H}} }
    \\
    \\
    &
    &&
    \mathbb{T}^{ 10,1\vert \mathbf{32} }
    \ar[rr]_-{ (\mu_{{}_{\rm M2}}, 2 \mu_{{}_{\rm M5}}) }
    &&
    S^4_{\mathbb{R}}
    \ar@{=>}_{h_3} "s"; "t"
  }
$$
Any such outer square factors universally through the homotopy pullback
square shown on the right, which exhibits
the M2-brane extended super spacetime
$ \widehat{ \mathbb{T}^{10,1\vert \mathbf{32}} }$ \eqref{TheMBraneCocycles}
as the super moduli 3-stack classifying M5-brane sigma-model fields \cite{FSS15}.
But now with the (extended) worldvolume $\Sigma$
itself a super manifold, this factorization
through $\widehat{ \mathbb{T}^{10,1\vert \mathbf{32}} }$ is to be further factorized,
as shown by the dashed map,
through an actual super-manifold still classifying these fields.
This is the role of the super-exceptional spacetime
$\big(\mathbb{R}^{10,1\vert \mathbf{32}}\big)_{\mathrm{ex}_s}$ \cite[Sec. 4.6]{FSS18}: It is the actual super manifold which serves
as a stand-in for the classifying super space
$\widehat{ \mathbb{T}^{10,1\vert \mathbf{32}} }$
(which is not itself a super-manifold).

\item {\bf (iv)}  The key consequence of
$\big( \mathbb{T}^{10,1\vert \mathbf{32}}\big)_{\mathrm{ex}_s}$
being a super-manifold, is that the indecomposable degree-3
generator $h_3$ on $\widehat{\mathbb{T}^{10,1\vert \mathbf{32}}}$
\eqref{TheMBraneCocycles},
which has trivial contraction with any vector field,
pulls back to the decomposable
3-form $H_{\mathrm{ex}_s}$ on $\big( \mathbb{T}^{10,1\vert \mathbf{32}}\big)_{\mathrm{ex}_s}$ (by Lemma \ref{Decomposition}),
which, like any differential form on a super manifold,
in general has non-trivial contraction with vector fields.
Below in \cref{SuperExceptionalLagrangian}
it is such a non-trivial contraction of $H_{\mathrm{ex}_s}$
with a vector field on the super-exceptional spacetime,
which makes the Perry-Schwarz Lagrangian appear.
The analogous contraction with $h_3$ on
$\widehat{\mathbb{T}^{10,1\vert \mathbf{32}}}$ would
vanish, which is the reason why the embedding construction
of the M5-brane does not work with ewxtended super-spacetime,
but requires passing to super-exceptional spacetime.

\end{remark}

\begin{lemma}[{\cite[Prop. 4.26]{FSS18}}]
\label{ReflectionAutomorphismOnExceptionalTangentSuperSpacetime}
The canonical $\mathrm{Pin}^{+}(10,1)$ action
on
super-spacetime $\mathbb{R}^{10,1\vert \mathbf{32}}$ lifts to
super-exceptional spacetime
$\big(\mathbb{R}^{10,1\vert \mathbf{32}}\big)_{\mathrm{ex}_s}$
(Def. \ref{ExceptionalTangentSuperSpacetime}),
such that
$\mathrm{Spin}(10,1) \hookrightarrow \mathrm{Pin}^+(10,1)$
acts in the evident way,
while reflection $\rho$
along the $a_{r}$-axis acts dually as follows:
\footnote{Beware that this is saying that the  generator $e_{a_1 a_2}$ picks up a sign precisely if its indices
do \emph{not} take the value $a_r$.}
$$
  \rho^\ast
  \;:\;
  \left\{
 \begin{aligned}
    e^a
    & \longmapsto
    \left\{
      \begin{array}{rcl}
        - e^a &\vert& a = a_r
        \\
        e^a &\vert& \mbox{otherwise}
      \end{array}
    \right.
    \\
    e_{a_1 a_2}
    &
    \longmapsto
    \left\{
      \begin{array}{rcl}
        - e_{a_1 a_2} &\vert& a_1, a_2 \neq a_r
        \\
        e_{a_1 a_2} &\vert& \mbox{otherwise}
      \end{array}
    \right.
    \\
    e_{a_1 \cdots a_5}
    &
    \longmapsto
    \left\{
      \begin{array}{rcl}
        - e_{a_1 \cdots a_5} &\vert& \mbox{one of the} \; a_i = a_r
        \\
        e_{a_1 \cdots a_5} &\vert& \mbox{otherwise}
      \end{array}
    \right.
    \\
    \psi & \longmapsto \phantom{-}\pmb{\Gamma}_{a_r} \psi
    \\
    \eta & \longmapsto - \pmb{\Gamma}_{a_r} \eta
 \end{aligned}
 \right.
$$
\end{lemma}

\section{Super-exceptional $\mathrm{MK6}$- and $\tfrac{1}{2}\mathrm{M5}$-geometry}
\label{SuperExceptionalM5Locus}

We now consider the exceptionalization (Def. \ref{HalfM5LocusAndItsExceptionalTangentBundle} below)
of the local super-geometry of M-theory compactified on
$\mathbb{H}_{\sslash G_{\mathrm{ADE}}} \times S^1_{\sslash \mathbb{Z}_2}$
\cite{KSTY99}\cite{CabreraHananySperling19},
hence of Ho{\v r}ava-Witten heterotic M-theory on an ADE-orbifold, or, equivalently, of
non-perturbative type $\mathrm{I}^\prime$ string theory
with D6-branes intersecting the O8-plane in a $\tfrac{1}{2}\mathrm{NS}5$-brane \cite[6]{GKST01} \cite[p. 18]{ApruzziFazzi17}.
Therefore, following \cite[Ex. 2.2.7]{HSS18}, we will speak here of the
``$\tfrac{1}{2}\mathrm{M5}$ super-spacetime'', for emphasis of the
full 11-dimensional perspective; details are given in Remark \ref{TheHalfM5Spacetime}.

\medskip
It is traditionally understood that this compactification is one of two possible ways
of obtaining classes of $D = 6$, $\mathcal{N}= (1,0)$ superconformal field
theories from M-theory \cite[Sec. 6]{DHTV14}.
We make this mathematically concrete with Prop. \ref{LiftOfVectorField} below,
whose proof shows that the particular spinor structure of the
$\tfrac{1}{2}\mathrm{M5}$-locus (Def. \ref{SpinorProjection} below)
is what allows a lift of the spacetime isometry along the M-theory circle
to a symmetry also of the exceptionalized super-spacetime, including
the exceptionalized 3-flux density on the M5-brane.

\medskip
We use this to identify in Prop. \ref{ExceptionalPreimageOfPerrySchwarzLagrangian}
the exceptional pre-image of the bosonic 2-flux density $F$. Finally, we show in
Prop. \ref{calFexcDecomposed} that this induces the super 2-flux of super Yang-Mills
theory thereby, in particular, identifying the extra fermion field on exceptional
super-spacetime with the M-theoretic avatar of the gaugino field in
10d heterotic string theory.

\medskip

\begin{defn}[Spinor projection of $\tfrac{1}{2}\mathrm{M5}$-locus]
\label{SpinorProjection}
For $\mathbf{32}$ the Pin-representation from Def. \ref{CliffordAlgebraRepresentation},
let $P_{{}_{\mathbf{8}}}$ denote the linear projection
$$
  \xymatrix@R=2pt@C=4em{
    \mathrm{Pin}^+(10,1)
    \ar@{<-^{)}}[r]
    & ~
    \mathrm{Pin}^+(6,1)
    \ar@{<-^{)}}[r]
    & ~
    \mathrm{Pin}^+(5,1)
    \\
    \mathbf{32}
    \ar@/_2pc/[ddddrr]_-{ P_{{}_{\mathbf{8}}} }
    \ar[ddr]_{\;\;\;\;\;\; P_{{}_{\mathbf{16}}} }
    \ar[r]^-{ \simeq_{{}_{\mathbb{R}}} }
    &
    2 \cdot \mathbf{16}
    \ar[dd]
    \ar[r]^-{ \simeq_{{}_{\mathbb{R}}} }
    &
    4 \cdot \mathbf{8}
    \ar[dd]
    \\
    {\phantom{a}}
    \\
    &
    {\phantom{1 \cdot\;}} \mathbf{16}
    \ar[r]^-{ \simeq_{{}_{\mathbb{R}}} }
    \ar[ddr]
    &
    2 \cdot \mathbf{8}
    \ar[dd]
    \\
    {\phantom{a}}
    \\
    & &
    {\phantom{2 \cdot\;}} \mathbf{8}
  }
$$
onto the joint fixed locus of $\pmb{\Gamma}_5$ and $\pmb{\Gamma}_{6789}$.
Hence this is characterized by the following
\begin{equation}
  \label{FermionProjections}
  \begin{array}{rclclcl}
    \mathrm{MK6}
    &&
    \pmb{\Gamma}_{\phantom{5}6789} (P_{{}_{\mathbf{16}}} \psi)
      & = &
      (P_{{}_{\mathbf{16}}} \psi)
    &&
    \mathcal{O}6
    \\
    \tfrac{1}{2}\mathrm{M5}
    &&
    \pmb{\Gamma}_{\phantom{5}6789} (P_{{}_{\mathbf{8}}} \psi)
      & = &
      (P_{{}_{\mathbf{8}}} \psi)
    &&
    \mathcal{I}5
    \\
    &&
    \pmb{\Gamma}_{5\phantom{6789}} (P_{{}_{\mathbf{8}}} \psi)
      & = &
      (P_{{}_{\mathbf{8}}} \psi)
    &
    \\
    &&
    \pmb{\Gamma}_{56789} (P_{{}_{\mathbf{8}}} \psi)
      & = & (P_{{}_{\mathbf{8}}} \psi)
    &&
  \end{array}
\end{equation}
\end{defn}

\begin{remark}[The $\tfrac{1}{2}\mathrm{M5}$-locus]
\label{TheHalfM5Spacetime}
The spinor projections in Def. \ref{SpinorProjection} correspond to
the $\tfrac{1}{2}\mathrm{M5} = \mathrm{MK6} \cap \mathrm{MO9}$
super spacetime \cite[6]{GKST01} \cite[6]{DHTV14} \cite[p. 18 ]{ApruzziFazzi17};
see \cite[Ex. 2.2.7]{HSS18} in the present context:
\begin{equation}
 \label{HalfM5BraneSetup}
 \raisebox{40pt}{
 \xymatrix@C=4pt@R=9pt{
   &
   &
   &
   &
   &
   &
   S^1_B
   &
   \!
   S^1_{\mathrm{HW}}\!
   \ar@(ul,ur)^{
     \mathclap{
       G_{\mathrm{HW}} = \langle {\pmb\Gamma}_{5} \rangle
     }
   }
   \ar@{}[rrrrr]_-{ \!\!\!\!\!\!\!\!\!\!\!\!\!\!\!\!\underset{\overbrace{\phantom{------.}}}{} }
   \ar@{}[rrrrr]|-{ \!\!\!\!\!\!\!\!\!\!\!\!\!\!\!\!\mbox{ $\mathbb{H}$ } }
   &
   &
   &
   &
   &
   \\
   \fbox{ I' }
   &
   0
   &
   1
   &
   2
   &
   3
   &
   4
   &
   5'
   &
   5
   &
   6
   &
   7
   \ar@<+18pt>@/^2.5pc/[r]^{ G_{\mathrm{ADE}} = \langle {\pmb\Gamma}_{6789}  \rangle }
   &
   8
   &
   9
   &
   \fbox{HET}
   \\
   \mathrm{MO9}
   &
   \mbox{---}
   &
   \mbox{---}
   &
   \mbox{---}
   &
   \mbox{---}
   &
   \mbox{---}
   &
   \mbox{---}
   &
   &
   \mbox{---}
   &
   \mbox{---}
   &
   \mbox{---}
   &
   \mbox{---}
   \\
   \mathrm{MK6}
   &
   \mbox{---}
   &
   \mbox{---}
   &
   \mbox{---}
   &
   \mbox{---}
   &
   \mbox{---}
   &
   \mbox{---}
   &
   \mbox{---}
   &
   &
   &
   &
   &
   \mathcal{O}6
   \\
   \tfrac{1}{2}\mathrm{NS5}
   &
   \mbox{---}
   &
   \mbox{---}
   &
   \mbox{---}
   &
   \mbox{---}
   &
   \mbox{---}
   &
   \mbox{---}
   &
   &
   &
   &
   &
   &
   \mathcal{I}5
   \\
   &
   \ar@{}[rrrrr]_-{
     \underset{
       \mbox{ M5 worldvolume }
     }{
       \underbrace{ \phantom{--------------.} }
     }
   }
   \mbox{---}
   &
   \mbox{---}
   &
   \mbox{---}
   &
   \mbox{---}
   &
   \mbox{---}
   &
   \ar@{}[r]|-{\ddots}
   &
   &
   &
   &
   &
   &
 }
 }
 \mbox{
   \hspace{-.4cm}
   \scalebox{.76}{
   \raisebox{-96pt}{
   \includegraphics[width=.5\textwidth]{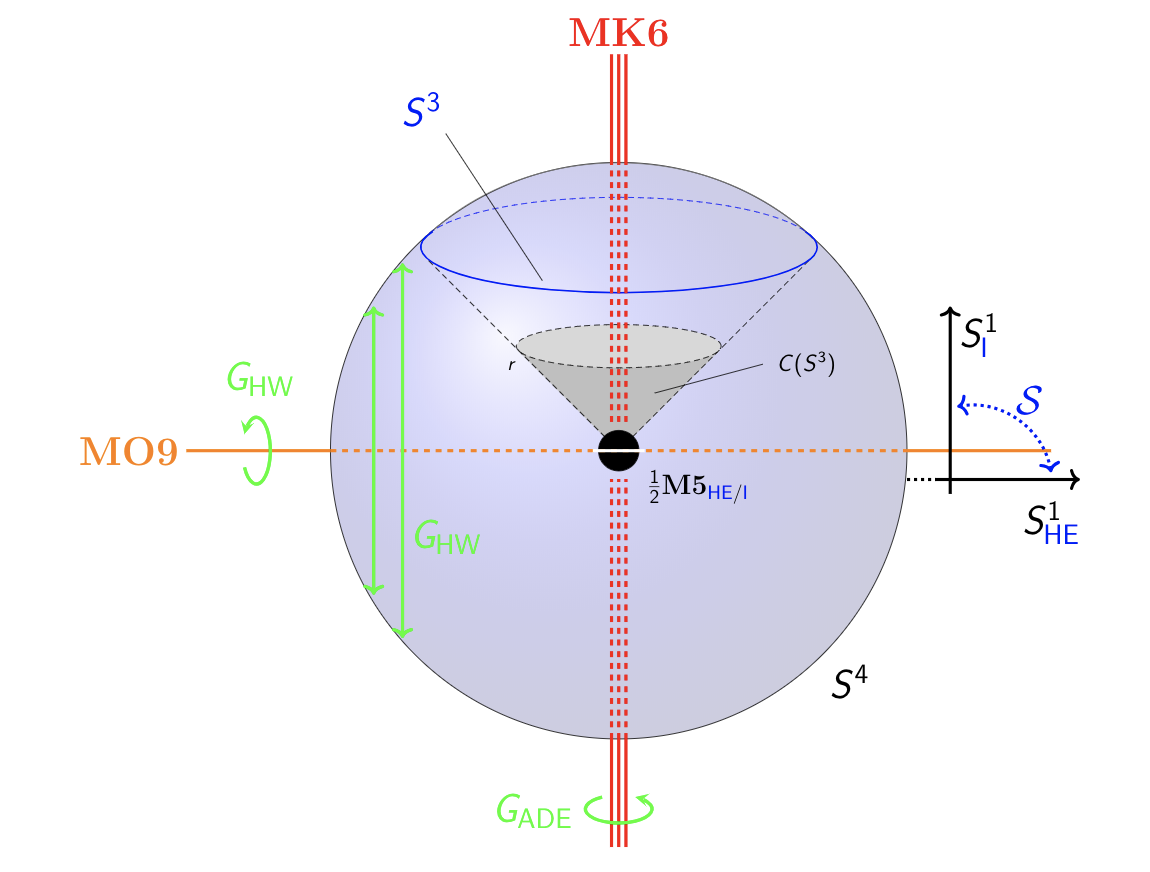}
   }}
 }
\end{equation}

\begin{enumerate}
\vspace{-3mm}
 \item {\bf Black branes. }
The appropriate names of the fixed loci depend on the duality frame:
if the circle on which $G_{\mathrm{A}}$ acts is taken for
M/IIA duality, then the result is type I' and the fixed loci are named as shown on the left
(see also e.g. \cite[around Fig. 2]{HananyZaffaroni97}\cite{HKLTY15}).
On the other hand, if the circle on which $G_{\mathrm{HW}}$  acts is taken
for M/IIA duality, the result is HET on a $G_{\mathrm{ADE}}$-orbifold,
and labels as used in \cite[p. 8]{GKST01} are as shown on the right.
Moreover, under T-duality along $S^1_B$ the I'-perspective
turns into a  configuration of a IIB  NS5-brane parallel to an O9-plane
(see, e.g., \cite{HananyZaffaroni99}).

\vspace{-2mm}
\item {\bf  Fundamental (sigma-model) brane. } The last line in the above diagram
indicates the extension of the actual M5-brane worldvolume for which we are to
construct a Lagrangian density. This is not required to be fixed by $\pmb{\Gamma}_5$
and hence may stretch along the $5'$-$5$-plane/torus ``at an angle'',
partially wrapping around the $S^1_{\mathrm{HW}}$, partially
running parallel to it \cite[p. 9]{Witten97},
corresponding to D4/NS5 bound states in IIA \cite[Sec. 2]{Witten97}
and $(p,q)$-5brane webs in IIB \cite{AharonyHanany97},
or rather their further intersection with an O-plane
\cite{HananyZaffaroni99} (see \cite[around Fig. 10]{HKLTY15}).
\end{enumerate}

\end{remark}

The following definition generalizes the super $\mathrm{MK6}$-spacetime
regarded as a $\mathbb{Z}_2$-fixed locus inside $D = 11$, $\mathcal{N} = 1$
super Minkowski spacetime \cite[Prop. 4.7, Ex. 2.2.5]{FSS18} to
a $\mathbb{Z}_2$-fixed locus inside the super-exceptional M-theory spacetime:

\begin{defn}[Super-exceptional $\mathrm{MK6}$-spacetime]
  \label{SuperExceptionalMK6Spacetime}
  For $s \in \mathbb{R} \setminus \{0\}$,  we say that
  \emph{super-exceptional $\mathrm{MK6}$-spacetime}
  $\big(\mathbb{R}^{6,1\vert \mathbf{16}}\big)_{\mathrm{ex}_s}$
  is the fixed locus inside the  super-exceptional M-theory spacetime
  (Def. \ref{ExceptionalTangentSuperSpacetime})
  of the $\mathbb{Z}_2 \hookrightarrow \mathrm{Pin}^+(10,1)$-action
  of Lemma \ref{ReflectionAutomorphismOnExceptionalTangentSuperSpacetime}
  generated from the element $\pmb{\Gamma}_{6789} \in \mathrm{Pin}^+(10,1)$ as in Prop. \ref{SpinorProjection}.
  Hence, from the super dgc-algebra \eqref{FermionicExtensionOfExceptionalTangentSuperspacetime}
  and Lemma \ref{ReflectionAutomorphismOnExceptionalTangentSuperSpacetime},
  this is the rational super space given dually by the following
  super dgc-algebra:
\begin{equation}
  \label{CESuperExceptionalMK6}
  \hspace{-7.5cm}
  \raisebox{90pt}{
  \xymatrix{
    \big(
      \mathbb{T}^{6,1\vert \mathbf{16}}
    \big)_{\mathrm{ex}_s}
    \ar[dd]_{ \pi^{\mathrm{MK6}}_{\mathrm{ex}_s} }
    &&
  \mathbb{R}\left[\!\!
    {\begin{array}{l}
      \underset{
        \mathrm{deg} = (1,\mathrm{even})
       }{
        \underbrace{
          \big\{
            e^a
          \}_{ a \in \{0,1,2,3,4,5',5\} }
        }
      }
      \\
      \underset{\mathrm{deg} = (1,\mathrm{even}) }{
        \underbrace{
        \big\{
          e_{a_1 a_2}
        \big\}_{
          \!\!\!\!\!\!
            \mbox{
              \tiny
              $
              \begin{array}{l}
                a_i \in \{0,1,2,3,4,5',5,6,7,8,9\}
                \\
                \mbox{with 0 or 2 of $a_i$s in}
                \{0,1,2,3,4,5',5\}
              \end{array}
              $
            }
        }
        }
      }
      \\
      \underset{ \mathrm{deg} = (1,\mathrm{even}) }{
        \underbrace{
          \big\{
            e_{a_1 \cdots a_5}
          \big\}_{
            \!\!\!\!\!\!
            \mbox{
              \tiny
              $
              \begin{array}{l}
                a_i \in \{0,1,2,3,4,5',5,6,7,8,9\}
                \\
                \mbox{with 1 or 3 of $a_i$s in}
                \{0,1,2,3,4,5',5\}
              \end{array}
              $
            }
          }
        }
      }
      \\
      \underset{ \mathrm{deg} = (1,\mathrm{odd}) }{
        \underbrace{
          \big\{
            (P_{{}_{\mathbf{16}}}\psi)^\alpha
          \big\}_{0 \leq \alpha \leq 16 }
        }
      }
      \\
      \underset{ (1,\mathrm{odd}) }{
        \underbrace{
          \big\{
            (P_{{}_{\mathbf{16}}}\eta)^\alpha
          \big\}_{ 0 \leq \alpha \leq 16 }
         }
      }
      \end{array}}
    \!\!\right]
    \Big/
    \mathrlap{
  \left(\!\!\!
    {\begin{array}{lcl}
      d\,(P_{{}_{\mathbf{16}}}\psi)^\alpha
      & \!\!\!\!\!\!\!\!\!  = & \!\!\!\!\!\!
      0,
      \\
      d\,e^a
      & \!\!\!\!\!\!\!\!\!  = & \!\!\!\!\!\!
        (\overline{P_{{}_{\mathbf{16}}}\psi})
          {\Gamma}^a
        (P_{{}_{\mathbf{16}}}\psi),
      \\
      d \, e_{a_1 a_2}
        & \!\!\!\!\!\!\!\!\!  = & \!\!\!\!\!\!
      \tfrac{i}{2}
      (\overline{P_{{}_{\mathbf{16}}}\psi})
        {\Gamma}_{a_1 a_2}
      (P_{{}_{\mathbf{16}}}\psi),
      \\
      d \, e_{a_1 \cdots a_5}
        & \!\!\!\!\!\!\!\!\!  = & \!\!\!\!\!\!
      \tfrac{1}{5!}
      (\overline{P_{{}_{\mathbf{16}}}\psi})
        {\Gamma}_{a_1 \cdots a_5}
      (P_{{}_{\mathbf{16}}}\psi),
      \\
      d \,\eta
       & \!\!\!\!\!\!\!\!\!  = & \!\!\!\!\!\!
       (s+1)
       e^a
         \wedge
       {\pmb{\Gamma}}_a
       (P_{{}_{\mathbf{16}}} \psi)
       \\
       & &
       +
       e_{a_1 a_2}
         \wedge
       {\pmb{\Gamma}}^{a_1 a_2}
       (P_{{}_{\mathbf{16}}} \psi)
       \\
       & &
       +
       (1 + \tfrac{s}{6})
       e_{a_1 \cdots a_5}
         \wedge
       {\pmb{\Gamma}}^{a_1 \cdots a_5}
       (P_{{}_{\mathbf{16}}} \psi)
    \end{array}}
  \!\!\! \right)
  }
  \ar@{<-}[dd]^-{
    \mbox{
      \tiny
      $
      \begin{array}{ccc}
        \psi^\alpha & e^a
        \\
        \mapsup & \mapsup
        \\
        \psi^\alpha & e^a
      \end{array}
      $
    }
  }
  \\
  \\
  \mathbb{T}^{6,1\vert \mathbf{16}}
  &&
  \mathbb{R}\left[
    {\begin{array}{l}
      \underset{
        \mathrm{deg} = (1,\mathrm{even})
       }{
        \underbrace{
          \big\{
            e^a
          \}_{ a \in \{ 0,1,2,3,4,5',5 \} }
        }
      }
      \\
      \underset{ \mathrm{deg} = (1,\mathrm{odd}) }{
        \underbrace{
          \big\{
            (P_{{}_{\mathbf{16}}} \psi)^\alpha
          \big\}_{0 \leq \alpha \leq 16 }
        }
      }
      \end{array}}
    \right]
    \Big/
    \mathrlap{
  \left(\!\!
    {\begin{array}{lcl}
      d \,(P_{{}_{\mathbf{16}}}\psi)^\alpha
        & \!\!\!\!\!\!\!  = & \!\!\!\!\!\!
      0,
      \\
      d e^a & \!\!\!\!\!\!\!  = & \!\!\!\!\!\!
        (\overline{P_{{}_{\mathbf{16}}}\psi}) {\Gamma}^a (P_{{}_{\mathbf{16}}}\psi)
    \end{array}}
  \!\! \right)
  }
  }
  }
\end{equation}
\end{defn}

We now incorporate spacetime isometries.

\begin{prop}[Lift of isometries to super-exceptional $\mathrm{MK6}$-spacetime]
\label{LiftOfVectorField}
 For all parameters $s$ from \eqref{s}, distinct from -6, i.e., for
\begin{equation}
  \label{sNotMinusSix}
  s \in \mathbb{R} \setminus \{0,-6\},
\end{equation}
and for all $a \in \{1,2,3,4,5',5\}$,
the infinitesimal superspace symmetry $v_a$
\begin{equation}
  \label{TheVectorField}
  v_a
    \in
  \Gamma\big(
    T \mathbb{R}^{6,1\vert \mathbf{16}}
  \big)
  \,,
  \phantom{AA}
  \iota_{v_a}
  \;\colon\;
  \left\{
  \begin{array}{rcl}
    e^b &\mapsto& \delta^a_b
    \\
    (P_{{}_{\mathbf{16}}} \psi)^\alpha &\mapsto& 0
  \end{array}
  \right.
  \,,
  \phantom{AA}
  \mathcal{L}_{v_a}
  \coloneqq
  \big[ d, \iota_{v_a}\big]
  = 0
\end{equation}
of the super $\mathrm{MK6}$-spacetime
lifts to the super-exceptional $\mathrm{MK6}$-spacetime
(Def. \ref{SuperExceptionalMK6Spacetime})
as
\begin{equation}
  \label{LiftedVectorField}
  v_a^{\mathrm{ex}_s}
  \;\coloneqq\;
  v_a
    -
  \tfrac
  {1+s}
  {1 + \frac{s}{6}}
  \tfrac{1}{5!}
  v^{a6789}
  +
  \chi^\alpha v^\eta_{\alpha}
  \in
  \Gamma
  \big(
    T
    \mathbb{R}^{5,1\vert \mathbf{8}}
  \big)
  \,,
  \phantom{AAA}
  \iota_{v_a^{\mathrm{ex}_s}}
  \;\colon\;
  \left\{
  \begin{array}{lcl}
    e^b &\mapsto& \delta^a_b
    \\
    e_{a_1 a_2} &\mapsto&  0
    \\
    e_{a_1 \cdots a_5}
      &\mapsto&
      -
      \frac
      {1+s}
      {1 + \frac{s}{6}}
      \epsilon^{a_1 \cdots a_5}_{a 6 7 8 9}
    \\
    (P \psi)^\alpha &\mapsto& 0
    \\
    (P\eta)^\alpha &\mapsto& (P\chi)^\alpha
  \end{array}
  \right.
\end{equation}
for any odd-graded constants
\begin{equation}
  \label{OddParameters}
  \underset{
    \mathclap{ \mathrm{deg} = (0,\mathrm{odd}) }
  }{
    \underbrace{ (P\chi)^\alpha }
  }
\end{equation}
in that
\begin{equation}
  \label{LieDerivativeAlongExceptionalLiftOfv5VanishesOnHalfM5}
  \mathcal{L}_{v_a^{\mathrm{ex}_s}}
  \coloneqq
  \big[ d_{\mathrm{ex}_s} , \iota_{v_a^{\mathrm{ex}_s}} \big]
  \;=\;
  0
  \phantom{AAAA}
  \mbox{
    on
    $\big(
      \mathbb{R}^{6,1\vert \mathbf{16}}
    \big)_{\mathrm{ex}_s}$.
  }
\end{equation}
\end{prop}
\begin{proof}
  Direct inspection of the defining
  differential relations in  \eqref{CESuperExceptionalMK6}
  shows that the derivation
  $\mathcal{L}_{v_a^{\mathrm{ex}_s}}$ evidently vanishes on all generators
  except possibly on $\eta$.
  The action on $\eta$ is computed as follows:
  \begin{equation}
    \label{MakeSuperExceptionalIsometryWork}
    \begin{aligned}
      \mathcal{L}_{v_a^{\mathrm{ex}_s}}
      (P_{{}_{\mathbf{16}}} \eta)
      & =
      \iota_{v_a^{\mathrm{ex}_s}} d (P\eta)
      +
      \underset{
         = 0
      }{
        \underbrace{
          d
          \;\,
          \overset{
             = (P \chi)^\alpha \; \mathrlap{= \mathrm{const}}
          }{
          \overbrace{
            \iota_{v_5^{\mathrm{ex}_s}} (P \eta)
          }
          }
        }
      }
      \\
      & =
      \iota_{ v_a^{\mathrm{ex}_s} }
      \big(
        (s+1)
        e^{a_1} {\pmb{\Gamma}}_{a_1} (P\psi)
        +
        e_{a_1 a_2} {\pmb{\Gamma}}^{a_1 a_2} (P\psi)
        +
        (1 + \tfrac{s}{6})
        e_{a_1 \cdots a_5}
        {\pmb{\Gamma}}^{a_1 \cdots a_5}
        (P\psi)
      \big)
      \\
      & =
        (s+1)
        {\pmb{\Gamma}}_a (P_{{}_{\mathbf{16}}}\psi)
        -
        (1 + \tfrac{s}{6})
        \tfrac{1 + s}{ 1 + \tfrac{s}{6} }\;
        \pmb{\Gamma}_a
        \underset{
          = (P_{{}_{\mathbf{16}}} \psi)
        }{
        \underbrace{
          {\pmb{\;\Gamma}}^{6 7 8 9} (P_{{}_{\mathbf{16}}}\psi)
        }}
      \\
      & = 0\;,
    \end{aligned}
  \end{equation}
  where under the brace we used the defining property
  \eqref{FermionProjections}
  of the spinor projections of Def. \ref{SpinorProjection}.
\end{proof}

With the super-exceptional lift of the circle isometry given, we have the
corresponding super-exceptional version of the projection \eqref{HorizontalProjection}
onto horizontal differential forms.
\begin{defn}[Super-exceptional horizontal projection]
  \label{SuperExceptionalHorizontalProjection}
  We say that
  \emph{projection onto the super-exceptional horizontal component}
  is the operation
  \begin{equation}
    (-)^{\mathrm{hor}_{\mathrm{ex}_s}}
    \;\coloneqq\;
    \tfrac{1}{2}
    \big(
      \mathrm{id}
      -
      e^5
      \wedge
      \iota_{v_5^{\mathrm{ex}_s}}
    \big)
    \;\colon\;
    \mathrm{CE}
    \big(
      \mathbb{R}^{5,1\vert \mathbf{8}}
      \times
      \mathbb{R}^1
    \big)_{\mathrm{ex}_s}
    \xymatrix{\ar[r]&}
    \mathrm{CE}
    \big(
      \mathbb{R}^{5,1\vert \mathbf{8}}
      \times
      \mathbb{R}^1
    \big)_{\mathrm{ex}_s}
  \end{equation}
  on the CE-algebra (FDA) of the super-exceptional $\tfrac{1}{2}\mathrm{M5}$-spacetime
  (Def. \ref{HalfM5LocusAndItsExceptionalTangentBundle}), where
    $\iota_{v_5}^{\mathrm{ex}_s}$ is the contraction \eqref{LiftedVectorField}
  from Prop. \ref{LiftOfVectorField}
\end{defn}

The next Definition \ref{HalfM5LocusAndItsExceptionalTangentBundle} formalizes
the  $\tfrac{1}{2}\mathrm{M5}$-locus (as in Remark \ref{TheHalfM5Spacetime})
inside the super-exceptional $\mathrm{MK6}$-spacetime (as formalized by Def.
\ref{SuperExceptionalMK6Spacetime}) without discarding the ambient MK6 spacetime,
but breaking its supersymmetry from $D = 7$, $\mathcal{N} = 1$, to $D = 6$,
$\mathcal{N} = (1,0)$. Hence the generators of the
$\mathrm{MK6}$ super dgc-algebra (``FDA'') are
all retained, but all spinors on the right of the differential
relations \eqref{CESuperExceptionalMK6} get projected
not just by $P_{{}_{\mathbf{16}}}$ but by $P_{{}_{\mathbf{8}}}$
(as in Def. \ref{SpinorProjection}):

\begin{defn}[Super-exceptional $\tfrac{1}{2}\mathrm{M5}$ spacetime]
\label{HalfM5LocusAndItsExceptionalTangentBundle}
For $s \in \mathbb{R} \setminus \{0, -6\}$, we say that the \emph{super-exceptional
$\tfrac{1}{2}\mathrm{M5}$-spacetime} is the rational super space given dually by
the same super dgc-algebra \eqref{CESuperExceptionalMK6} as that of the
super-exceptional $\mathrm{MK6}$ of Def. \ref{SuperExceptionalMK6Spacetime}, but
with spinor projections $P_{{}_{\mathbf{8}}}$
instead of just $P_{{}_{\mathbf{16}}}$
(Def. \ref{SpinorProjection})
on the right of the differential relations:
\begin{equation}
  \label{CESuperExceptionalMK6-part2}
  \hspace{-7.5cm}
  \raisebox{90pt}{
  \xymatrix{
    \big(
      \mathbb{T}^{5,1\vert \mathbf{8}}
    \big)_{\mathrm{ex}_s}
    \ar[dd]_{ \pi^{\sfrac{1}{2}\mathrm{M5}}_{\mathrm{ex}_s} }
    &&
  \mathbb{R}\left[
    {\begin{array}{l}
      \underset{
        \mathrm{deg} = (1,\mathrm{even})
       }{
        \underbrace{
          \big\{
            e^a
          \}_{ a \in \{0,1,2,3,4,5',5\} }
        }
      }
      \\
      \underset{\mathrm{deg} = (1,\mathrm{even}) }{
        \underbrace{
        \big\{
          e_{a_1 a_2}
        \big\}_{
          \!\!\!\!\!\!
            \mbox{
              \tiny
              $
              \begin{array}{l}
                a_i \in \{0,1,2,3,4,5',5,6,7,8,9\}
                \\
                \mbox{with 0 or 2 of $a_i$s in}
                \{0,1,2,3,4,5',5\}
              \end{array}
              $
            }
        }
        }
      }
      \\
      \underset{ \mathrm{deg} = (1,\mathrm{even}) }{
        \underbrace{
          \big\{
            e_{a_1 \cdots a_5}
          \big\}_{
            \!\!\!\!\!\!
            \mbox{
              \tiny
              $
              \begin{array}{l}
                a_i \in \{0,1,2,3,4,5',5,6,7,8,9\}
                \\
                \mbox{with 1 or 3 of $a_i$s in}
                \{0,1,2,3,4,5',5\}
              \end{array}
              $
            }
          }
        }
      }
      \\
      \underset{ \mathrm{deg} = (1,\mathrm{odd}) }{
        \underbrace{
          \big\{
            (P_{{}_{\mathbf{16}}}\psi)^\alpha
          \big\}_{0 \leq \alpha \leq 16 }
        }
      }
      \\
      \underset{ (1,\mathrm{odd}) }{
        \underbrace{
          \big\{
            (P_{{}_{\mathbf{16}}}\eta)^\alpha
          \big\}_{ 0 \leq \alpha \leq 16 }
         }
      }
      \end{array}}
    \right]
    \Big/
    \mathrlap{
  \left(\!\!\!
    {\begin{array}{lcl}
      d\,(P_{{}_{\mathbf{16}}}\psi)^\alpha
        & \!\!\!\!\!\!\!\!\!  = & \!\!\!\!\!\!
      0,
      \\
      d\,e^a
      & \!\!\!\!\!\!\!\!\!  = & \!\!\!\!\!\!
        (\overline{P_{{}_{\mathbf{8}}}\psi})
          {\Gamma}^a
        (P_{{}_{\mathbf{8}}}\psi),
      \\
      d \, e_{a_1 a_2}
        & \!\!\!\!\!\!\!\!\!  = & \!\!\!\!\!\!
      \tfrac{i}{2}
      (\overline{P_{{}_{\mathbf{8}}}\psi})
        {\Gamma}_{a_1 a_2}
      (P_{{}_{\mathbf{8}}}\psi),
      \\
      d \, e_{a_1 \cdots a_5}
       & \!\!\!\!\!\!\!\!\!  = & \!\!\!\!\!\!
      \tfrac{1}{5!}
      (\overline{P_{{}_{\mathbf{8}}}\psi})
        {\Gamma}_{a_1 \cdots a_5}
      (P_{{}_{\mathbf{8}}}\psi),
      \\
      d \,\eta
       & \!\!\!\!\!\!\!\!\!  = & \!\!\!\!\!\!
       (s+1)
       e^a
         \wedge
       {\pmb{\Gamma}}_a
       (P_{{}_{\mathbf{8}}} \psi)
       \\
       & &
       +
       e_{a_1 a_2}
         \wedge
       {\pmb{\Gamma}}^{a_1 a_2}
       (P_{{}_{\mathbf{8}}} \psi)
       \\
       & &
       +
       (1 + \tfrac{s}{6})
       e_{a_1 \cdots a_5}
         \wedge
       {\pmb{\Gamma}}^{a_1 \cdots a_5}
       (P_{{}_{\mathbf{8}}} \psi)
    \end{array}}
 \!\!\! \right)
  }
  \ar@{<-}[dd]^-{
    \mbox{
      \tiny
      $
      \begin{array}{ccc}
        \psi^\alpha & e^a
        \\
        \mapsup & \mapsup
        \\
        \psi^\alpha & e^a
      \end{array}
      $
    }
  }
  \\
  \\
  \mathbb{T}^{5,1\vert \mathbf{8}}
  &&
  \mathbb{R}\left[
    {\begin{array}{l}
      \underset{
        \mathrm{deg} = (1,\mathrm{even})
       }{
        \underbrace{
          \big\{
            e^a
          \}_{ a \in \{ 0,1,2,3,4,5',5 \} }
        }
      }
      \\
      \underset{ \mathrm{deg} = (1,\mathrm{odd}) }{
        \underbrace{
          \big\{
            (P_{{}_{\mathbf{8}}} \psi)^\alpha
          \big\}_{0 \leq \alpha \leq 8 }
        }
      }
      \end{array}}
    \right]
    \Big/
    \mathrlap{
  \left(\!\!\!
    {\begin{array}{lcl}
      d \,(P_{{}_{\mathbf{8}}}\psi)^\alpha
        & \!\!\!\!\!\!\!  = & \!\!\!\!\!\!
      0,
      \\
      d e^a
      & \!\!\!\!\!\!\!  = & \!\!\!\!\!\!
        (\overline{P_{{}_{\mathbf{8}}}\psi}) {\Gamma}^a (P_{{}_{\mathbf{8}}}\psi)
    \end{array}}
  \!\!\!  \right)
  }
  }
  }
\end{equation}
\end{defn}

Directly analogous to Def. \ref{HalfM5LocusAndItsExceptionalTangentBundle}
we may apply the heterotic spinor projection already on the
full super-exceptional M-theory spacetime:
\begin{defn}[Super-exceptional heterotic M-theory spacetime]
\label{SuperExceptionalMTheorySpacetime}
For $s \in \mathbb{R} \setminus \{0\}$, we say that the
$n = 11$
\emph{super-exceptional heterotic M-theory spacetime}
$$
  \xymatrix{
    \big(
      \mathbb{T}^{9,1\vert \mathbf{16}}
      \times
      \mathbb{R}^1
    \big)_{\mathrm{ex}_s}
    \ar@{^{(}->}[rrrr]^-{
      \mbox{
        \tiny
        $
        \begin{array}{lcl}
          (P_{{}_{\mathbf{16}}}\psi) & \!\!\!\!\!\! \mapsfrom  \!\!\!\!\!\!&
            \psi
          \\
          \\
          \eta & \!\!\!\!\!\! \mapsfrom  \!\!\!\!\!\!&
            \eta
          \\
          e^{a} & \!\!\!\!\!\! \mapsfrom  \!\!\!\!\!\!&
            e^a
          \\
          e_{a_1 a_2} & \!\!\!\!\!\! \mapsfrom  \!\!\!\!\!\!&
            e_{a_1 a_2}
          \\
          e_{a_1 \cdots a_5} & \!\!\!\!\!\! \mapsfrom  \!\!\!\!\!\!&
            e_{a_1 \cdots a_5}
        \end{array}
        $
      }
    }
    &&&&
    \big(
      \mathbb{T}^{10,1\vert \mathbf{32}}
    \big)_{\mathrm{ex}_s}
  }
$$
is the rational super space given dually by
the same super dgc-algebra \eqref{FermionicExtensionOfExceptionalTangentSuperspacetime}
as that of the
$n = 11$
super-exceptional M-theory spacetime
of Def. \ref{ExceptionalTangentSuperSpacetime}, but
with spinor projections $P_{{}_{\mathbf{16}}}$
(Def. \ref{SpinorProjection})
on the right of the differential relations.
\end{defn}

The following is a direct consequence of the above definitions:
\begin{lemma}[Super-exceptional embeddings]
\label{SuperExceptionalEmbeddings}
We have a diagram of consecutive super embeddings covered by
super-exceptional embeddings of

\item {\bf (i)}
the super-exceptional $\tfrac{1}{2}\mathrm{M5}$-spacetime (Def.
\ref{HalfM5LocusAndItsExceptionalTangentBundle}) inside
the $\mathrm{MK6}$-spacetime (Def. \ref{SuperExceptionalMK6Spacetime})
inside the $n=11$ super-exceptional M-theory spacetime
(Def. \ref{ExceptionalTangentSuperSpacetime}) as follows:
\begin{equation}
  \label{M5FixedLocusNormallyEnhanced}
  \hspace{-.45cm}
  \xymatrix@R=1.2em@C=2.75em{
    &
    {\phantom{
    \big(
      \mathbb{T}^{5,1\vert \mathbf{8}}
      \times
      \mathbb{T}^1
    \big)_{\mathrm{ex}_s}
    }}
    \ar@{<-|}[rrrr]^{
      \mbox{
        \tiny
        $
        \begin{array}{lcl}
          e^a
          &\!\!\!\!\!\!\!\!\!\!\mapsfrom\!\!\!\!\!\!\!\!\!\!&
          e^a
          \\
          e_{a_1 a_2}
          &\!\!\!\!\!\!\!\!\!\!\mapsfrom\!\!\!\!\!\!\!\!\!\!&
          e_{a_1 a_2}
          \\
          e_{a_1 \cdots a_5}
          &\!\!\!\!\!\!\!\!\!\!\mapsfrom\!\!\!\!\!\!\!\!\!\!&
          e_{a_1 \cdots a_5}
          \\
          (P_{{}_{\mathbf{8}}})^\alpha
          &\!\!\!\!\!\!\!\!\!\!\mapsfrom\!\!\!\!\!\!\!\!\!\!&
          (P_{{}_{\mathbf{16}}}\psi)^\alpha
       \end{array}
        $
      }
    }
    &&&&
    {\phantom{
    \big(
      \mathbb{T}^{6,1\vert \mathbf{16}}
    \big)_{\mathrm{ex}_s}
    }}
    \ar@{<-|}[rrrr]^{
      \mbox{
        \tiny
        $
        \begin{array}{lcl}
          \left.
          \begin{array}{lcl}
            e^a
              &\!\!\!\!\!\!\!\!\!\!\!\!\vert\!\!\!\!\!\!\!\!\!\!\!\!&
              \mbox{if $a \in \{0,1,2,3,4,5',5\}$}
            \\
            0
              &\!\!\!\!\!\!\!\!\!\!\!\!\vert\!\!\!\!\!\!\!\!\!\!\!\!&
              \mbox{otherwise}
          \end{array}
          \!\!\!\!\!\!\!\!
          \right\}
          &\!\!\!\!\!\!\!\!\!\!\!\mapsfrom\!\!\!\!\!\!\!\!\!\!\!&
          e^a
          \\
          \left.
          \begin{array}{lcl}
            e_{a_1 a_2}
              &\!\!\!\!\!\!\!\!\!\!\!\!\vert\!\!\!\!\!\!\!\!\!\!\!\!&
              \mbox{if even num. of $a_i$s $\in \{0,1,2,3,4,5',5\}$}
            \\
            0
              &\!\!\!\!\!\!\!\!\!\!\!\!\vert\!\!\!\!\!\!\!\!\!\!\!\!&
              \mbox{otherwise}
          \end{array}
          \!\!\!\!\!\!\!\!
          \right\}
          &\!\!\!\!\!\!\!\!\!\!\mapsfrom\!\!\!\!\!\!\!\!\!\!&
          e_{a_1 a_2}
          \\
          \left.
          \begin{array}{lcl}
            e_{a_1 \cdots a_5}
              &\!\!\!\!\!\!\!\!\!\!\!\!\vert\!\!\!\!\!\!\!\!\!\!\!\!&
              \mbox{if odd num. of $a_i$s $\in \{0,1,2,3,4,5',5\}$}
            \\
            0
              &\!\!\!\!\!\!\!\!\!\!\!\!\vert\!\!\!\!\!\!\!\!\!\!\!\!&
              \mbox{otherwise}
          \end{array}
          \!\!\!\!\!\!\!\!
          \right\}
          &\!\!\!\!\!\!\!\!\!\!\mapsfrom\!\!\!\!\!\!\!\!\!\!&
          e_{a_1 \cdots a_5}
          \\
          (P_{{}_{\mathbf{16}}}\psi)^\alpha
          &\!\!\!\!\!\!\!\!\!\!\mapsfrom\!\!\!\!\!\!\!\!\!\!&
          \psi^\alpha
          \\
          \eta
          &\!\!\!\!\!\!\!\!\!\!\mapsfrom\!\!\!\!\!\!\!\!\!\!&
          \eta
        \end{array}
        $
      }
    }
    &&&&
    {\phantom{
    \big(
      \mathbb{T}^{10,1\vert \mathbf{32}}
    \big)_{\mathrm{ex}_s}
    }}
    \\
    \\
    \mbox{
      \tiny
      \color{blue}
      \begin{tabular}{c}
        super-exceptional
        \\
        spacetime
      \end{tabular}
    }
       \hspace{-1.5cm}
   &
    \big(
      \mathbb{T}^{5,1\vert \mathbf{8}}
      \times
      \mathbb{T}^1
    \big)_{\mathrm{ex}_s}
    \ar[dd]_-{ \pi^{\sfrac{1}{2}\mathrm{M5}}_{\mathrm{ex}_s} }
    \ar@{^{(}->}[rrrr]
    \ar@{^{(}->}@/^2pc/[rrrrrrrr]|-{\; i_{\mathrm{ex}_s} \;}
    &&&&
    \big(
      \mathbb{T}^{6,1\vert \mathbf{16}}
    \big)_{\mathrm{ex}_s}
    \ar[dd]|-{ \pi_{\mathrm{ex}_s}^{ \mathrm{HET} } }
    \ar@{^{(}->}[rrrr]
    &&&&
    \ar@(ul,ur)^{G_{\mathrm{ADE}} = \langle \pmb{\Gamma}_{6789} \rangle}
    \big(
      \mathbb{T}^{10,1\vert \mathbf{32}}
    \big)_{\mathrm{ex}_s}
    \ar[dd]^-{ \pi_{\mathrm{ex}_s} }
    \\
    &
    &&
    \mathclap{
    \mbox{
      \tiny
      \color{blue}
      \begin{tabular}{c}
        breaking supersymmetry
        \\
        from $D=7$, $\mathcal{N} =1$
        \\
        to $D = 6$, $\mathcal{N} = (1,0)$
      \end{tabular}
    }
    }
    &&
    &&
    \mathclap{
    \mbox{
      \tiny
      \color{blue}
      \begin{tabular}{c}
        embedding of
        \\
        fixed/singular locus
      \end{tabular}
    }
    }
    &&
    \\
    \mbox{
      \tiny
      \color{blue}
      \begin{tabular}{c}
        super-
        \\
        spacetime
      \end{tabular}
    }
    \hspace{-1.8cm}
    &
    \mathbb{T}^{5,1\vert \mathbf{8}}
    \times
    \mathbb{T}^1
    \ar@{_{(}->}[rrrr]
    \; \ar@{_{(}->}@/_2pc/[rrrrrrrr]|-{\; i \; }
    &&&&
    \mathbb{T}^{6,1\vert \mathbf{16}}
    \; \ar@{_{(}->}[rrrr]
    &&&&
    \mathbb{T}^{10,1\vert \mathbf{32}}
    \ar@(dl,dr)_{G_{\mathrm{ADE}} = \langle \pmb{\Gamma}_{6789} \rangle}
    \\
    \\
    &
    \mbox{
      \tiny
      \color{blue}
      $\tfrac{1}{2}\mathrm{M5}$
    }
    &&&&
    \mbox{
      \tiny
      \color{blue}
      $\mathrm{MK6}$
    }
  }
\end{equation}

\item  {\bf (ii)}
the super-exceptional $\tfrac{1}{2}\mathrm{M5}$-spacetime (Def.
\ref{HalfM5LocusAndItsExceptionalTangentBundle})
inside the super-exceptional heterotic M-theory spacetime
(Def. \ref{SuperExceptionalMTheorySpacetime})
inside the $n=11$ super-exceptional M-theory spacetime
(Def. \ref{ExceptionalTangentSuperSpacetime})
as follows:
\begin{equation}
  \hspace{-.45cm}
  \xymatrix@R=1.2em@C=2.5em{
    \mbox{
      \tiny
      \color{blue}
      \begin{tabular}{c}
        super-exceptional
        \\
        spacetime
      \end{tabular}
    }
       \hspace{-1.5cm}
   &
    \big(
      \mathbb{T}^{5,1\vert \mathbf{8}}
      \times
      \mathbb{T}^1
    \big)_{\mathrm{ex}_s}
    \ar[dd]_-{ \pi^{\sfrac{1}{2}\mathrm{M5}}_{\mathrm{ex}_s} }
    \ar@{^{(}->}[rrrr]
    \ar@{^{(}->}@/^3.4pc/[rrrrrrrr]|-{\; i_{\mathrm{ex}_s} \;}
    &&&&
    \ar@(ul,ur)^{G_{\mathrm{ADE}} = \langle \pmb{\Gamma}_{6789} \rangle}
    \big(
      \mathbb{T}^{9,1\vert \mathbf{16}}
      \times
      \mathbb{T}^1
    \big)_{\mathrm{ex}_s}
    \ar[dd]|-{~\pi_{\mathrm{ex}_s}^{ \mathrm{MK6} }~ }
    \ar@{^{(}->}[rrrr]
    &&&&
    \big(
      \mathbb{T}^{10,1\vert \mathbf{32}}
    \big)_{\mathrm{ex}_s}
    \ar[dd]^-{ \pi_{\mathrm{ex}_s} }
    \\
    &
    &&
    \mathclap{
    \mbox{
      \tiny
      \color{blue}
      \begin{tabular}{c}
        embedding of
        \\
        fixed/singular locus
      \end{tabular}
    }
    }
    &&
    &&
    \mathclap{
    \mbox{
      \tiny
      \color{blue}
      \begin{tabular}{c}
        breaking supersymmetry
        \\
        from $D=11$, $\mathcal{N} =1$
        \\
        to $D = 10$, $\mathcal{N} = (1,0)$
      \end{tabular}
    }
    }
    &&
    \\
    \mbox{
      \tiny
      \color{blue}
      \begin{tabular}{c}
        super-
        \\
        spacetime
      \end{tabular}
    }
    \hspace{-1.8cm}
    &
    \mathbb{T}^{5,1\vert \mathbf{8}}
    \times
    \mathbb{T}^1
    \ar@{_{(}->}[rrrr]
    \; \ar@{_{(}->}@/_3.4pc/[rrrrrrrr]|-{\; i \; }
    &&&&
    \mathbb{T}^{9,1\vert \mathbf{16}}
    \ar@(dl,dr)_{G_{\mathrm{ADE}} = \langle \pmb{\Gamma}_{6789} \rangle}
    \ar@{_{(}->}[rrrr]
    &&&&
    \mathbb{T}^{10,1\vert \mathbf{32}}
    \\
    \\
    &
    \mbox{
      \tiny
      \color{blue}
      $\tfrac{1}{2}\mathrm{M5}$
    }
    &&&&
    \mbox{
      \tiny
      \color{blue}
      $\mathrm{HET}$
    }
  }
\end{equation}

\end{lemma}

\begin{remark}[The M5 locus for admissible  values of the parameters]
Henceforth we declare  that the parameter $s$ from \eqref{s}
is distinct from -6, i.e. $s \in \mathbb{R} \setminus \{0,-6\}$
as in \eqref{sNotMinusSix}.
In this case,  Prop. \ref{LiftOfVectorField} applies
and we have, in summary,
the situation as shown in the following diagram:
\begin{equation}
  \label{theM5LocusWithLieDerivative}
  \hspace{-1cm}
  \xymatrix@C=4.5em@R=1.5em{
    \Sigma
    \ar[ddr]
    \ar@{-->}[r]^-{ (f,H) }
    &
    \big(
      \mathbb{R}^{ 5,1\vert \mathbf{8} } \times \mathbb{R}^1
    \big)
    _{\mathrm{ex}_s}
    \ar@(ul,ur)^{ \color{blue} \mathcal{L}_{\widehat {v_5}} }
    \ar[dd]^-{ \pi^{\sfrac{1}{2}\mathrm{M5}}_{\mathrm{ex}_s} }
    \ar@{^{(}->}[r]^-{ i_{\mathrm{ex}_s} }
    &
    \mathbb{R}^{ 10,1\vert \mathbf{16} }_{\mathrm{ex}_s}
    \ar[dd]^>>>>>>{\ }="t"_-{ \pi_{\mathrm{ex}_s} }
    \ar[rrr]^-{
      H_{\mathrm{ex}_s}
      \;\wedge\;
      \overset{
         = d H_{\mathrm{ex}_s}
      }{\footnotesize
        \overbrace{
          \pi^\ast_{\mathrm{ex}_s}\mu_{{}_{\rm M2}}
        }
      }
      +
      2
      i^\ast_{\mathrm{ex}_s}
      \mu_{{}_{\rm M5}}
    }_>>>>>>{\ }="s"
    &&&
    S^7_{\mathbb{R}}
    \ar[dd]^-{ h_{\mathbb{H}} }
    \\
    \\
    &
    \mathbb{R}^{ 5,1\vert \mathbf{8} } \times \mathbb{R}^{1}
    \ar@(dl,dr)_{ \color{blue} \mathcal{L}_{v_5} }
    \ar@{^{(}->}[r]^-{ i }
    &
    \mathbb{R}^{ 10,1\vert \mathbf{32} }
    \ar[rrr]_{ (\mu_{{}_{\rm M2}}, 2 \mu_{{}_{\rm M5}}) }
    &&&
    S^4_{\mathbb{R}}\;.
    \ar@{=>}^{H_{\mathrm{ex}_s}} "s"; "t"
  }
\end{equation}
\end{remark}

\begin{remark}[Alternative brane configuration $\mathrm{MO1} \Vert \mathrm{MO9}$]
  \label{MWaveInMO9}
  For the case when the parameter value is $s = - 6$ after all, there is an alternative
  brane configuration one may consider, namely the configuration
  of an ``M-wave'' (see \cite[2.2.3]{HSS18}) inside an MO9-plane,
  i.e., with the spinor projection \eqref{FermionProjections}
  replaced by the following:
  $$
  \begin{array}{rclclcl}
    \mathrm{MO9}
    &&
    \pmb{\Gamma}_{9\phantom{1}} (P \psi) & = & (P \psi)\;,
    &
    \\
    \mathrm{MO1}
    &&
    \pmb{\Gamma}_{01} (P \psi) & = & (P \psi)\;.
    &&
  \end{array}
  $$
  In this case, there is a corresponding alternative to the
  super-exceptional isometry
  \eqref{LiftedVectorField} given by
  $$
    v_9^{\mathrm{ex}_s}
    \;\coloneqq\;
    \delta^a_9 v_a
      -
    (1+s)
    \tfrac{1}{2}
    \epsilon^{a_1 a_2}_{0 1}
    v_{a_1 a_2}
    +
    \chi^\alpha v_{\alpha}
    \,.
  $$
  With this alternative
  brane configuration and alternative super-exceptional isometry, all of the
  following constructions go through for all of $s \in \mathbb{R} \setminus \{0\}$,
  including $s = -6$; but then there is no value of $s$ for which the leading
  term of the super-exceptional Perry-Schwarz Lagrangian equals the original
  bosonic Perry-Schwarz Lagrangian, i.e., what fails is item
  $(\hyperlink{SuperExceptionalPSLagrangianForsMinusThree}{\bf ii})$ of
  Prop. \ref{ExceptionalPreimageOfPerrySchwarzLagrangian} below.
  This does not mean that this alternative case
  is not of interest, but its interpretation will
  need to be discussed elsewhere.
\end{remark}

We record the following basic fact:
\begin{lemma}[Vanishing bilinears on $\tfrac{1}{2}$M5 spacetime]
  \label{VanishingSpinorPairingsOnHalfM5Spacetime}
  The following bispinor pairings vanish identically
  $$
    \begin{array}{lcl}
      (\overline{P_{{}_{\mathbf{8}}}\psi})
      \Gamma_5
      (P_{{}_{\mathbf{8}}}\psi)
      & = &
      0
      \\
      (\overline{P_{{}_{\mathbf{8}}}\psi})
      \Gamma_{a_1 a_2}
      (P_{{}_{\mathbf{8}}}\psi)
      & = &
      0 \phantom{AA} \mbox{ for $a_1, a_2 \neq 5$ }\!,
      \\
      (\overline{P_{{}_{\mathbf{8}}}\psi})
      \Gamma_{a_1 \cdots a_5}
      (P_{{}_{\mathbf{8}}}\psi)
      & = &
      0 \phantom{AA} \mbox{ if one $a_i = 5$  }\!.
    \end{array}
  $$
\end{lemma}
\begin{proof}
Consider the following computation, for any
$a_i \in \{1,2,3,4,5',5,6,7,8,9\}$:
$$
  \begin{aligned}
    (\overline{P_{{}_{\mathbf{8}}}\psi})
    \Gamma_{a_1 \cdots a_n}
    (P_{{}_{\mathbf{8}}}\psi)
    & =
    (P_{{}_{\mathbf{8}}}\psi)^\dagger
    \Gamma_{0 a_1 \cdots a_n}
    (P_{{}_{\mathbf{8}}}\psi)
    \\
    & =
    (P_{{}_{\mathbf{8}}}\psi)^\dagger
    \Gamma_{0 a_1 \cdots a_n}
    \pmb{\Gamma}_5
    (P_{{}_{\mathbf{8}}}\psi)
    \\
    & =
    \pm
    (P_{{}_{\mathbf{8}}}\psi)^\dagger
    \pmb{\Gamma}_5
    \Gamma_{0 a_1 \cdots a_n}
    (P_{{}_{\mathbf{8}}}\psi)
    \\
    & =
    \sigma
    \cdot
    (\pmb{\Gamma}_5 P_{{}_{\mathbf{8}}}\psi)^\dagger
    \Gamma_{0 a_1 \cdots a_n}
    (P_{{}_{\mathbf{8}}}\psi)
    \\
    & =
    \sigma
    \cdot
    (P_{{}_{\mathbf{8}}}\psi)^\dagger
    \Gamma_{0 a_1 \cdots a_n}
    (P_{{}_{\mathbf{8}}}\psi)
    \\
    & =
    \sigma
    \cdot
    (\overline{P_{{}_{\mathbf{8}}} \psi})
    \Gamma_{0 a_1 \cdots a_n}
    (P_{{}_{\mathbf{8}}}\psi)\;.
  \end{aligned}
$$
Here the first line is the definition of the Dirac adjoint,
the second line uses that $\pmb{\Gamma}_5$ is the identity
on the projected spinors, by definition. In the third step
we commute $\pmb{\Gamma}_5 = i \Gamma_5$ with $\Gamma_{0 a_1 \cdots a_n}$,
thereby picking up a sign
\begin{equation}
  \label{SignsFromCommutingGamma5}
  \sigma
  \;=\;
  \left\{
    \begin{array}{ccc}
      +1 &\vert& \mbox{ odd number of $a_i$s $\neq 5$ }
      \\
      -1 &\vert& \mbox{ even number of $a_i$s $\neq 5$ }
    \end{array}
  \right.
\end{equation}
Finally we use $(\pmb{\Gamma}_5)^\dagger = \pmb{\Gamma}_5$
from \eqref{CMinusMajoranaRepresentationForPin} to
absorb the $\pmb{\Gamma}_5$ again, this time into the left
spinor factor.
Hence the expression we started with equals its product with
$\sigma$, and so vanishes when $\sigma = -1$, hence
when $\Gamma_{a_1 \cdots a_n}$ has an even number of indices
differing from 5.
\end{proof}

For the M-brane cochains, and in terms of
super-exceptional embedding, this means the following:

\begin{lemma}[Pullback of M-brane cocycles to $\tfrac{1}{2}\mathrm{M5}$]
\label{PullbackOfM2CocycleToHalfM5}
For the M-brane cocycles \eqref{TheMBraneCocycles},
pulled back along the super-embedding
of $\tfrac{1}{2}$M5-embedding $i$ in \eqref{M5FixedLocusNormallyEnhanced},
we have
\begin{eqnarray}
  \label{RestrictionOfM2CocycleToMO9}
  i^\ast_{{}_{}} \mu_{{}_{\rm M2}}
  & = &
  e^5 \wedge \iota_{v_5} i^\ast \big( \mu_{{}_{\rm M2}} \big),
  \\
  i^\ast_{{}_{}} \mu_{{}_{\rm M5}}
  & = &
  (\mathrm{id} - e^5 \wedge \iota_{v_5})
  i^\ast \mu_{{}_{\rm M5}},
\end{eqnarray}
with $\iota_{v_5}$ given in \eqref{TheVectorField}.
Hence, for the super-exceptional M-brane cocycles
(Def. \ref{ExceptionalM5SuperCocycle}), we have
$$
  \begin{aligned}
    (i_{\mathrm{ex}_s})^\ast
    (\pi_{\mathrm{ex}_s})^\ast
    \mu_{{}_{\rm M2}}
    & =
    e^5 \wedge \iota_{v_5^{\mathrm{ex}_s}}
    (i_{\mathrm{ex}_s})^\ast
    (\pi_{\mathrm{ex}_s})^\ast
    \mu_{{}_{\rm M2}},
    \\
    (i_{\mathrm{ex}_s})^\ast
    (\pi_{\mathrm{ex}_s})^\ast
    \mu_{{}_{\rm M5}}
    & =
    (\mathrm{id} - e^5 \wedge \iota_{v_5^{\mathrm{ex}_s}})
    (i_{\mathrm{ex}_s})^\ast
    (\pi_{\mathrm{ex}_s})^\ast
    \mu_{{}_{\rm M5}}
    \,.
  \end{aligned}
$$
\end{lemma}
\begin{proof}
  By Lemma \ref{SuperExceptionalEmbeddings}, the pullbacks along $i$ and
  $i_{\mathrm{ex}_s}$ act on the spinors by applying the projection
  $P_{{}_{\mathbf{8}}}$. Hence Lemma \ref{VanishingSpinorPairingsOnHalfM5Spacetime}
  applies to the spinor bilinears after restriction
  and says that all summands in the pullbacks of the
  M-brane cochains that do not contain an index = 5
  (for the M2-brane cocycle) or do contain an index = 5
  (for the M5-brane cochain) vanish. This is just what
  is expressed by the projection operators in
  \eqref{RestrictionOfM2CocycleToMO9}.
\end{proof}

\begin{example}[M-brane cochains pulled back to super-exceptional brane loci]
  \label{MBraneCochainsPulledBackToSuperExceptionalBraneLoci}
  The pullbacks of the
  super M2-brane cocycle $\mu_{{}_{\rm M2}}$
  and of the super M5-brane cochain
  $\mu_{{}_{\rm M5}}$ \eqref{TheMBraneCocycles}
  horizontally along the
  embeddings in \eqref{M5FixedLocusNormallyEnhanced}
  are as follows:
\begin{equation}
  \hspace{-.5cm}
  \xymatrix@C=6em@R=1.2em{
    \big(
      \mathbb{T}^{5,1\vert \mathbf{8}}
      \times
      \mathbb{T}^1
    \big)_{\mathrm{ex}_s}
    \ar[dd]_-{ \pi^{\sfrac{1}{2}\mathrm{M5}}_{\mathrm{ex}_s} }
    \ar@{^{(}->}[rr]
    \ar@{^{(}->}@/^2pc/[rrrr]^-{ i_{\mathrm{ex}_s} }
    &&
    \big(
      \mathbb{T}^{6,1\vert \mathbf{16}}
    \big)_{\mathrm{ex}_s}
    \ar[dd]
    \ar[rr]
    \ar@{^{(}->}[rr]
    &&
    \big(
      \mathbb{T}^{10,1\vert \mathbf{32}}
    \big)_{\mathrm{ex}_s}
    \ar[dd]^-{ \pi_{\mathrm{ex}_s} }
    \\
    \\
    \mathbb{T}^{5,1\vert \mathbf{8}}
    \times
    \mathbb{T}^1
    \ar@{_{(}->}[rr]
    \; \ar@{_{(}->}@/_2pc/[rrrr]|-{ i }
    &&
    \mathbb{T}^{6,1\vert \mathbf{16}}
    \ar@{_{(}->}[rr]
    &&
    \mathbb{T}^{10,1\vert \mathbf{32}}
    }
    \end{equation}

    \begin{equation*}
    \footnotesize
\xymatrix@C=1em@R=1pt{
              \tfrac{i}{2}
    \underset{\tiny \mathclap{
        a \in \{0,1,2,3,4,5'\}
      }
    }{\sum}
     (
    (P_{{}_{\mathbf{8}}}\overline{\psi})
    \Gamma_{a 5}
    (P_{{}_{\mathbf{8}}}\psi)
    )
    \wedge
    e^{a} \wedge e^{5}
            \ar@{<-|}[rr]
    &&
    \tfrac{i}{2}
    \underset{\tiny
      \mathclap{
        a_i \in \{0,1,2,3,4,5',5\}
      }
    }{\sum}
    (
    (P_{{}_{\mathbf{16}}}\overline{\psi})
    \Gamma_{a_1 a_2}
    (P_{{}_{\mathbf{16}}}\psi)
    )
    \wedge
    e^{a_1} \wedge e^{a_2}
    \ar@{<-|}[rr]
    &&
    \underset{
      \eqqcolon \mu_{{}_{\rm M2}}
    }{
    \underbrace{
    \tfrac{i}{2}
    \underset{\tiny
      \mathclap{
        a_i \in \{0,1,2,3,4,5',5,6,7,8,9\}
      }
    }{\sum}
    (
    \overline{\psi}\;
    \Gamma_{a_1 a_2}
    \psi
    )
    \wedge
    e^{a_1} \wedge e^{a_2}
    }
    }
    \\
    \tfrac{1}{5!}
    \underset{\tiny
      \mathclap{
        a_i \in \{0,1,2,3,4,5'\}
      }
    }{\sum}
    (
    (\overline{P_{{}_{\mathbf{8}}}\psi})
    \Gamma_{a_1 \cdots a_5}
    (P_{{}_{\mathbf{8}}}\psi)
    )
    \wedge
    e^{a_1} \wedge \cdots \wedge e^{a_5}
    \ar@{<-|}[rr]
    &&
    \tfrac{1}{5!}
    \underset{\tiny
      \mathclap{
        a_i \in \{0,1,2,3,4,5',5\}
      }
    }{\sum}
    (
    (\overline{P_{{}_{\mathbf{16}}}\psi})
    \Gamma_{a_1 \cdots a_5}
    (P_{{}_{\mathbf{16}}}\psi)
    )
    \wedge
    e^{a_1} \wedge \cdots \wedge e^{a_5}
    \ar@{<-|}[rr]
    &&
    \underset{
      \eqqcolon \mu_{{}_{\rm M5}}
    }{
    \underbrace{
    \tfrac{1}{5!}
    \underset{\tiny
      \mathclap{
        a_i \in \{0,1,2,3,4,5',5,6,7,8,9\}
      }
    }{\sum}
    (
    \overline{\psi}\;
    \Gamma_{a_1 \cdots a_5}
    \psi
    )
    \wedge
    e^{a_1} \wedge \cdots \wedge e^{a_5}
    }
    }
  }
\end{equation*}
Here $P_{{}_{\mathbf{16}}}$ and $P_{{}_{\mathbf{8}}}$ are the spinor projections
from Def. \ref{SpinorProjection}, and on the far left we used
Lemma \ref{PullbackOfM2CocycleToHalfM5}
to recognize that the M2-brane cocycle
on the $\tfrac{1}{2}\mathrm{M5}$-locus retains only the summands
proportional to $e^5$, while the M5-brane cochain
on the $\tfrac{1}{2}\mathrm{M5}$-locus retains only the summands
not proportional to $e^5$.
Notice that the vertical pullback is syntactically the identity,
due to \eqref{FermionicExtensionOfExceptionalTangentSuperspacetime}.
This makes manifest that the vertical pullback
to the exceptional spacetimes intertwines the
contraction operations for $v_5$ \eqref{TheVectorField}
and for $v_5^{\mathrm{ex}_s}$ \eqref{LiftedVectorField}:
\begin{equation}
  \label{PullbacktoExceptionalIntertwinesContractions}
  (\pi^{\sfrac{1}{2}\mathrm{M5}}_{\mathrm{ex}_s})^\ast \circ \iota_{v_5}
  \;=\;
  \iota_{v_5}^{\mathrm{ex}_s}
  \circ
  (\pi^{\sfrac{1}{2}\mathrm{M5}}_{\mathrm{ex}_s})^\ast
  \,.
\end{equation}
Since we also have
$$
  i^\ast \circ \iota_{v_5}
  \;=\;
  \iota_{v_5} \circ i^\ast ,
$$
this implies
\begin{equation}
  (i \circ \pi^{\sfrac{1}{2}\mathrm{M5}}_{\mathrm{ex}_s})^\ast
    \circ
  \iota_{v_5}
  \;=\;
  \iota_{v_5}^{\mathrm{ex}_s}
    \circ
  (i \circ \pi^{\sfrac{1}{2}\mathrm{M5}}_{\mathrm{ex}_s})^\ast
  \,.
\end{equation}
\end{example}

\section{Super-exceptional Perry-Schwarz \& Yang-Mills Lagrangians}
\label{SuperExceptionalLagrangian}

In Prop. \ref{ExceptionalPreimageOfPerrySchwarzLagrangian}
we find a natural super-exceptional pre-image of the bosonic Perry-Schwarz Lagrangian,
recorded as Def. \ref{SuperExceptionalPSLagrangian} below.
This allows us to extract the super-components (in Prop. \ref{calFexcDecomposed} below)
and identify the super-exceptional M-theory avatar of the gaugino field
(Remark \ref{Super2FormGaugeFieldStrengthAndGauginos} below).
We also find the super-exceptional lift of the topological Yang-Mills Lagrangian
(Def. \ref{SuperExceptionalTopologicalYM} below) and its
relation to the super-exceptional Perry-Schwarz Lagrangian
(Lemma \ref{SuperExceptionaltYMIsContractionalOfPS} below).
This plays a crucial role when we unify all
this super-exceptional data in \cref{SuperExceptionalReductionOfMTheoryCircle}.
Then we show (Prop. \ref{TrivializationOf7CocycleOnExceptionalHalfM5} below)
that the super-exceptional Perry-Schwarz Lagrangian arises
as the trivialization of the super-exceptional M5-brane cocycle
restricted along the super-embedding
of the $\tfrac{1}{2}\mathrm{M5}$-spacetime and compactified
on the M/HW-theory circle (Def. \ref{DimensionalReductionOfM5BraneSuperCocycle} below).
This is a key ingredient in the full super-embedding
Theorem \ref{TheTrivialization}
further below in \cref{SuperExceptionalM5Lagrangian}

\medskip

We start with
identifying the super-exceptional lift of the PS Lagrangian.
\begin{prop}[Super-exceptional lift of bosonic 2-flux and PS Lagrangian]
  \label{ExceptionalPreimageOfPerrySchwarzLagrangian}
  Consider
  $\Sigma_{\mathrm{bos}} \coloneqq \mathbb{R}^{5,1} \times \mathbb{R}^1$
  and a section $\sigma$
  of the super-exceptional $\tfrac{1}{2}$M5-brane spacetime projection
  \eqref{M5FixedLocusNormallyEnhanced}
$$
\hspace{10mm}
  \xymatrix@C=6.8em@R=1.5em{
    \mathllap{
      \Sigma_{\mathrm{bos}} =
    }\;
    \mathbb{R}^{5,1} \times \mathbb{R}^1
    \ar[ddr]
    \ar@(ul,ur)^{ \color{blue} v_5 }
    \ar@{-->}[r]^-{ \sigma }
    &
    \big(
      \mathbb{R}^{ 5,1\vert \mathbf{8} } \times \mathbb{R}^1
    \big)
    _{\mathrm{ex}_s}\;,
    \ar[dd]^-{ \pi^{\sfrac{1}{2}\mathrm{M5}}_{\mathrm{ex}_s} }
    &
    \!\!\!\!\!\!\!\!\!\!\!\!\!\!\!\!\!\!\!\!\!\!\!\!\!\!\!\!\!\!\!\!\!\!\!\!\!\!\!\!\!\!\!\!\!\!\!\!\!
    \Omega^\bullet
    \big(
      \mathbb{R}^{5,1\vert \mathbf{8}} \times \mathbb{R}^1
    \big)
    \ar@{<-}[ddr]
    \ar@{<-}[r]^-{
      \mbox{
        \raisebox{20pt}{
        \tiny
        $
        \begin{array}{ccl}
          d x^a
            &\!\!\!\!\!\!\!\!\!\!\!\!\!\mapsfrom\!\!\!\!\! & e^{a \leq 5}
          \\
          \tfrac{1}{\alpha_0(s)}
          d (B_{a_1 a_2}) &\!\!\!\!\!\!\!\!\!\!\!\!\!\mapsfrom\!\!\!\!\!& e_{a_1 a_2}
          \\
          0
            &\!\!\!\!\!\!\!\!\!\!\!\!\!\mapsfrom\!\!\!\!\!& e^{a \gt 5}
          \\
          0 &\!\!\!\!\!\!\!\!\!\!\!\!\!\mapsfrom\!\!\!\!\!& e_{a_1 \cdots a_5}
        \end{array}
        $
        }
      }
    }
    &
    \mathrm{CE}
    \Big(\!
      \big(
        \mathbb{R}^{5,1\vert \mathbf{8}} \times \mathbb{R}^1
      \big)_{\mathrm{ex}_s}
 \!   \Big).
    \ar@{<-}[dd]^{  (\pi^{\sfrac{1}{2}\mathrm{M5}}_{\mathrm{ex}_s})^\ast }
    \\
    \\
    &
    \mathbb{R}^{ 5,1\vert \mathbf{8} } \times \mathbb{R}^{1}
    \ar@(dl,dr)_{ \color{blue} v_5 }
    &&
    \mathrm{CE}
    \big(
      \mathbb{R}^{5,1 \vert \mathbf{8}} \times \mathbb{R}
    \big)
    }
$$

\item {\bf (i)}
If $\sigma$ is such that the normal forms $e^{a \gt 5}$ and the 5-index forms
$e_{a_1 \cdots a_5}$ pull back to zero, as shown above on the right, then the
pullback of the contraction of the transgression element $H_{\mathrm{ex}_s}$
(Prop. \ref{TransgressionElementForM2Cocycle}) with the lifted vector field
$v_5^{\mathrm{ex}_s}$ (Prop. \ref{LiftOfVectorField}) is the bosonic 2-form flux
$F$ \eqref{FcalF}:
\begin{equation}
  \label{PullbackOfContractedHs}
  \xymatrix@R=-8pt{
    \Omega^\bullet(\Sigma_{\mathrm{bos}})
    \ar@{<-}[rr]^-{ \sigma^{\ast} }
    &&
    \mathrm{CE}
    \Big(
      \big(
        \mathbb{R}^{5,1\vert \mathbf{8}}
        \times
        \mathbb{R}^1
      \big)_{\mathrm{ex}_s}
    \Big)
    \\
    \underset{
      = \,
      \color{blue}
      F
    }{
      \underbrace{
        \overset{
           = \mathcal{F}
        }{
          \overbrace{
            \iota_{v_5}H
          }
        }
        -
        \mathcal{L}_{v_5} B^{\mathrm{hor}}
      }
    }
    \ar@{<-|}[rr]
    &&
    \underset{
      \eqqcolon
      \,
      \color{blue}
      F_{\mathrm{ex}_s}
    }{
      \underbrace{
        \iota_{v_5^{\mathrm{ex}_s}} H_{\mathrm{ex}_s}
      }
    }
  }
\end{equation}
where on the left
$H \coloneqq  d(B_{a_2 a_2}) \wedge d x^{a_2} \wedge d x^{a_3}$
denotes the plain H-flux
\eqref{DecomposedCFieldAsSumOfBosonicAndFermionicContribution}
and $F$ its induced 2-form flux \eqref{FcalF}
according to Lemma \ref{RelatingFcalF};
and hence on the right we find a super-exceptional pre-image
$F_{\mathrm{ex}_s}$ of the 2-form flux.

\item \hypertarget{SuperExceptionalPSLagrangianForsMinusThree}{}
  {\bf (ii)} If, moreover, $s = - 3$
  (as in Example \ref{ExceptionalSuperspacetimeAtSEqualsMinusOne}),
we have
\begin{equation}
  \label{PullbackOfContractedHs}
  \xymatrix@R=-12pt{
    \Omega^\bullet(\Sigma_{\mathrm{bos}})
    \ar@{<-}[rr]^-{ \sigma^{\ast} }
    &&
    \mathrm{CE}
    \Big(
      \big(
        \mathbb{R}^{5,1\vert \mathbf{8}}
        \times
        \mathbb{R}^1
      \big)_{\mathrm{ex}_s}
    \Big)
    \\
    \underset{
      =
      \,
      { \color{blue} \mathbf{L}^{\!\!\mathrm{PS}} }
      \wedge
      dx^5
    }{
    \underbrace{
      -
      \tfrac{1}{2}
      \big(
        \iota_{v_5}H
        -
        \mathcal{L}_{v_5} B^{\mathrm{hor}}
      \big)
      \wedge
      H
      \wedge dx^5
    }
    }
    \ar@{<-|}[rr]
    &&
    \underset{
      =:
       \,
      {\color{blue} \mathbf{L}^{\!\!\mathrm{PS}}_{\mathrm{ex}_s}}
      \wedge e^5
    }{
      \underbrace{
        -
        \tfrac{1}{2}
        \big(
          \iota_{v_5^{\mathrm{ex}_{s}}} H_{\mathrm{ex}_{s}}
        \big)
        \wedge
        H_{\mathrm{ex}_{s}}
        \wedge e^5
      }
    }
  }
\end{equation}
where on the left we have the Perry-Schwarz Lagrangian
\eqref{NonCovariantLagrangianInFlatSpacetime},
and hence on the right we find a super-exceptional pre-image
$\mathbf{L}^{\!\!\mathrm{PS}}_{\mathrm{ex}_s}$.
\end{prop}
\begin{proof}
  By the assumption that $\sigma^\ast e_{a_1 \cdots a_5} = 0$,
  and since the odd forms $\sigma^\ast \psi$ and $\sigma^\ast \eta$ vanish
  after pullback to the bosonic
  space $\mathbb{R}^{5,1}\times \mathbb{R}^1$,
  we find from \eqref{DecomposedCFieldAsSumOfBosonicAndFermionicContribution}
  by direct computation that
  $$
    \begin{aligned}
      \sigma^\ast
      \big( \iota_{v_5^{\mathrm{ex}_s}}  H_{\mathrm{ex}_s} \big)
      & =
      \alpha_0(s)
      \cdot
      \sigma^\ast
      \big(
        \iota_{v_5^{\mathrm{ex}_s}}
        e_{a_2 a_3} \wedge e^{a_2} \wedge e^{a_3}
      \big)
      \\
      & =
      \alpha_0(s)
      \cdot
      \sigma^\ast
      \big(
        - 2
        e_{5 a_3} \wedge e^{5} \wedge e^{a_3}
      \big)
      \\
      & =
      - 2 ( d B_{5 a_3} ) \wedge d x^{5} \wedge d x^{a_3}
      \\
      & =
      - 2
        (\partial_{v^{a_1}} B_{5 a_3} )
        \wedge dx^{a_1} \wedge d x^{5} \wedge d x^{a_3}
      \\
      & =
      \iota_{v_5}
      \big(
          \partial_{a_1} B_{ a_2 a_3  }
          \wedge dx^{a_1} \wedge d x^{a_2} \wedge d x^{a_3}
      \big)
      -
      \mathcal{L}_{v^{5}}
      \underset{a_2, a_3 \neq 5}{\sum}
      B_{ a_2 a_3  }
      \wedge d x^{a_2} \wedge d x^{a_3}
      \\
      & =
      \iota_{v_5}
      \big(
        \underset{
          = H
        }{
        \underbrace{
          d B_{a_2 a_3}
          \wedge d x^{a_2} \wedge d x^{a_3}
        }
        }
      \big)
      -
      \mathcal{L}_{v^{5}}
      \underset{a_2, a_3 \neq 5}{\sum}
      B_{ a_2 a_3  }
      \wedge d x^{a_2} \wedge d x^{a_3}
      \\
      & =
      \iota_{v_5} H - \mathcal{L}_{v_5} B^{\mathrm{hor}}.
    \end{aligned}
  $$
  This proves the first statement. For the second, it is
  now sufficient to observe with \eqref{HwithsMinusThree}
  that, by the assumption $s = -3$, we have in the present case
  $\sigma^\ast H_{\mathrm{ex}_s} = H$. Hence
   the second claim now follows directly from the first.
\end{proof}

\begin{defn}[Super-exceptional (dual) 2-flux and PS Lagrangian]
  \label{SuperExceptionalPSLagrangian}
  On the super-exceptional $\tfrac{1}{2}\mathrm{M5}$ spacetime
  $\big(
    \mathbb{R}^{5,1\vert \mathbf{8}} \times \mathbb{R}^1 \big)_{\mathrm{ex}_s}$
     from Def. \ref{HalfM5LocusAndItsExceptionalTangentBundle}
    define the following forms
    in $\mathrm{CE}\big((
    \mathbb{R}^{5,1\vert \mathbf{8}} \times \mathbb{R}^1)_{\mathrm{ex}_s}\big)$:

\item {\bf (i)}   The \emph{super-exceptional 2-flux} and
  \emph{dual super-exceptional 2-flux}, respectively:
\begin{equation}
    \label{ExceptionalF}
    F_{\mathrm{ex}_s}
    \;\coloneqq\;
    \iota_{v_5^{\mathrm{ex}_s}}
    (i_{\mathrm{ex}_s})^\ast
    H_{\mathrm{ex}_s}
    \,,
    \phantom{AAAAA}
    \widetilde F_{\mathrm{ex}_s}
    \;\coloneqq\;
    (i_{\mathrm{ex}_s})^\ast
    H_{\mathrm{ex}_s}
    -
    e^5
      \wedge
    F_{\mathrm{ex}_s}
    \,,
\end{equation}
where $H_{\mathrm{ex}_s}$ is from Prop. \ref{TransgressionElementForM2Cocycle} and
$\iota_{v_5^{\mathrm{ex}_s}}$ from Prop. \ref{LiftOfVectorField}.
\item {\bf (ii)} The \emph{super-exceptional PS Lagrangian}:
\begin{equation}
    \label{LagrangianPSExceptional}
    \mathbf{L}^{\!\!\mathrm{PS}}_{\mathrm{ex}_s}
    \;\coloneqq\;
    -
    \tfrac{1}{2}
    F_{\mathrm{ex}_s}
    \wedge
    \widetilde F_{\mathrm{ex}_s}
    \,.
\end{equation}
\end{defn}

With the exceptional pre-image of the bosonic 2-form flux identified,
we find the induced supersymmetric completion, keeping in mind the notation
${\rm deg}=({\rm bosonic}, {\rm fermionic})$:
\begin{prop}[Super 2-flux from super-exceptional 2-flux]
  \label{calFexcDecomposed}
  For paraneter value $s = -3$ \cref{s}, the exceptional 2-flux density \eqref{PullbackOfContractedHs}
  is the sum of the bosonic term plus a fermionic term
  $F_{(1,1)}$ as follows:
  \begin{equation}
    \label{DecomposingcalFexc}
    F_{\mathrm{ex}_{s}}
    \;=\;
    \underset{
      \eqqcolon F_{(2,0)}
    }{
      \underbrace{
        F
      }
    }
    +
    \underset{
      \eqqcolon F_{(1,1)}
    }{
      \underbrace{
        (\overline{\psi} \;\Gamma_a \chi)\wedge e^a
      }
    }
    \;+\;
    \mathcal{O}\big(\{ e_{a_1 \cdots a_5} \}\big)
    \,.
  \end{equation}
\end{prop}
\begin{proof}
  By Example \ref{ExceptionalSuperspacetimeAtSEqualsMinusOne} we have
  $
    H_{\mathrm{ex}_{-3}}
    \;=\;
    H
      +
    (\overline{\psi} \; \Gamma_a \eta)
     \wedge
    e^a + \mathcal{O}(\{e_{a_1 \cdots a_5}\})\,.
  $
  From this the statement follows by the definition
  \eqref{LiftedVectorField} of $v_5^{\mathrm{ex}_s}$ and
  using the identities
  $(\overline{P \psi})\Gamma_5 (P \chi) = 0$
  and
  $(\overline{P \psi})\Gamma_{56789} (P \chi) = 0$
  from
  Lemma \ref{VanishingSpinorPairingsOnHalfM5Spacetime}.
\end{proof}
\begin{remark}[Super 2-form gauge field strength and gauginos]
\label{Super2FormGaugeFieldStrengthAndGauginos}
The summand
$$
  F_{(1,1)}
  \;=\;
  (\overline{\psi} \;\Gamma_a \chi) \wedge e^a
  $$
in \eqref{DecomposingcalFexc} is exactly the supersymmetic enhancement
of the gauge curvature in 10d SYM, with $\chi$ identified
as the gaugino field (\cite{Witten86}\cite[(4.14)]{ADR86}\cite[(2.27)]{BBLPT88}).
But by the last line of \eqref{LiftedVectorField},
$\chi$ is the component of the super-exceptional lift
$v_5^{\mathrm{ex}_s}$
of the isometry along $S^1_{\mathrm{HW}}$ in the
fermionic direction defined by the extra super-exceptional 1-form
$\eta$ \eqref{FermionicExtensionOfExceptionalTangentSuperspacetime}
$$
  \iota_{v_5^{\mathrm{ex}_s}}
  \;\colon\;
  \eta \mapsto \chi
  \,.
$$
In this way, it is the extra super-exceptional fermionic
coordinate $\eta$ which is the avatar on the super-exceptional
M-theory spacetime of what becomes the gaugino field
upon compactification to heterotic M-theory on $S^1_{\mathrm{HW}}$.

Note that  an
approximate construction of the 11d gravitino
in the context of $E_8$ gauge theory
as a condensate of the gauge theory fields is given in \cite{ES}.
\end{remark}

\medskip

From Prop. \ref{ExceptionalPreimageOfPerrySchwarzLagrangian}
it is clear that we have a super-exceptional lift of the
topological Yang-Mills Lagrangian
$\mathbf{L}^{\!\!\mathrm{tYM}} = -\tfrac{1}{2} F \wedge F$
\eqref{TopologicalYM}:

\begin{defn}
  \label{SuperExceptionalTopologicalYM}
The
\emph{super-exceptional topological Yang-Mills Lagrangian}
is the wedge square of the super-exceptional 2-flux \eqref{ExceptionalF}
from Def. \ref{SuperExceptionalPSLagrangian}:
\begin{equation}
  \label{SuperExceptionalThetaAngleTerm}
  \mathbf{L}^{\!\!\mathrm{tYM}}_{\mathrm{ex}_s}
  \;\coloneqq\;
  -
  \tfrac{1}{2}
  F_{\mathrm{ex}_s}
  \wedge
  F_{\mathrm{ex}_s}\;.
\end{equation}
\end{defn}

\begin{lemma}[Super-exceptional topological Yang-Mills
as compactification of super-exceptional Perry-Schwarz]
  \label{SuperExceptionaltYMIsContractionalOfPS}
  The contraction of the super-exceptional
  Perry-Schwarz Lagrangian
  \eqref{LagrangianPSExceptional}
  with the super-exceptional isometry
  $v_5^{\mathrm{ex}_s}$ \eqref{LiftedVectorField} along $S^1_{\mathrm{HW}}$
  is the super-exceptional topological Yang-Mills Lagrangian
  \eqref{SuperExceptionalThetaAngleTerm}:
  $$
    \iota_{v_5^{\mathrm{ex}_s}}
    \mathbf{L}^{\!\mathrm{PS}}_{\mathrm{ex}_s}
    \;=\;
    \mathbf{L}^{\!\mathrm{tYM}}_{\mathrm{ex}_s}
    \,.
  $$
\end{lemma}
\begin{proof}
We need to show that
  $$
    \iota_{v_5^{\mathrm{ex}_s}}
    \big(
      \widetilde F_{\mathrm{ex}_s}
      \wedge
      F_{\mathrm{ex}_s}
    \big)
    \;=\;
    F_{\mathrm{ex}_s}
    \wedge
    F_{\mathrm{ex}_s}\;.
  $$
But this follows directly from the definitions and the
the fact that the contraction is a graded derivation of degree
$(-1,\mathrm{even})$, hence in particular nilpotent. Indeed, we have
$$
  \begin{aligned}
    \iota_{v_5^{\mathrm{ex}_s}}
    \big(
      \widetilde F_{\mathrm{ex}_s}
      \wedge
      F_{\mathrm{ex}_s}
    \big)
    & =
    \iota_{v_5^{\mathrm{ex}_s}}
    \Big(
      \big(
        H_{\mathrm{ex}_s}
        -
        e^5 \wedge \iota_{v_5^{\mathrm{ex}_s}}H_{\mathrm{ex}_s}
      \big)
      \wedge
      \iota_{v_5^{\mathrm{ex}_s}}H_{\mathrm{ex}_s}
    \Big)
    \\
    & =
    \big(
      \iota_{v_5^{\mathrm{ex}_s}} H_{\mathrm{ex}_s}
    \big)
    \wedge
    \big(
      \iota_{v_5^{\mathrm{ex}_s}}H_{\mathrm{ex}_s}
    \Big)
    \\
    & =
    F_{\mathrm{ex}_s}
    \wedge
    F_{\mathrm{ex}_s}\;.
  \end{aligned}
$$

\vspace{-7mm}
\end{proof}

\medskip

In order to see the super-exceptional Perry-Schwarz Lagrangian
arise from the super-exceptional M5-brane cocycle,
we now first consider the M5-brane sigma-model
wrapped on the $S^1_{\mathrm{HW}}$-fiber
(see Remark \ref{TheHalfM5Spacetime})
of super-exceptional M-theory spacetime (Def. \ref{ExceptionalTangentSuperSpacetime}).
By the general rules of (double-)dimensional reduction of super $p$-brane
cocycles \cite[Sec. 3]{FSS16a}\cite[Sec. 3]{FSS16b}\cite[Sec. 2.2]{BSS18},
this means that we are to contract the plain M5-brane cocycle
\eqref{TheMBraneCocycles} with the vector field corresponding
to the flow along this fiber, hence with $v_5$ \eqref{TheVectorField}.
The following definition lifts this situation to
super-exceptional spacetime.

\begin{defn}[Super-exceptional circle compactification of M5 cocycle]
\label{DimensionalReductionOfM5BraneSuperCocycle}
The \emph{compactification on $S^1_{\mathrm{HW}}$}
of the super-exceptional M5-brane cocycle
$\mathbf{dL}^{\!\!\mathrm{WZ}}_{\mathrm{ex}_s}$
(Def. \ref{ExceptionalM5SuperCocycle}) pulled back along the
super-exceptional embedding $i_{\mathrm{ex}_s}$
\eqref{M5FixedLocusNormallyEnhanced} to the normally thickened
super-exceptional $\tfrac{1}{2}\mathrm{M5}$-spacetime
(Def. \ref{HalfM5LocusAndItsExceptionalTangentBundle})
is its contraction with the super-exceptional lift $v_5^{\mathrm{ex}_s}$
(Prop. \ref{LiftOfVectorField}) of the vector field $v_5$ along $S^1_{\mathrm{HW}}$
\eqref{HalfM5BraneSetup}:
\begin{equation}
  \label{DimensionalReductionOfM5Cocycle}
  \xymatrix{
    \big(
      \mathbb{R}^{5,1\vert \mathbf{8}}
      \times
      \mathbb{R}^1
    \big)_{\mathrm{ex}_s}
    \ar[rrrrrr]|-{\footnotesize
      \;\;\overset{
        \mathclap{
        \mbox{
          \tiny
          \color{blue}
          \begin{tabular}{c}
            contraction with
            \\
            super-exceptional
            \\
            isometry along $S^1_{\mathrm{HW}}$
          \end{tabular}
        }
        }
      }{
      \overbrace{
        \mathclap{ \phantom{ a \atop a } }
        \iota_{v_5^{\mathrm{ex}_s}}
      }}
      \;\;\;
      \underset{
        \mathclap{
        \mbox{
          \tiny
          \color{blue}
          \begin{tabular}{c}
            restriction to
            \\
            super-exceptional
            \\
            $\tfrac{1}{2}\mathrm{M5}$
          \end{tabular}
        }
        }
      }{
        \underbrace{
          \mathclap{ \phantom{ a \atop a } }
          (i_{\mathrm{ex}_s})^\ast
        }
      }
      \;\;\;
      \overset{
        \mathclap{
        \mbox{
          \tiny
          \color{blue}
          \begin{tabular}{c}
            super-exceptional
            \\
            M5-brane cocycle
          \end{tabular}
        }
        }
      }{
      \overbrace{
        \mathclap{ \phantom{ a \atop a } }
        \mathbf{dL}^{\!\!\mathrm{WZ}}_{\mathrm{ex}_s}
        }
      }
      \;\;
    }
        &&&&&&
    B^6 \mathbb{R}
  }.
\end{equation}
  Notice that this is indeed still a cocycle,
  in that it is closed,
  $
    d\,
    \big(
      \,
      \iota_{v_5^{\mathrm{ex}_s}}
      (i_{\mathrm{ex}_s})^\ast
      \mathbf{dL}^{\!\!\mathrm{WZ}}_{\mathrm{ex}_s}
    \big)
    =
    0
      $,
  by \eqref{LieDerivativeAlongExceptionalLiftOfv5VanishesOnHalfM5}
  in Prop. \ref{LiftOfVectorField}.
\end{defn}
\begin{remark}[Cocycle for $\mathrm{M5}_{/ S^1}$ as $\mathrm{D4}+\mathrm{KK}$]
  \label{CocycleForM5OnS1AsD4PlusKK}
  From the type $\mathrm{I}^\prime$ perspective on
  the $\tfrac{1}{2}\mathrm{M5}$-locus (as in Remark \ref{TheHalfM5Spacetime})
  the compactified M5-cocycle of Def. \ref{DimensionalReductionOfM5BraneSuperCocycle}
  would be that of a D4-brane
  inside a $\tfrac{1}{2}\mathrm{NS5} = \mathrm{NS5} \cap \mathrm{O8}$,
  by the general rules of dimensional reduction of brane cocycles
  \cite[Sec. 3]{FSS16a}\cite[Sec. 3]{FSS16b}\cite[Sec. 2.2]{BSS18}.
  However, since we need not consider \emph{double} dimensional reduction,
  in that the super-exceptional coordinate functions along $v_5$
  are still present in the normally thickened $\tfrac{1}{2}\mathrm{M5}$ locus
  $\big( \mathbb{R}^{5,1\vert \mathbf{8}} \times \mathbb{R}^1\big)_{\mathrm{ex}_s}$
  (Def. \ref{HalfM5LocusAndItsExceptionalTangentBundle}),
  fields on this would-be D4 may still depend on the M-theory circle direction,
  hence have KK-modes along the circle.
  In this sense, Def. \ref{DimensionalReductionOfM5BraneSuperCocycle}
  exhibits the brane cocycle corresponding to the
  perspective on the
  M5-brane as a non-perturbative D4-brane with KK-modes
  included, as considered
  in \cite{Douglas10}\cite{LPS10} (see \cite[3.4.3]{Lambert19}).
\end{remark}

We now establish the following super-exceptional analog of the super-embedding mechanism.

\begin{prop}[PS Lagrangian trivializes compactified M5-cocycle along super-exceptional embedding]
\label{TrivializationOf7CocycleOnExceptionalHalfM5}
The compactification of the restriction of the
super-exceptional
M5-brane cocycle (Def. \ref{DimensionalReductionOfM5BraneSuperCocycle})
along the embedding
of the super-exceptional $\tfrac{1}{2}\mathrm{M5}$ spacetime
(Def. \ref{HalfM5LocusAndItsExceptionalTangentBundle})
is trivialized by the super-exceptional PS Lagrangian
\eqref{LagrangianPSExceptional}:
\begin{equation}
  \label{TheTrivialization}
  \iota_{v_5^{\mathrm{ex}_s}}
  (i_{\mathrm{ex}_s})^\ast
  \mathbf{dL}^{\!\!\mathrm{WZ}}_{\mathrm{ex}_s}
  \;=\;
  d\,\mathbf{L}^{\!\!\mathrm{PS}}_{\mathrm{ex}_s}.
\end{equation}
\end{prop}
\begin{proof}
Unravelling the definitions, we have to show that
$$
  \begin{aligned}
    \iota_{v_5^{\mathrm{ex}_s}}
    (i_{\mathrm{ex}_s})^\ast
    \big(
      (\pi_{\mathrm{ex}_s})^\ast \mu_{{}_{\rm M5}}
      +
      \tfrac{1}{2}
      H_{\mathrm{ex}_s}
      \wedge
      d H_{\mathrm{ex}_s}
    \big)
    & =
    d
    \big(
      -
      \tfrac{1}{2}
      F_{\mathrm{ex}_s}
      \wedge
      \widetilde F_{\mathrm{ex}_s}
    \big)
    \,.
  \end{aligned}
$$
For the first summand on the left, we immediately obtain
\begin{equation}
\label{FirstSummandForContractedCocycle}
\begin{aligned}
  \iota_{v_5^{\mathrm{ex}_s}}
  (i_{\mathrm{ex}_s})^\ast
  (\pi_{\mathrm{ex}_s})^\ast
  \mu_{{}_{\rm M5}}
  & =
  (\pi^{\sfrac{1}{2}\mathrm{M5}}_{\mathrm{ex}_s})^\ast
  \underset{
    = 0
  }{
  \underbrace{
    \big(
      i^\ast
      \iota_{v_5}
      \mu_{{}_{\rm M5}}
    \big)
  }
  }
  \\
  & =
  0\;,
\end{aligned}
\end{equation}
by Lemma \ref{PullbackOfM2CocycleToHalfM5};
see Example \ref{MBraneCochainsPulledBackToSuperExceptionalBraneLoci}.
For the second summand
(or rather twice the second summand, for
notational convenience) we compute as follows:
\begin{equation}
  \label{SummandForContractedCocycleSecond}
  \begin{aligned}
    2
    \,
    \iota_{v_5}^{\mathrm{ex}_s}
    (i_{\mathrm{ex}_s})^\ast
    \Big(
      \tfrac{1}{2}
      H_{\mathrm{ex}_s}
      \wedge
      d H_{\mathrm{ex}_s}
    \Big)
    & =
    \iota_{v_5}^{\mathrm{ex}_s}
    (i_{\mathrm{ex}_s})^\ast
    \Big(
      H_{\mathrm{ex}_s}
      \wedge
      d H_{\mathrm{ex}_s}
    \Big)
    \\
    & =
    \iota_{v_5}^{\mathrm{ex}_s}
    \Big(
      \big(
        (i_{\mathrm{ex}_s})^\ast
        H_{\mathrm{ex}_s}
      \big)
      \wedge
      \big(
        (i_{\mathrm{ex}_s})^\ast
        d
        H_{\mathrm{ex}_s}
      \big)
    \Big)
    \\
    & =
      \big(
        \iota_{v_5}^{\mathrm{ex}_s}
        (i_{\mathrm{ex}_s})^\ast
        H_{\mathrm{ex}_s}
      \big)
      \wedge
      \big(
        (i_{\mathrm{ex}_s})^\ast
        d
        H_{\mathrm{ex}_s}
      \big)
    -
      \big(
        (i_{\mathrm{ex}_s})^\ast
        H_{\mathrm{ex}_s}
      \big)
      \wedge
      \big(
        \iota_{v_5}^{\mathrm{ex}_s}
        (i_{\mathrm{ex}_s})^\ast
        d
        H_{\mathrm{ex}_s}
      \big)
    \\
    & =
      \big(
        \iota_{v_5}^{\mathrm{ex}_s}
        (i_{\mathrm{ex}_s})^\ast
        H_{\mathrm{ex}_s}
      \big)
      \wedge
      \big(
        (i_{\mathrm{ex}_s})^\ast
        d
        H_{\mathrm{ex}_s}
      \big)
    +
      \big(
        (i_{\mathrm{ex}_s})^\ast
        H_{\mathrm{ex}_s}
      \big)
      \wedge
      \Big(
        d
        \big(
          \iota_{v_5}^{\mathrm{ex}_s}
          (i_{\mathrm{ex}_s})^\ast
          H_{\mathrm{ex}_s}
        \big)
      \Big)
    \\
    & =
      2
      \big(
        \iota_{v_5}^{\mathrm{ex}_s}
        (i_{\mathrm{ex}_s})^\ast
        H_{\mathrm{ex}_s}
      \big)
      \wedge
      \big(
        (i_{\mathrm{ex}_s})^\ast
        d
        H_{\mathrm{ex}_s}
      \big)
    -
    d
    \Big(
      \big(
        (i_{\mathrm{ex}_s})^\ast
        H_{\mathrm{ex}_s}
      \big)
      \wedge
      \big(
        \iota_{v_5}^{\mathrm{ex}_s}
        (i_{\mathrm{ex}_s})^\ast
        H_{\mathrm{ex}_s}
      \big)
    \Big)
    \\
    & =
      2
      \big(
        \iota_{v_5}^{\mathrm{ex}_s}
        (i_{\mathrm{ex}_s})^\ast
        H_{\mathrm{ex}_s}
      \big)
      \wedge
      \big(
        e^5
        \wedge
        \iota_{v_5}^{\mathrm{ex}_s}
        (i_{\mathrm{ex}_s})^\ast
        d
        H_{\mathrm{ex}_s}
      \big)
    -
    d
    \Big(
      \big(
        (i_{\mathrm{ex}_s})^\ast
        H_{\mathrm{ex}_s}
      \big)
      \wedge
      \big(
        \iota_{v_5}^{\mathrm{ex}_s}
        (i_{\mathrm{ex}_s})^\ast
        H_{\mathrm{ex}_s}
      \big)
    \Big)
    \\
    & =
      -
      2
      e^5
      \wedge
      \big(
        \iota_{v_5}^{\mathrm{ex}_s}
        (i_{\mathrm{ex}_s})^\ast
        H_{\mathrm{ex}_s}
      \big)
      \wedge
      d
      \big(
        \iota_{v_5}^{\mathrm{ex}_s}
        (i_{\mathrm{ex}_s})^\ast
        H_{\mathrm{ex}_s}
      \big)
    -
    d
    \Big(
      \big(
        (i_{\mathrm{ex}_s})^\ast
        H_{\mathrm{ex}_s}
      \big)
      \wedge
      \big(
        \iota_{v_5}^{\mathrm{ex}_s}
        (i_{\mathrm{ex}_s})^\ast
        H_{\mathrm{ex}_s}
      \big)
    \Big)
    \\
    & =
    d
    \Big(
      e^5
      \wedge
      \big(
        \iota_{v_5}^{\mathrm{ex}_s}
        (i_{\mathrm{ex}_s})^\ast
        H_{\mathrm{ex}_s}
      \big)
      \wedge
      \big(
        \iota_{v_5}^{\mathrm{ex}_s}
        (i_{\mathrm{ex}_s})^\ast
        H_{\mathrm{ex}_s}
      \big)
    \Big)
    -
    d
    \Big(
      \big(
        (i_{\mathrm{ex}_s})^\ast
        H_{\mathrm{ex}_s}
      \big)
      \wedge
      \big(
        \iota_{v_5}^{\mathrm{ex}_s}
        (i_{\mathrm{ex}_s})^\ast
        H_{\mathrm{ex}_s}
      \big)
    \Big)
    \\
    & =
    d
    \Big(
      -
      \big(
        (i_{\mathrm{ex}_s})^\ast
        H_{\mathrm{ex}_s}
        -
        e^5 \wedge
        \iota_{v_5}^{\mathrm{ex}_s}
        (i_{\mathrm{ex}_s})^\ast
        H_{\mathrm{ex}_s}
      \big)
      \wedge
      \big(
        \iota_{v_5}^{\mathrm{ex}_s}
        (i_{\mathrm{ex}_s})^\ast
        H_{\mathrm{ex}_s}
      \big)
    \Big)
    \\
    & =
    d
    \Big(
      -
      \widetilde F_{\mathrm{ex}_s}
      \wedge
      F_{\mathrm{ex}_s}
    \Big)
    \\
    & =
    2\,
    d
    \Big(
      -
      \tfrac{1}{2}
      \widetilde F_{\mathrm{ex}_s}
      \wedge
      F_{\mathrm{ex}_s}
    \Big)
    \,.
  \end{aligned}
\end{equation}
Here the first step just collects the factors.
The second fact uses that pullback is an algebra homomorphism,
by definition. The third step uses that contraction
with $v_5^{\mathrm{ex}_s}$ is a graded derivation of bi-degree
$(-1,\mathrm{even})$.
In the fourth step we commute
the differential in the second summand,
first with the pullback operation
(using that pullback is in fact a dg-algebra homomorphism, by definition)
and then with the contraction operation, using that the
corresponding Lie derivative vanishes, by
\eqref{LieDerivativeAlongExceptionalLiftOfv5VanishesOnHalfM5} in
Prop. \ref{LiftOfVectorField}.
In the fifth step we use that the differential
(commutes with pullback, as before, and) is
a graded derivation of degree $(1,\mathrm{even})$.
In the sixth step we realize the presence of the projection
$e^5 \wedge \iota_{v_5}^{\mathrm{ex}_s}$ according to
Lemma \ref{PullbackOfM2CocycleToHalfM5}, in view of
$d H_{\mathrm{ex}_s} = (\pi_{\mathrm{ex}_s})^\ast \mu_{{}_{\rm M2}}$
\eqref{Hexs}.
In the seventh step we again commute the differential
with pullback and with contraction, as before.
In the eighth step we again use the derivation property
of the differential to collect a total differential,
observing that $d e^5 = 0$ holds on the super-exceptional
$\tfrac{1}{2}\mathrm{M5}$-spacetime,
by Lemma \ref{VanishingSpinorPairingsOnHalfM5Spacetime}.
In the ninth step we collect terms and identify,
in the tenth step, the super-exceptional 2-flux and its dual,
from Def. \ref{SuperExceptionalPSLagrangian}.
Finally, in the last step we split off the factor of
2 again, just for emphasis.
\end{proof}

As a corollary we observe the following:

\begin{prop}[Super-exceptional YM-Lagrangian is closed and horizontal]
  \label{SuperExceptionalYangMillsLagrangianIsClosedAndHorizontal}
  The super-exceptional topological Yang-Mills Lagrangian
  (Def. \ref{SuperExceptionalTopologicalYM})
  is
  \vspace{-1mm}
\item  {\bf (i)} closed:
  $
    d \mathbf{L}^{\!\!\mathrm{tYM}}_{\mathrm{ex}_s}
    =
    0;
  $
\vspace{-1mm}
\item  {\bf (ii)} super-exceptionally horizontal (Def. \ref{SuperExceptionalHorizontalProjection}):
  $
    \iota_{v_5^{\mathrm{ex}_s}}
    \mathbf{L}^{\!\!\mathrm{tYM}}_{\mathrm{ex}_s}
    =
    0.
  $
\end{prop}
\begin{proof}
For the first statement, we may compute as follows:
 $$
     \begin{aligned}
       d \mathbf{L}^{\!\!\mathrm{tYM}}_{\mathrm{ex}_s}
       & =
       d
       \iota_{v_5}^{\mathrm{ex}_s} \mathbf{L}^{\!\!\mathrm{PS}}_{\mathrm{ex}_s}
       \\
       & =
       -
       \iota_{v_5}^{\mathrm{ex}_s}
       d
       \mathbf{L}^{\!\!\mathrm{PS}}_{\mathrm{ex}_s}
       \\
       & =
       \underset{
         = 0
       }{
         \underbrace{
           \iota_{v_5}^{\mathrm{ex}_s}
           \iota_{v_5}^{\mathrm{ex}_s}
         }
       }
       (i_{\mathrm{ex}_s})^\ast
       \mathbf{dL}^{\!\!\mathrm{WZ}}_{\mathrm{ex}_s}
       \\
       & = 0
       \,,
     \end{aligned}
 $$
 where we used first Lemma \ref{SuperExceptionaltYMIsContractionalOfPS},
 then \eqref{LieDerivativeAlongExceptionalLiftOfv5VanishesOnHalfM5}
 from Lemma \ref{LiftOfVectorField} and
 then Prop. \ref{TrivializationOf7CocycleOnExceptionalHalfM5},
 and finally, under the brace, we observe that contraction with
 elements in degree $(1,\mathrm{even})$ is nilpotent.
 This nilpotency also directly implies the second statement,
 by Lemma \ref{SuperExceptionaltYMIsContractionalOfPS}.
\end{proof}

\section{Super-exceptional equivariance along M-theory circle}
\label{SuperExceptionalReductionOfMTheoryCircle}

We show (Theorem \ref{SuperExceptionalPSEquivariantEnhancement} below)
that the super-exceptional Perry-Schwarz Lagrangian and the super-exceptional
topological Yang-Mills Lagrangian unify with the super-exceptional M5 WZ curvature
term into the Borel-equivariant enhancement of the super-exceptional M5-brane
cocycle with respect to the super-exceptional $S^1_{\mathrm{HW}}$ isometry
left-induced to an $\Omega S^2_{\mathrm{HW}}$-action
on the super-exceptional $\tfrac{1}{2}\mathrm{M5}$-spacetime
(Def. \ref{HomotopyQuotientOfSuperExceptionalHalfM5Spacetime}).
In order to put this in perspective, we first show
(Prop. \ref{RationalFibrationOf6dSuperspacetimeOver5d} below) that,
similarly, the little-string-extended $D=6$, $\mathcal{N} = (1,1)$,
superspacetime carries an $\Omega S^2$-action whose homotopy
quotient is the $D = 5$, $\mathcal{N} = 2$, superspacetime.

\medskip
To set the scene, we first recall how homotopy quotients are represented in
rational cohomology by (Borel-) equivariant de Rham cohomology.
From general homotopy theory we need the following
two basic facts (see \cite{NSS12}). For any kind of higher geometric spaces
(here, rational super spaces), we have:

\begin{enumerate}[{\bf (i)}]
  \vspace{-3mm}
  \item Forming based loop spaces is an equivalence
  from pointed connected spaces to $\infty$-groups,
  whose inverse is the classifying space construction
  \begin{equation}
    \label{LoopingAndDelooping}
    \xymatrix{
      \infty \mathrm{Groups}
      \ar@<+14pt>@{}[rrr]^-{ \Omega X\; \mapsfrom \;X   }
      \ar@<+6pt>@{<-}[rrr]^-{ \mbox{\color{blue} \tiny form loop space}  }_-{ \simeq }
      \ar@<-6pt>@{->}[rrr]_-{ \mbox{\color{blue} \tiny form classifying space} }
      \ar@<-14pt>@{}[rrr]_-{ G \;\mapsto  \; B G   }
      &&&
      \mathrm{Spaces}_{\!\!{\mathrm{pointed} \& \atop \mathrm{connected}}}
    }
  \end{equation}

\vspace{-3mm}
\item For $X$ a space and $G$ an $\infty$-group,
an $\infty$-action $\rho$ of $G$ on $X$
is equivalently a homotopy fiber sequence of the
following form
\begin{equation}
  \label{HomotopyQuotientFiberSequence}
  \xymatrix{
    X
    \ar[rr]^-{ \mathrm{hofib}( p_{\rho} ) }
    &&
    X \!\sslash\! G
    \ar[rr]^-{ p_{\rho} }
    &&
    B G
  },
\end{equation}
which then exhibits the space in the middle as the
homotopy quotient of $X$ by the $\infty$-action $\rho$.
\end{enumerate}

To prepare for Prop. \ref{RationalFibrationOf6dSuperspacetimeOver5d}
and Theorem \ref{SuperExceptionalPSEquivariantEnhancement}
below, we now consider a sequence of examples of homotopy quotients of
rational super spaces as in \eqref{HomotopyQuotientFiberSequence}.
In the following diagrams we always show the systems of spaces on the left
with their super dgc-algebra (FDA) models shown on the right. Throughout we use that
homotopy pullbacks of super spaces are modeled by pushouts of
semi-free super dgc-algebras (FDAs) as soon as the morphism
pushed out along is a cofibration in that it exhibits iterated addition of
generators. For more background see,
for instance,
\cite{Hess06}\cite{GrMo13}\cite{FSS16a}\cite{BSS18}.

\begin{example}[Rational $S^1$-equivariant cohomology and Cartan model]
\label{RationalS1EquivariantCohomologyAndCartanModel}
Let $X$ be any rational super-space of finite type,
hence $\mathrm{CE}\big( \mathfrak{l}X \big)$ any
finitely generated super-dgc algebra,
with differential to be denoted $d_X$,
and equipped with
a graded derivation $\iota_v$ of degree $(-1,\mathrm{even})$
such that the corresponding Lie derivative vanishes identically:
\begin{equation}
  \label{VanishinLiDerivativeOndgAlgebra}
  \xymatrix{
    \mathrm{CE}( \mathfrak{l}X)
    \ar[r]^-{ \iota_v }
    &
    \mathrm{CE}( \mathfrak{l}X)
  }
  \phantom{AAA}
  \mathcal{L}_v
  \;\coloneqq\;
  \big[
    d_{X}, \iota_v
  \big]
  \;=\; 0
  \;.
\end{equation}
The homotopy quotient
by the corresponding rational $S^1$-action
as in \eqref{HomotopyQuotientFiberSequence} is given by
\begin{equation}
  \label{CartanModelS1Quotient}
  \xymatrix@R=1.2em{
    X
    \ar[rr]^-{ \mathrm{hofib}(p_\rho) }
    &&
    X \sslash S^1
    \ar[dd]^-{ p_\rho }
    &&
    \mathrm{CE}
    (
      \mathfrak{l}X
    )
    \ar@{<-}[rr]^-{
      \mbox{
       \tiny
       $
       \begin{array}{lcl}
         0 & \!\!\! \mapsfrom \!\!\! & \omega_2
         \\
         \alpha & \!\!\! \mapsfrom \!\!\! & \alpha
       \end{array}
       $
      }
    }
    &&
    \mathrm{CE}
    (
      \mathfrak{l}X
    )
    [
      \omega_2
    ]
    \Big/
    \left(
      \!\!\!\!\!
      \mbox{
        \small
        $
        {\begin{array}{lcl}
          d\,\omega_2 & \!\!\! = \!\!\! & 0
          \\
          d\,\alpha
            & \!\!\! = \!\!\! &
          d_{{}_X} \alpha
            +
          \omega_2 \wedge \iota_{v} \alpha
        \end{array}}
        $
      }
      \!\!\!\!\!\!
    \right).
    \ar@{<-^{)}}[dd]^{
      \mbox{
        \tiny
        $
        \begin{array}{cc}
          \omega_2
          \\
          \mapsup
          \\
          \omega_2
        \end{array}
        $
      }
    }
    \\
    \\
    &&
    B S^1
    &&
    &&
    \mathbb{R}
    [
      \omega_2
    ]
    \Big/
    \left( \!\!\!\!\!
      \mbox{
        \small
        $
        {\begin{array}{rcl}
          d\,\omega_2 & \!\!\! = \!\!\! & 0
        \end{array}}
        $
      }
   \!\!\!\!\! \right)
  }
\end{equation}
This is the algebraic structure of the\emph{Cartan model}
for $G$-equivariant Borel cohomology
(see e.g. \cite[Sec. 5]{MathaiQuillen86} \cite{GSt}\cite{Mein}), here for $G = S^1$.
\end{example}

\begin{example}[Complex Hopf fibration]
\label{ComplexHopfFibration}
The complex Hopf fibration $h_{\mathbb{C}}$ realizes the 2-sphere
$S^2$ as the homotopy quotient of the 3-sphere by an action which is
classified by the canonical comparison map from $B \Omega S^2$ to $B S^1$:
\begin{equation}
\hspace{2.3cm}
  \label{CircleMapsToOmegaS2}
  \;\;\;\;\;\;\;\;\;\;
  \xymatrix@R=1.3em{
    S^3
    \ar[dd]^-{ h_{\mathbb{C}} }
    & &&
    \left(
      \!\!
      {\begin{array}{lcl}
        d\,\omega_3 &\!\!\!=\!\!\!& 0
      \end{array}}
      \!\!
    \right)
    \ar@{<-}[dd]^{
      \mbox{
        \tiny
        $
        \begin{array}{cc}
          \omega_2 & 0
          \\
          \mapsup & \mapsup
          \\
          \omega_2 & \omega_3
        \end{array}
        $
      }
    }
    \\
    \\
    \mathllap{
      S^3 \!\sslash\! S^1
      \simeq
      B\big(\Omega S^2 \big)
      \simeq
      \;
    }
    S^2
    \ar[rr]^-{c_1}
    &&
    B S^1\;,
       &
    \left(
      \!\!
      {\begin{array}{lcl}
        d\,\omega_2 &\!\!\!=\!\!\!& 0
        \\
        d\,\omega_3 &\!\!\!=\!\!\!& - \omega_2 \wedge \omega_2
      \end{array}}
      \!\!
    \right)
    \ar@{<-^{)}}[rr]^-{
      \mbox{
        \tiny
        $
        \omega_2 \mapsfrom \omega_2
        $
      }
    }
    &&
    \left(
      \!\!
      {\begin{array}{lcl}
        d\,\omega_2 &\!\!\!\!\!=\!\!\!\!\!& 0
      \end{array}}
      \!\!
    \right).
  }
\end{equation}
Notice that the classifying map $c_1$ here exhibits,
by \eqref{LoopingAndDelooping},
a canonical comparison homomorphism of $\infty$-groups
\begin{equation}
  \label{LoopsOfS2}
  \xymatrix{
    \Omega S^2
    \ar[rr]^-{ \Omega c_1 }
    &&
    S^1
  }.
\end{equation}
This loop $\infty$-group of the 2-sphere, $\Omega S^2$,
is an $\infty$-group
very similar to but just slightly richer than the plain circle.
\end{example}

\begin{example}[Left-induced $\Omega S^2$-equivariant cohomology]
\label{LeftInducedOmegaS2Action}
This means that an $\infty$-action
by $S^1$ as in Example \ref{RationalS1EquivariantCohomologyAndCartanModel}
left-induces an $\infty$-action by $\Omega S^2$,
with its homotopy quotient fiber sequence \eqref{HomotopyQuotientFiberSequence}
given by homotopy pullback along
\eqref{LoopsOfS2}, as shown in the following:
\begin{equation}
\footnotesize
  \label{OmegaS2Quotient}
  \hspace{-1mm}
  \xymatrix@R=1.2em@C=14pt{
    X
    \ar[r]^-{ {\mathrm{hofib}(p_{\Omega S^2})} \atop {\phantom{A}} }
    &
    X \!\sslash\! \Omega S^2
    \ar@{}[ddr]|{ \mbox{ \tiny (pb) } }
    \ar[r]_{\ }="s"
    \ar[dd]_-{ p_{\Omega S^2} }^<<<<<<<{\ }="t"
    &
    X \!\sslash\! S^1,
    \ar[dd]^{ p_{S^1}  }
    &
    \mathrm{CE}
    (
      \mathfrak{l}X
    )
    \ar@{<-}[r]^-{
      \mbox{
       \tiny
       $
       \begin{array}{lcl}
         0 & \!\!\!\!\!\!\!\!\!\!\!\! \mapsfrom \!\!\!\!\!\!\!\!\!\!\!\! & \omega_2
         \\
         0 & \!\!\!\!\!\!\!\!\!\!\!\! \mapsfrom \!\!\!\!\!\!\!\!\!\!\!\! & \omega_3
         \\
         \alpha & \!\!\!\!\!\!\!\!\!\!\!\! \mapsfrom \!\!\!\!\!\!\!\!\!\!\!\! & \alpha
         \\
         \\
       \end{array}
       $
      }
    }
    &
    {\color{blue}
    \mathrm{CE}
    (
      \mathfrak{l}X
    )
    \left[
      \!\!\!\!
      {\begin{array}{c}
      \omega_2,
      \\
      \omega_3
      \end{array}}
      \!\!\!\!
    \right]
    \!\!\Big/\!\!
    \left(\!\!\!\!\!\!\!
      \mbox{
        \small
        $
        {\begin{array}{lcl}
          d\,\omega_2 & \!\!\!\!\!\!\!  = \!\!\!\!\!\!\!\!\!  & 0
          \\
          d\,\omega_3 & \!\!\! \!\!\!\! = \!\!\!\!\!\!\!\!\!  & - \omega_2 \wedge \omega_2
          \\
          d\,\alpha
            & \!\!\! \!\!\!\!\!\! = \!\!\!\!\!\!\!\!\!  &
          d_{{}_X} \alpha
            +
          \omega_2 \wedge \iota_{v} \alpha
        \end{array}}
        $
    }
    \!\!\!\!\!\!\!\! \right)\!
    \ar@{<-}[dd]^{
      \mbox{
        \tiny
        $
        \begin{array}{cc}
          \omega_2 & \omega_3
          \\
          \mapsup & \mapsup
          \\
          \omega_2 & \omega_3
        \end{array}
        $
      }
    }
    }
    \ar@{<-}[r]
    \ar@{}[ddr]|-{ \mbox{ \tiny (po) } }
    &
    \mathrm{CE}
    (
      \mathfrak{l}X
    )
    [\omega_2]
    \!\!\Big/\!\!
    \left(\!\!\!\!\!\!\!
      \mbox{
        \small
        $
        {\begin{array}{lcl}
          d\,\omega_2 & \!\!\!\!\!\!\!\!  = \!\!\!\!\!\!\!\!\!  & 0
          \\
          d\,\alpha
            & \!\!\!\!\!\!\!\!  = \!\!\!\!\!\!\!\!\!  &
          d_{{}_X} \alpha
            +
          \omega_2 \wedge \iota_{v} \alpha
        \end{array}}
        $
      }
    \!\!\!\!\!\!\!\!\! \right)\!
    \ar@{<-}[dd]^{
      \mbox{
        \tiny
        $
        \begin{array}{c}
          \omega_2
          \\
          \mapsup
          \\
          \omega_2
        \end{array}
        $
      }
    }
    \\
    \\
    &
    B(\Omega S^2)
    \ar[r]
    &
    B S^1
    &
    &
    \mathbb{R}
    [
      \omega_2, \omega_3
    ]
    \Big/
    \left(\!\!\!\!\!
      \mbox{
        \small
        $
        {\begin{array}{rcl}
          d\,\omega_2 & \!\!\! = \!\!\! & 0
          \\
          d\,\omega_3 & \!\!\! = \!\!\! & - \omega_2 \wedge \omega_2
        \end{array}}
        $
      }
   \!\!\!\!\!
   \right)
   \ar@{<-^{)}}[r]_-{
     \mbox{
       \tiny
       $
       \begin{array}{lcl}
         \omega_2 & \!\!\!\!\!\!\mapsfrom\!\!\!\!\!\! & \omega_2
       \end{array}
       $
     }
   }
   &
   ~ \mathbb{R}
    [
      \omega_2
    ]
    \big/
    \big(\!\!\!\!\!\!
      \mbox{
        \small
       $
        {\begin{array}{rcl}
          d\,\omega_2 & \!\!\! = \!\!\! & 0
        \end{array}}
        $
       }
   \!\!\!\!\!\!
   \big)
   %
  }
\end{equation}
The resulting dgc-algebra,
shown in blue,
is much like the Cartan model for $S^1$-equivariant cohomology
as in Example \ref{RationalS1EquivariantCohomologyAndCartanModel},
except that here all even powers of the
generator $\omega_2$ in bi-degree $(2,\mathrm{even})$,
which classifies the circle action,
are trivialized in cohomology, by
the new generator $\omega_3$ in bi-degree $(3,\mathrm{even})$.
\end{example}

\begin{lemma}[$\Omega S^2$-equivariant cocycles]
  \label{OmegaS2EquivariantCocycles}
  Rational cocycles on a homotopy quotient $X \sslash \Omega S^2$
  for an $\Omega S^2$-action that is left-induced
  according to Example \ref{LeftInducedOmegaS2Action}
  from an $S^1$-action as in Example \ref{RationalS1EquivariantCohomologyAndCartanModel}
  are, if they are at most linear in the generator $\omega_2$,
  precisely given by pairs consisting of a cocycle
  $\alpha$ on $X$ and a trivialization $\beta$
  of its contraction with $v$:
  \begin{equation}
    \label{FormulaOmega2EquivariantCocycles}
    d
    \big(
      \,
      \alpha
      \;-\;
      \omega_2 \wedge \beta
      \;-\;
      \omega_3 \wedge \iota_v \beta
      \,
    \big)
    \;=\;
    0
    \phantom{AAAA}
    \mbox{if $\;d_X \alpha = 0\;$ and $\;d_X \beta = \iota_v \alpha\;$. }
  \end{equation}
\end{lemma}
\begin{proof}
To see that cocycles of this form are closed,
we compute as follows, directly unwinding the definitions,
where the three lines correspond to application of the three summands
of the differential
$d = d_{{}_{X}} + \omega_2 \wedge \iota_v + d_{{}_{S^2}}$:
$$
  \begin{aligned}
    d
    \big(
    \alpha
    \;-\;
    \omega_2 \wedge \beta
    \;-\;
    \omega_3 \wedge \iota_v \beta
    \big)
    & =
    \underset{
      = 0
    }{
      \underbrace{
        d_X \alpha
      }
    }
    -
    \omega_2 \wedge
    \underset{
      = \iota_v \alpha
    }{
      \underbrace{
        d_X \beta
      }
    }
    -
    \omega_3 \wedge
    \underset{
      \mathclap{
        = \iota_v \iota_v \alpha = 0
      }
    }{
      \underbrace{
        \iota_v d_X \beta
      }
    }
    \\
    &
    \phantom{=} \;
    +
    \omega_2 \wedge \iota_v \alpha
    \;-\;
    \omega_2 \wedge \omega_2 \wedge \iota_v \beta
    \;+\;
    \omega_3 \wedge \omega_2 \wedge
    \underset{
     = 0
    }{
      \underbrace{
        \iota_v \iota_v \beta
      }
    }
    \\
    &
    \phantom{=} \;
    +
    \omega_2 \wedge \omega_2 \wedge \iota_v \beta
    \\
    & = 0\;.
  \end{aligned}
$$
Conversely, reading this same equation as a condition
for the vanishing of the coefficients of the products
of generators shows that every cocycle of the form on the left
of \eqref{FormulaOmega2EquivariantCocycles} satisfies
the conditions shown on the right.
\end{proof}

\medskip

We next observe, in Prop \ref{RationalFibrationOf6dSuperspacetimeOver5d}
below, that an $\infty$-action by the $\infty$-group
$\Omega S^2$ \eqref{LoopsOfS2} characterizes the
$D = 6$, $\mathcal{N} = (1,1)$,
super-spacetime fibered over the $D = 5$, $\mathcal{N} =2$, superspacetime
and exhibits the little-string cocycle in $D = 6$
as coming from a 2-sphere-valued super-cocycle in $D = 5$
(hence in rational Cohomotopy in degree 2).
To put this in perspective, we first recall
the analogous situation in $D = 11$:
\begin{example}[Rational sphere-valued cocycles for M-branes in $D =11$]
  We have the situation in \eqref{11dexts} below:
  \begin{enumerate}[{\bf (i)}]
  \vspace{-2mm}
  \item {\cite[Rem. 4.4, Prop. 4.5]{FSS13}}:
  The $D = 11$ $\mathcal{N} = 1$
  super Minkowski spacetime (as in Def. \ref{ExceptionalTangentSuperSpacetime})
  is the rational $S^1$-extension of
  the $D = 10$ $\mathcal{N} = (1,1)$ (i.e. type IIA)
  superspacetime, classified by the D0-brane cocycle $\mu_{{}_{\rm D0}}$.
    \vspace{-2mm}
  \item
  \cite[Sec. 2.5]{Sati13}\cite{FSS15}\cite[Sec. 2]{FSS16a}:
  On that  extension, the super M2/M5-brane cochains
  $(\mu_{{}_{\rm M2}}, 2 \mu_{{}_{\rm M5}})$ \eqref{TheMBraneCocycles}
  constitute a super rational $S^4$-valued cocycle,
  i.e., a cocycle in rational Cohomotopy in degree 4.
  \end{enumerate}
  \begin{equation}
    \label{11dexts}
    \xymatrix{
      \mathfrak{m}5\mathfrak{brane}
      \ar[rr]_>>>>>>{\ }="s1"
      \ar[dd]_-{ \mathrm{hofib}(\mu_{{}_{\rm M2}},2 \mu_{{}_{\rm M5}}) }
      ^>>>>>>{\ }="t1"
      \ar@{}[ddrr]|-{
        \mbox{
          \tiny
          (pb)
        }
      }
      &&
      \ast
      \ar[dd]
      &&
      \hspace{-1.2cm}
      \underset{
        a \in \{0,1,2,3,4,5',{\color{blue}5},6,7,8,9\}
      }{
      \left(
        \!\!\!\!
        {\begin{array}{lcl}
          d\,\psi^\alpha & \!\!\!\! = \!\!\!\!  & 0
          \\
          d\,e^{a} & \!\!\!\! = \!\!\!\! & \overline{\psi}\;\Gamma^a\psi
          \\
          d\,h_3 & \!\!\!\! = \!\!\!\! & \mu_{{}_{\rm M2}}
          \\
          d\,h_6 & \!\!\!\! = \!\!\!\! &
            \mu_{{}_{\rm M5}} +\tfrac{1}{2}h_3 \wedge \mu_{{}_{\rm M2}}
        \end{array}}
        \!\!\!\!
      \right)
      }
      \ar@{<-}[rr]^-{
        \mbox{
          \tiny
          $
          \begin{array}{lcl}
            h_3 & \!\!\!\!\!\!\!\!\!\!\!\!\!\!\! \mapsfrom \!\!\!\!\!\!\!\!\!\!\!\! & h_3
            \\
            h_6 & \!\!\!\!\!\!\!\!\!\!\!\!\!\!\! \mapsfrom \!\!\!\!\!\!\!\!\!\!\!\! & h_6
            \\
            \mu_{{}_{\rm M2}} & \!\!\!\!\!\!\!\!\!\!\!\!\!\!\! \mapsfrom \!\!\!\!\!\!\!\!\!\!\!\! & \omega_4
            \\
            2\mu_{{}_{\rm M5}} & \!\!\!\!\!\!\!\!\!\!\!\!\!\!\! \mapsfrom \!\!\!\!\!\!\!\!\!\!\!\! & \omega_7
          \end{array}
          $
        }
      }
      \ar@{<-}[dd]_-{
        \mbox{
          \tiny
          $
          \begin{array}{cc}
            \psi^\alpha & e^a
            \\
            \mapsup & \mapsup
            \\
            \psi^\alpha & e^a
          \end{array}
          $
        }
      }
      \ar@{}[ddrr]|-{ \mbox{ \tiny (po) } }
      &&
      \left(
        \!\!\!\!
        {\begin{array}{lcl}
          d\,h_3 & \!\!\!\! = \!\!\!\! & \mu_{{}_{\rm M2}}
          \\
          d\,h_6 & \!\!\!\! = \!\!\!\! &
            \mu_{{}_{\rm M5}} + \tfrac{1}{2} h_3 \wedge \mu_{{}_{\rm M2}}
          \\
          d\,\omega_4 & \!\!\!\! = \!\!\!\! & 0
          \\
          d\,\omega_7 & \!\!\!\! = \!\!\!\! & - \omega_4 \wedge \omega_4
        \end{array}}
        \!\!\!\!\!\!\!
      \right)
      \ar@{<-^{)}}[dd]^-{
        \mbox{
          \tiny
          $
          \begin{array}{cc}
            \omega_4 & \omega_7
            \\
            \mapsup & \mapsup
            \\
            \omega_4 & \omega_7
          \end{array}
          $
        }
      }
      \\
      \\
      \mathbb{T}^{10,1\vert \mathbf{32}}
      \ar[dd]_-{ \mathrm{hofib}(\mu_{{}_{\rm D0}}) }
      \ar[rr]_-{ (\mu_{{}_{\rm M2}}, 2\mu_{{}_{\rm M5}}) }
      &&
      S^4_{\mathbb{R}}
      &&
      \underset{
        a \in \{0,1,2,3,4,5',{\color{blue}5},6,7,8,9\}
      }{
      \left(
        \!\!\!\!
        {\begin{array}{lcl}
          d\,\psi^\alpha & \!\!\!\! = \!\!\!\!  & 0
          \\
          d\,e^{a} & \!\!\!\! = \!\!\!\! & \overline{\psi}\;\Gamma^a\psi
        \end{array}}
        \!\!\!\!
      \right)
      }
      \ar@{<-}[rr]_-{
        \mbox{
          \tiny
          $
          \begin{array}{lcl}
            \mu_{{}_{\rm M2}} & \!\!\!\!\!\!\!\!\!\!\!\!   \mapsfrom \!\!\!\! \!\!\!\!\!\!\!\!  & \omega_4
            \\
            2\mu_{{}_{\rm M5}} & \!\!\!\!\!\!\!\!\!\!\!\!   \mapsfrom \!\!\!\! \!\!\!\!\!\!\!\!  & \omega_7
          \end{array}
          $
        }
      }
      \ar@{<-}[dd]^-{
        \mbox{
          \tiny
          $
          \begin{array}{cc}
            \psi^\alpha & e^a
            \\
            \mapsup & \mapsup
            \\
            \psi^\alpha & e^a
          \end{array}
          $
        }
      }
      &&
      \left(
        \!\!\!\!
        {\begin{array}{lcl}
          d\,\omega_4 & \!\!\!\! = \!\!\!\!  & 0
          \\
          d\,\omega_7 & \!\!\!\! = \!\!\!\! & - \omega_4 \wedge \omega_4
        \end{array}}
        \!\!\!\!
      \right)
      \\
      \\
      \mathbb{T}^{9,1\vert \mathbf{16} + \overline{\mathbf{16}}}
      \ar@{->}[rr]_-{
        \mu_{{}_{\rm D0}}
        =
        \overline{\psi}\; \Gamma^5 \psi      }
      &&
      B S^1_{\mathbb{R}}
      &&
      \underset{
        a \in \{0,1,2,3,4,5',6,7,8,9\}
      }{
      \left(
        \!\!\!\!
        {\begin{array}{lcl}
          d\,\psi^\alpha & \!\!\!\! = \!\!\!\!  & 0
          \\
          d\,e^{a} & \!\!\!\! = \!\!\!\! & \overline{\psi}\;\Gamma^a\psi
        \end{array}}
        \!\!\!\!
      \right)
      }
      \ar@{<-}[rr]_-{
        \overline{\psi}\;\Gamma^5\psi
        \;\mapsfrom \;
        \omega_2
      }
      &&
      \left(
        \!\!\!\!
        {\begin{array}{lcl}
          d\,\omega_2 & \!\!\!\! = \!\!\!\!  & 0
        \end{array}}
        \!\!\!\!
      \right)
      \ar@{=>} "s1"; "t1"
    }
  \end{equation}
  Notice that, when expressed in the canonical super coordinate functions
  as in \eqref{LeftInvariantVielbeingInTermsOfCanonicalCoordinates},
  the bi-fermionic component of the potential $e^5$
  for the D0-brane cocycle is
  identified as the \emph{super 1-form Ramond-Ramond (RR) potential} $C_1$,
  i.e., the \emph{graviphoton} of KK-reduction on $S^1_{\mathrm{HW}}$
  (e.g. \cite[below (51)]{APPS97b}):
  \begin{equation}
    \label{RRPotential}
    e^5
    \;=\;
    d x^5
    \;+\;
    \underset{
      \eqqcolon C_1
    }{
      \underbrace{
        \overline{\theta} \Gamma^5 d\theta
      }
    }
    \;\;\overset{ d }{\longmapsto}\;\;
    (\overline{d\theta}) \Gamma^5 d\theta
    \;=\;
    \mu_{{}_{\rm D0}}.
  \end{equation}
\end{example}

This is analogous to the following situation:

\begin{prop}[Fibration of $D = 6$, $\mathcal{N} = (1,1)$,
super-spacetime over $D=5$ super-spacetime]
  \label{RationalFibrationOf6dSuperspacetimeOver5d}
  The $D = 6$, $\mathcal{N} = (1,1)$, super-spacetime
  is a rational circle fibration over the $D = 5$, $\mathcal{N}= 2$,
  super-spacetime, which comes via homotopy pullback from the complex
  Hopf fibration (Example \ref{ComplexHopfFibration}).
  Moreover, the super-cocycle $\mu_{{}_{\rm L1}}$
  for the little string in 6d is induced by a 2-sphere-valued super-cocycle
  in $D = 5$
  $$
    \xymatrix{
      \mathbb{R}^{4,1 \vert \mathbf{8} + \mathbf{8}}
      \ar[rrr]^-{ ( \overline{\psi}\;\Gamma_5\psi, \; \mu^{5d}_{{}_{\rm L1}}) }
      &&&
      S^2_{\mathbb{R}}
    },
  $$
  realizing the little-string extended $D = 6$
  superspacetime
  as a rational $\Omega S^2$-fibration over $D =5$ super-spacetime:
  \begin{equation}
  \label{ltlstringHomotopyPullbackDiagram}
  \hspace{-.3cm}
    \xymatrix
    {
      \mathfrak{ltl}\mathfrak{string}
      \ar[dd]^>>>>>>{\ }="t1"
      \ar[rr]_>>>>>>{\ }="s1"
      \ar@{}[ddrr]|-{ \mbox{\rm \tiny (pb)} }
      &&
      \ast
      \ar[dd]
      &&
      \;\;\;\;
      \left(
      \!\!\!\!\!
      \mbox{
        \small
        $
        \begin{array}{lcl}
          d\,\psi^\alpha & \!\!\!\!\!\! =\!\!\!\!\!\! & 0
          \\
          d\,e^{a \leq 5} & \!\!\!\!\!\!=\!\!\!\!\!\! & \overline{\psi}\; \Gamma^a\psi
          \\
          d\,f_2 & \!\!\!\!\!\!=\!\!\!\!\!\! & \mu_{{}_{\rm L1}}
        \end{array}
        $
      }
      \!\!\!\!\!
      \right)
      \ar@{<-}[dd]^-{
        \mbox{
          \tiny
          $
          \begin{array}{cc}
            \psi^\alpha & e^{a \leq 4}
            \\
            \mapsup & \mapsup
            \\
            \psi^\alpha & e^{a \leq 4}
          \end{array}
          $
        }
      }
      \ar@{<-}[rrr]^-{
        \mbox{
          \tiny
          $
          \begin{array}{ccl}
            e^5
              &\!\!\!\!\!\!\!\!\mapsfrom\!\!\!\!\!\!\!\!\!\!\!\!&
            e
            \\
            \overline{\psi}\; \Gamma^5\psi
              &\!\!\!\!\!\!\!\!\mapsfrom\!\!\!\!\!\!\!\!\!\!\!\!&
            \omega_2
            \\
            \underoverset{a = 0}{4}{\sum}
              (\overline{\psi}\;\Gamma_a\psi) \wedge e^a
              &\!\!\!\!\!\!\!\!\mapsfrom\!\!\!\!\!\!\!\!\!\!\!\!&
            \omega_3
            \\
            f_2
              &\!\!\!\!\!\!\!\!\mapsfrom\!\!\!\!\!\!\!\!\!\!\!\!&
            f_2
          \end{array}
          $
        }
      }
      \ar@{}[ddrrr]|{
        \mbox{\rm
          \tiny (po)
        }
      }
      &
      {\phantom{AA}}
      &&
      \left(
      \!\!\!\!\!
      \mbox{
        \small
        $
        \begin{array}{lcl}
          d\, e & \!\!\!\!\!\!=\!\!\!\!\!\! & \omega_2
          \\
          d\,\omega_2 & \!\!\!\!\!\!=\!\!\!\!\!\! & 0
          \\
          d\,\omega_3 & \!\!\!\!\!\!=\!\!\!\!\!\! & -\omega_2 \wedge \omega_2
          \\
          d\,f_2 & \!\!\!\!\!\!=\!\!\!\!\!\! & \omega_3 + \omega_2 \wedge e
        \end{array}
        $
      }
      \!\!\!\!\!
      \right)
      \ar@{<-^{)}}[dd]^-{
        \mbox{
          \tiny
          $
          \begin{array}{ccc}
            e & \omega_2 & \omega_3
            \\
            \mapsup & \mapsup & \mapsup
            \\
            e & \omega_2 & \omega_3
          \end{array}
          $
        }
      }
      \\
      \\
      \mathbb{R}^{5,1\vert \mathbf{8} + \mathbf{8}}
      \ar[dd]^>>>>>>{\ }="t"
      \ar[rr]^{
        \mu_{{}_{\rm L1}}
      }_>>>>>>{\ }="s"
      \ar@{}[ddrr]|{ \mbox{\rm \tiny (pb) } }
      &&
      S^3_{\mathbb{R}}
      \ar[dd]^-{ h_{\mathbb{C}} }
      &&
      \;\;\;\;
      \left(
      \!\!\!\!\!
      \mbox{
        \small
        $
        \begin{array}{lcl}
          d\,\psi^\alpha & \!\!\!=\!\!\! & 0
          \\
          d\,e^{a \leq 5} & \!\!\!=\!\!\! & \overline{\psi}\; \Gamma^a\psi
        \end{array}
        $
      }
      \!\!\!\!\!
      \right)
      \ar@{<-}[dd]^-{
        \mbox{
          \tiny
          $
          \begin{array}{cc}
            \psi^\alpha & e^{a \leq 4}
            \\
            \mapsup & \mapsup
            \\
            \psi^\alpha & e^{a \leq 4}
          \end{array}
          $
        }
      }
      \ar@{<-}[rrr]|-{
        \mbox{
          \tiny
          $
          \begin{array}{ccl}
            e^5
              &\!\!\!\!\!\!\!\!\!\!\!\!\!\mapsfrom\!\!\!\!\!\!\!\!\!\!\!\!\!&
            e
            \\
            \overline{\psi}\; \Gamma^5\psi
              &\!\!\!\!\!\!\!\!\!\!\!\!\!\mapsfrom\!\!\!\!\!\!\!\!\!\!\!\!\!&
            \omega_2
            \\
            \underoverset{a = 0}{4}{\sum}
              (\overline{\psi}\;\Gamma_a\psi) \wedge e^a
              &\!\!\!\!\!\!\!\!\!\!\!\!\!\mapsfrom\!\!\!\!\!\!\!\!\!\!\!\!\!&
            \omega_3
            \\
          \end{array}
          $
        }
      }
      \ar@{}[ddrrr]|{
        \mbox{\rm
          \tiny (po)
        }
      }
      &
      {\phantom{AA}}
      &&
      \left(
      \!\!\!\!\!
      \mbox{
        \small
        $
        \begin{array}{lcl}
          d\, e & \!\!\!=\!\!\! & \omega_2
          \\
          d\,\omega_2 & \!\!\!=\!\!\! & 0
          \\
          d\,\omega_3 & \!\!\!=\!\!\! & -\omega_2 \wedge \omega_2
        \end{array}
        $
      }
      \!\!\!\!\!
      \right)
      \ar@{<-^{)}}[dd]^-{
        \mbox{
          \tiny
          $
          \begin{array}{cc}
            \omega_2 & \omega_3
            \\
            \mapsup & \mapsup
            \\
            \omega_2 & \omega_3
          \end{array}
          $
        }
      }
      \\
      \\
      \mathbb{R}^{4,1\vert \mathbf{8}+ \mathbf{8}}
      \ar[rr]_{ ( \overline{\psi}\Gamma_5\psi,  \mu^{5d}_{{}_{\rm L1}}) }
      &&
      S^2_{\mathbb{R}}
      \mathrlap{\; \simeq B \Omega S^2_{\mathbb{R}}\;,}
      &&
      \;\;\;\;
      \left(
      \!\!\!\!\!
      \mbox{
        \small
        $
        \begin{array}{lcl}
          d\,\psi^\alpha & \!\!\!=\!\!\! & 0
          \\
          d\,e^{a \leq 4} & \!\!\!=\!\!\! & \overline{\psi}\; \Gamma^a\psi
        \end{array}
        $
      }
      \!\!\!\!\!
      \right)
      \ar@{<-}[rrr]_-{
        \mbox{
          \tiny
          $
          \begin{array}{ccl}
            \overline{\psi}\; \Gamma^5\psi
              &\!\!\!\!\!\!\!\!\!\!\!\!\mapsfrom\!\!\!\!\!\!\!\!\!\!\!\!&
            \omega_2
            \\
            \underset{
              \eqqcolon \mu^{5d}_{{}_{\rm L1}}
            }{
            \underbrace{
              \underoverset{a = 0}{4}{\sum}
                (\overline{\psi}\;\Gamma_a\psi) \wedge e^a
            }
            }s
              &\!\!\!\!\!\!\!\!\!\!\!\!\mapsfrom\!\!\!\!\!\!\!\!\!\!\!\!&
            \omega_3
          \end{array}
          $
        }
      }
      &&&
      \left(
      \!\!\!\!\!
      \mbox{
        \small
        $
        \begin{array}{lcl}
          d\,\omega_2 & \!\!\!=\!\!\! & 0
          \\
          d\,\omega_3 & \!\!\!=\!\!\! & -\omega_2 \wedge \omega_2
        \end{array}
        $
      }
      \!\!\!\!\!
      \right).
      \ar@{=>} "s"; "t"
      \ar@{=>} "s1"; "t1"
    }
  \end{equation}
\end{prop}
\begin{proof}
  The statement comes down to proving that
  $$
    d \mu_{{}_{\rm L1}}
    \;=\;
    -
    \big( \overline{\psi} \;\Gamma_5 \psi \big)
    \big( \overline{\psi} \;\Gamma_5 \psi \big),
  $$
  hence that
  \begin{equation}
    \label{5dFI}
    \underoverset{a = 0}{4}{\sum}
    \big( \overline{\psi} \;\Gamma_a \psi \big)
    \big( \overline{\psi}\; \Gamma^a \psi \big)
    \;=\;
    -
    \big( \overline{\psi} \;\Gamma_5 \psi \big)
    \big( \overline{\psi} \;\Gamma_5 \psi \big)
    \,.
  \end{equation}
  This is indeed a Fierz identity satisfied by spinors in $D= 5$ (e.g. \cite[III.5.50]{CDF91}).
  Equivalently, from the perspective of
  $D = 6$, it reflects the fact that there is
  the super-cocycle for the little-string (e.g. \cite[III.7.14]{CDF91})
  \begin{equation}
    \label{LittleStringCocycle}
    \mu_{{}_{\rm L1}}
    \;\coloneqq\;
    \underoverset{a = 0}{5}{\sum}
    \big(
      \overline{\psi}
      \;\Gamma_a \psi
    \big)
      \wedge
    e^a
    \;=\;
    \underoverset{a = 0}{4}{\sum}
    \big(
      \overline{\psi}
      \;
      \Gamma_a \psi
    \big)
      \wedge
    e^a
    \;+\;
    \big(
      \overline{\psi}
      \;
      \Gamma_5 \psi
    \big)
      \wedge
    e^5
  \end{equation}
  and  that \eqref{5dFI} is equivalently its closure condition,
  written with the 5th spatial coordinate summand separated off:
  $$
    \begin{aligned}
      0
      & =
      d \mu_{{}_{\rm L1}}
      \\
      & =
      \underoverset{a=0}{5}{\sum}
        \big(\overline{\psi}\;\Gamma_a\psi\big) d e^a
      \\
      & =
      \underoverset{a = 0}{5}{\sum}
      \big( \overline{\psi}\; \Gamma_a \psi \big)
      \big( \overline{\psi} \;\Gamma^a \psi \big)
      \\
      & =
      \underoverset{a = 0}{4}{\sum}
      \big( \overline{\psi} \;\Gamma_a \psi \big)
      \big( \overline{\psi} \;\Gamma^a \psi \big)
      +
      \big( \overline{\psi}\; \Gamma_5 \psi \big)
      \big( \overline{\psi} \; \Gamma^5 \psi \big).
    \end{aligned}
  $$

  \vspace{-7mm}
\end{proof}

\begin{remark}[$D=5$ super spacetime as homotopy $\Omega S^2$-quotient of $D=6$]

\vspace{-2mm}
 \item {\bf (i)}  Conversely, the total rectangle on the left of \eqref{ltlstringHomotopyPullbackDiagram}
  exhibits, by \eqref{HomotopyQuotientFiberSequence},
  the $D=5$ super-spacetime
  $\mathbb{R}^{4,1\vert \mathbf{8}+\mathbf{8}}$
  as the homotopy $\Omega S^2$-quotient of
  the little string-extended $D = 6$, $\mathcal{N} = 2$, super-spacetime:

  $$
    \mathbb{T}^{4,1\vert \mathbf{8} + \mathbf{8}}
    \;\simeq\;
    \mathbb{T}^{5,1\vert \mathbf{8} + \mathbf{8}} \!\sslash\! S^1
    \;\simeq\;
    \mathfrak{ltl}\mathfrak{string} \!\sslash\! \Omega S^2
    \,.
  $$

  \item {\bf (ii)} To make this explicit
  as an $\Omega S^2$-quotient in the form
  of \eqref{OmegaS2Quotient},
  we may (co)fibrantly resolve
  the bottom instead of the right morphism of the
  rectangle in \eqref{ltlstringHomotopyPullbackDiagram}.
  This yields the following:
  \begin{equation}
  \label{OtherltlstringHomotopy}
  \hspace{-.2cm}
    \xymatrix
    @C=1em
    {
      \mathfrak{ltl}\mathfrak{string}
      \ar[dd]^>>>>>>{\ }="t1"
      \ar[rr]_>>>>>>{\ }="s1"
      \ar@{}[ddrr]|-{ \mbox{\rm \tiny (pb)} }
      &&
      \ast
      \ar[dd]
      &&
      \;\;\;\;
      \left(
      \!\!\!\!\!
      \mbox{
        \small
        $
        \begin{array}{lcl}
          d\,\psi^\alpha & \!\!\!=\!\!\! & 0
          \\
          d\,e^{a \leq 5} & \!\!\!=\!\!\! & \overline{\psi}\; \Gamma^5\psi
          \\
          d\,\theta_2 & \!\!\!=\!\!\! & \mu_{{}_{\rm L1}}
        \end{array}
        $
      }
      \!\!\!\!\!
      \right)
      \ar@{<-}[dd]|-{
        \mbox{
          \tiny
          $
          \begin{array}{ccccccc}
            \psi^\alpha & e^{a \leq 5} & f_2 & 0 & 0
            \\
            \mapsup & \mapsup & \mapsup & \mapsup & \mapsup
            \\
            \psi^\alpha & e^{a \leq 5} & f_2 & \omega_2 & \omega_3
          \end{array}
          $
        }
      }
      \ar@{<-}[rrr]^-{
        \mbox{
          \tiny
          $
          \begin{array}{lcl}
          \end{array}
          $
        }
      }
      \ar@{}[ddrrr]|{
        \mbox{
          \tiny (po)
        }
      }
      &
      {\phantom{AA}}
      &&
      \mathbb{R}
      \ar@{<-}[dd]^-{
        \mbox{
          \tiny
          $
          \begin{array}{cc}
            0 & 0
            \\
            \mapsup & \mapsup
            \\
            \omega_2 & \omega_3
          \end{array}
          $
        }
      }
      \\
      \\
      \mathbb{R}^{4,1\vert \mathbf{8}+ \mathbf{8}}
      \ar[rr]_{ }
      &&
      S^2_{\mathbb{R}}
      \mathrlap{\; \simeq B \Omega S^2_{\mathbb{R}}\;,}
     \;\;\;\;\;\;\;\;\;\;\; &&
      \;\;\;\;
      \left(
      \!\!\!\!\!
      \mbox{
        \small
        $
        \begin{array}{lcl}
          d\,\psi^\alpha & \!\!\!=\!\!\! & 0
          \\
          d\,e^{a \leq 4} & \!\!\!=\!\!\! & \overline{\psi}\; \Gamma^a\psi
          \\
          d\,e^5
            & \!\!\!=\!\!\! &
            \overline{\psi}\; \Gamma^a\psi
            +
            \omega_2
          \\
          d\,f_2 & \!\!\!=\!\!\! &
            \mu_{{}_{\rm L1}}
            +
            (\omega_3 + e \wedge \omega_2)
          \\
          d\,\omega_2
            & \!\!\!=\!\!\! &
            0
          \\
          d\,\omega_3
            & \!\!\!=\!\!\! &
            - \omega_2 \wedge \omega_2
        \end{array}
        $
      }
      \!\!\!\!\!
      \right)
      \ar@{<-^{)}}[rrr]_-{
        \mbox{
          \tiny
          $
          \begin{array}{lcl}
            \omega_2
              &\!\!\!\!\!\!\!\!\mapsfrom\!\!\!\!\!\!\!\!&
            \omega_2
            \\
            \omega_3
              &\!\!\!\!\!\!\!\!\mapsfrom\!\!\!\!\!\!\!\!&
            \omega_3
          \end{array}
          $
        }
      }
      &&&
      \left(
      \!\!\!\!\!
      \mbox{
        \small
        $
        \begin{array}{lcl}
          d\,\omega_2 & \!\!\!=\!\!\! & 0
          \\
          d\,\omega_3 & \!\!\!=\!\!\! & -\omega_2 \wedge \omega_2
        \end{array}
        $
      }
      \!\!\!\!\!
      \right).
      %
      \ar@{=>} "s1"; "t1"
    }
  \end{equation}
\end{remark}

In conclusion, $D = 6$, $\mathcal{N} = (1,1)$,
superspacetime extended by its brane content
reduces to $D = 5$ super-spacetime by enhancing the
$\infty$-action by $S^1$ to an $\infty$-action by
$\Omega S^2$ \eqref{LoopsOfS2}.

\medskip
This provides the motivation to similarly form the homotopy quotient
of the super-exceptional $\mathcal{N} = (1,0)$ spacetime
$  \big(
    \mathbb{R}^{5,1\vert \mathbf{8}}
      \times
    \mathbb{R}^1
  \big)_{\mathrm{ex}_s}
$
from Def. \ref{HalfM5LocusAndItsExceptionalTangentBundle}
not just by the $S^1_{\mathrm{HW}}$-action of flowing along the M-theory circle fiber
along the super-exceptional isometry $v_5^{\mathrm{ex}_s}$
from Def. \ref{LiftOfVectorField},
but the left-induced $\Omega S^2$-action,
via Example \ref{LeftInducedOmegaS2Action}
To indicate this,
we will write $\Omega S^2_{\mathrm{HW}}$
to denote this $\infty$-group with that $\infty$-action understood,
hence with the comparison \eqref{CircleMapsToOmegaS2}
specifically being
\begin{equation}
  \label{OmegaS2HW}
  \xymatrix{
    \Omega S^2_{\mathrm{HW}}
    \ar[r]
    &
    S^1_{\mathrm{HW}}
    \,.
  }
\end{equation}
Hence:
\begin{defn}[Homotopy $\Omega S^2_{\mathrm{HW}}$-quotient of super-exceptional
$\tfrac{1}{2}\mathrm{M5}$ along $S^1_{\mathrm{HW}}$]
\label{HomotopyQuotientOfSuperExceptionalHalfM5Spacetime}
Write
$
    \big(
      \mathbb{R}^{5,1\vert \mathbf{8}}
      \times
      \mathbb{R}^1
    \big)_{\mathrm{ex}_s}
    \!\sslash\! \Omega S^2_{\mathrm{HW}}
$
for the homotopy quotient of the super-exceptional
$\tfrac{1}{2}\mathrm{M5}$ spacetime \eqref{HalfM5LocusAndItsExceptionalTangentBundle}
by the rational $\Omega S^2_{\mathrm{HW}}$-action
which is left-induced, via Example \ref{LeftInducedOmegaS2Action},
by the rational
$S^1_{\mathrm{HW}}$-action given by the
the super-exceptional $S^1_{\mathrm{HW}}$-flow \eqref{LiftedVectorField}.
Hence, with \eqref{OmegaS2Quotient}, the defining
super dgc-algebra (FDA) is that on the right of the
following diagram:
\begin{equation}
  \label{SuperExceptionalHalfM5QuotientedByOmegaS2}
  \xymatrix@R=-2pt{
    \big(
      \mathbb{R}^{5,1\vert \mathbf{8}}
      \times
      \mathbb{R}^1
    \big)_{\mathrm{ex}_s}
    \ar[rr]^-{ q_{{}_{\Omega S^2_{\mathrm{HW}}}} }
    &&
    \big(
      \mathbb{R}^{5,1\vert \mathbf{8}}
      \times
      \mathbb{R}^1
    \big)_{\mathrm{ex}_s}
    \sslash \Omega S^2_{\mathrm{HW}}
    \\
    \\
    \mathrm{CE}
    \big(
      (
        \mathbb{R}^{5,1\vert \mathbf{8}}
        \times
        \mathbb{R}^1
      )_{\mathrm{ex}_s}
    \big)
    \ar@{<-}[rr]^-{
      \mbox{
       \tiny
       $
       \begin{array}{lcl}
         0 & \!\!\!\!\!\!\! \mapsfrom \!\!\!\!\!\!\! & \omega_2
         \\
         0 & \!\!\!\!\!\!\! \mapsfrom \!\!\!\!\!\!\! & \omega_3
         \\
         \alpha & \!\!\!\!\!\!\! \mapsfrom \!\!\!\!\!\!\! & \alpha
       \end{array}
       $
      }
    }
    &&
    \mathrm{CE}
    \big(
    (
      \mathbb{R}^{5,1\vert \mathbf{8}}
      \times
      \mathbb{R}^1
    )_{\mathrm{ex}_s}
    \big)
    [
      \omega_2, \omega_3
    ]
    \Big/
    \left(
      \!\!\!\!
      \mbox{
        \small
        $
        {\begin{array}{lcl}
          d\,\omega_2 & \!\!\! = \!\!\! & 0
          \\
          d\,\omega_3 & \!\!\! = \!\!\! & - \omega_2 \wedge \omega_2
          \\
          d\,\alpha
            & \!\!\! = \!\!\! &
          d_{{}_{\sfrac{1}{2}\mathrm{M5}}} \alpha
            +
          \omega_2 \wedge \iota_{v_5^{\mathrm{ex}_s}} \alpha
        \end{array}}
        $
      }
      \!\!\!\!\!\!
    \right)
    \,,
  }
\end{equation}
where on the right
$\alpha
  \in
  \mathrm{CE}\big(
    (
      \mathbb{R}^{5,1\vert \mathbf{8}}
      \times
      \mathbb{R}^1
    )_{\mathrm{ex}_s}
  \big)$
  is any element in the CE-algebra of the
  super-exceptional $\tfrac{1}{2}\mathrm{M5}$-spacetime
  (Def. \ref{HalfM5LocusAndItsExceptionalTangentBundle}),
  $\iota_{v_5^{\mathrm{ex}_s}}$ is
  contraction with the super-exceptional isometry \eqref{LiftedVectorField},
  and
  $d_{{}_{\sfrac{1}{2}\mathrm{M5}}}$
  now denotes the
  differential on that algebra, in contrast to the new
  differential $d$ defined above.
\end{defn}

Now we may state and prove the main statement of this section:

\begin{theorem}[Super-exceptional $\Omega S^2_{\mathrm{HW}}$-equivariant M5-cocycle]
\label{SuperExceptionalPSEquivariantEnhancement}
The super-exceptional Perry-Schwarz Lagrangian
$\mathbf{L}^{\!\mathrm{PS}}_{\mathrm{ex}_s}$ \eqref{LagrangianPSExceptional}
and the super-exceptional topological Yang-Mills Lagrangian
$\mathrm{L}^{\!\mathrm{tYM}}_{\mathrm{ex}_s}$ \eqref{SuperExceptionalThetaAngleTerm}
are the components that enhance
the super-exceptional M5-brane cocycle
$\mathbf{dL}^{\!\!\mathrm{WZ}}_{\mathrm{ex}}$ \eqref{SuperExceptionalM5Cocycle}
restricted along the embedding $i_{\mathrm{ex}_s}$ of the super-exceptional $\tfrac{1}{2}\mathrm{M5}$ spacetime \eqref{M5FixedLocusNormallyEnhanced}
to an equivariant
cocycle with respect to the
$\Omega S^2_{\mathrm{HW}}$-action \eqref{OmegaS2HW},
hence to a cocycle
on the homotopy $\Omega S^2_{\mathrm{HW}}$-quotient \eqref{SuperExceptionalHalfM5QuotientedByOmegaS2}
of the super-exceptional $\tfrac{1}{2}\mathrm{M5}$-spacetime,
as follows:
\begin{equation}
    \label{EquivariantSuperExceptionalM5BraneCocycle}
    \big(
      (i_{\mathrm{ex}_s})^\ast
      \mathbf{dL}^{\!\!\mathrm{WZ}}_{\mathrm{ex}_s}
    \big)_{
      \!\sslash \Omega S^2_{\mathrm{HW}}
    }
    \;\coloneqq\;
    \Big(
    \,
    (i_{\mathrm{ex}_s})^\ast
    \mathbf{dL}^{\!\!\mathrm{WZ}}_{\mathrm{ex}_s}
    \;-\;
    \omega_2
    \wedge
    \mathbf{L}^{\!\!\mathrm{PS}}_{\mathrm{ex}_s}
    \;-\;
    \omega_3
    \wedge
    \mathbf{L}^{\mathrm{tYM}}_{\mathrm{ex}_s}
    \,
    \Big)
    \;\;\in\;\;
  \mathrm{CE}
  \Big(
    \big(
      \mathbb{R}^{5,1\vert \mathbf{8}}
      \times
      \mathbb{R}^1
    \big)_{\mathrm{ex}_s}
    \!\sslash\!
    \Omega S^2_{\mathrm{HW}}
  \Big)
\end{equation}
in that
\begin{enumerate}[{\bf (i)}]
\vspace{-2mm}
\item this is indeed a cocycle with respect to the differential
$
  d
    =
    d_{\sfrac{1}{2}\mathrm{M5}}
    + \omega_2 \wedge \iota_{v_5^{\mathrm{ex}_s}}
    + d_{S^2}
$
from \eqref{SuperExceptionalHalfM5QuotientedByOmegaS2}:
\begin{equation}
  \label{dOfEquivariant7Cocycle}
  d
  \Big(
    \big(
    (i_{\mathrm{ex}_s})^\ast
    \mathbf{dL}^{\!\!\mathrm{WZ}}_{\mathrm{ex}_s}
    \big)_{\!\sslash \Omega S^2_{\mathrm{HW}}}
  \Big)
  \;=\;
  0
  \;\;\in\;\;
  \mathrm{CE}
  \Big(
    \big(
      \mathbb{R}^{5,1\vert \mathbf{8}}
      \times
      \mathbb{R}^1
    \big)_{\mathrm{ex}_s}
    \!\sslash\!
    \Omega S^2_{\mathrm{HW}}
  \Big);
\end{equation}

\vspace{-4mm}
\item it does enhance the super-exceptional M5-brane cocycle,
in that it extends it through
the homotopy quotient projection $q_{{}_{\Omega S^2_{\mathrm{HW}}}}$
\eqref{SuperExceptionalHalfM5QuotientedByOmegaS2}:
\vspace{-4mm}
\begin{equation}
  \label{OmegaS2EquivariantEnhancement}
  \hspace{-2mm}
  \xymatrix@R=-1pt@C=3em{
    \big(
      \mathbb{R}^{5,1\vert \mathbf{8}}
      \times
      \mathbb{R}^1
    \big)_{\mathrm{ex}_s}
    \!\sslash\!
    \Omega S^2_{\mathrm{HW}}
    &&
    \big(
      \mathbb{R}^{5,1\vert \mathbf{8}}
      \times
      \mathbb{R}^1
    \big)_{\mathrm{ex}_s}
    \ar@{^{(}->}[rr]^-{ i_{\mathrm{ex}_s} }
    \ar[ll]_-{ q_{{}_{\Omega S^2_{\mathrm{HW}}}} }
    &&
    \big(\mathbb{R}^{10,1\vert \mathbf{32}}\big)_{\mathrm{ex}_s}.
    \\
    \big(
      (i_{\mathrm{ex}_s})^\ast
      \mathbf{dL}^{\!\!\mathrm{WZ}}_{\mathrm{ex}_s}
    \big)_{
      \!\sslash \Omega S^2_{\mathrm{HW}}
    }
    &&
    (i_{\mathrm{ex}_s})^\ast\mathbf{dL}^{\!\!\mathrm{WZ}}_{\mathrm{ex}_s}
    \ar@{<-|}[rr]
    \ar@{<-|}[ll]_-{ }
    &&
    \mathbf{dL}^{\!\!\mathrm{WZ}}_{\mathrm{ex}_s}
  }
\end{equation}
\end{enumerate}
\end{theorem}
\begin{proof}
  By Lemma \ref{OmegaS2EquivariantCocycles}
  the first claim
  equation \eqref{dOfEquivariant7Cocycle}, is equivalent to the two statements
  \begin{enumerate}
  \vspace{-2mm}
    \item
    $
    d
    \mathbf{L}^{\!\!\mathrm{PS}}_{\mathrm{ex}_s}
    \;=\;
    \iota_{v_5^{\mathrm{ex}_s}}
    (i_{\mathrm{ex}_s})^\ast
    \mathbf{dL}^{\!\!\mathrm{WZ}}_{\mathrm{ex}_s}
    \,.
    $
    \vspace{-2mm}
    \item
    $
    \iota_{v_5^{\mathrm{ex}_s}}
    \mathbf{L}^{\!\mathrm{PS}}_{\mathrm{ex}_s}
    \;=\;
    \mathbf{L}^{\!\mathrm{tYM}}_{\mathrm{ex}_s}
    \,.
    $
  \end{enumerate}
  \vspace{-2mm}
  The first of these is the content of Prop. \ref{TrivializationOf7CocycleOnExceptionalHalfM5},
  while the second is Lemma \ref{SuperExceptionaltYMIsContractionalOfPS}.

  The second claim \eqref{OmegaS2EquivariantEnhancement} is
  immediate from the nature of the map \eqref{SuperExceptionalHalfM5QuotientedByOmegaS2}.
\end{proof}

Before moving on, we record the following further
properties of the compactified super-exceptional
$\tfrac{1}{2}\mathrm{M5}$-spacetime:

\begin{prop}[Equivariant closedness of tYM]
  \label{EquivariantClosedtYM}
  The super-exceptional topological Yang-Mills Lagrangian
  (Def. \ref{SuperExceptionalTopologicalYM})
  is closed on the homotopy $\Omega S^2_{\mathrm{HW}}$-quotient
  of the super-exceptional $\tfrac{1}{2}\mathrm{M5}$-spacetime
  (Def. \ref{HomotopyQuotientOfSuperExceptionalHalfM5Spacetime}):
 \begin{equation}
   \label{topologicalYMTermIsEquivariantlyClosed}
     d \mathbf{L}^{\!\!\mathrm{tYM}}_{\mathrm{ex}_s}
     \;=\;
     0
   \phantom{AAA}
   \in\;
   \mathrm{CE}
   \Big(
   \Big(
     \big(
       \mathbb{R}^{5,1\vert \mathbf{8}}
       \times
       \mathbb{R}^1
     \big)_{\mathrm{ex}_s}
   \Big)
   _{\sslash \Omega S^2_{\mathrm{HW}}}
   \Big).
 \end{equation}
\end{prop}
\begin{proof}
  By definition \eqref{SuperExceptionalHalfM5QuotientedByOmegaS2}
  of the equivariant differential,
  we have
  $$
     d \mathbf{L}^{\!\!\mathrm{tYM}}_{\mathrm{ex}_s}
     \;\coloneqq\;
     d_{\sfrac{1}{2}\mathrm{M5}} \mathbf{L}^{\!\!\mathrm{tYM}}_{\mathrm{ex}_s}
     \;+\;
     \omega_2 \wedge \iota_{v_5^{\mathrm{ex}_s}}
     \mathbf{L}^{\!\!\mathrm{tYM}}_{\mathrm{ex}_s}  .
  $$
  By
  Prop. \ref{SuperExceptionalYangMillsLagrangianIsClosedAndHorizontal},
  both summands here already vanish separately.
\end{proof}

\begin{prop}[Closedness in the homotopy quotient]
 \label{Closede5omega2Plusomega3}
 On the homotopy $\Omega S^2_{\mathrm{HW}}$-quotient
 of the super-exceptional $\tfrac{1}{2}\mathrm{M5}$-brane spacetime
 (Def. \ref{HomotopyQuotientOfSuperExceptionalHalfM5Spacetime}),
 we have
 \begin{equation}
   \label{omega2IsExactOnHalfM5HomotopyQuotent}
   d\,e^5 \;=\; \omega_2
  \quad
  \text{and}
  \quad
   d
   \big(
     e^5 \wedge \omega_2
     +
     \omega_3
   \big)
   \;=\;
   0
   \phantom{AAA}
   \in\;
   \mathrm{CE}
   \Big(
   \Big(
     \big(
       \mathbb{R}^{5,1\vert \mathbf{8}}
       \times
       \mathbb{R}^1
     \big)_{\mathrm{ex}_s}
   \Big)
   _{\sslash \Omega S^2_{\mathrm{HW}}}
   \Big).
 \end{equation}
 \end{prop}
 \begin{proof}
  For the first statement we compute as follows:
  \begin{equation}
    \begin{aligned}
      d e^5
      & \coloneqq
      (
      d_{\sfrac{1}{2}\mathrm{M5}}
        +
      \omega_2 \wedge \iota_{v_5^{\mathrm{ex}_s}}
      )
      e^5
      \\
      & =
      \underset{
        \mathclap{
          = (\overline{P \psi})\Gamma^5 (P \psi)
          = 0
        }
      }{
      \underbrace{
        d_{{}_{\sfrac{1}{2}\mathrm{M5}}} e^5
      }
      }
      \;+\;
      \omega_2 \wedge
      \underset{
        \mathclap{
          = \delta_5^5 = 1
        }
      }{
        \underbrace{
          \iota_{v_5^{\mathrm{ex}_s}} e^5
        }
      }
      \;=\;
      \omega_2
    \,.
    \end{aligned}
  \end{equation}
  Here the first step is the definition of the
  equivariant differential \eqref{SuperExceptionalHalfM5QuotientedByOmegaS2},
  while the second step uses the definition of
  $d_{\sfrac{1}{2}}\mathrm{M5}$ from \eqref{CESuperExceptionalMK6-part2}.
  Then under the first brace we used
  Lemma \ref{VanishingSpinorPairingsOnHalfM5Spacetime},
  and under the second brace we used
  the definition \eqref{LiftedVectorField}
  in Prop. \ref{LiftOfVectorField}.
  The second statement is directly implied by the first
  and by the differential relations $d \omega_2 = 0$ and $d \omega_3 = - \omega_2 \wedge \omega_2$
  from \eqref{SuperExceptionalHalfM5QuotientedByOmegaS2}.
 \end{proof}

\section{Super-exceptional M5 Lagrangian from super-exceptional embedding}
\label{SuperExceptionalM5Lagrangian}

Finally we discuss the super-exceptional embedding construction
of the M5-brane Lagrangian.
We consider the super-exceptional Nambu-Goto Lagrangian for the
$\tfrac{1}{2}\mathrm{M5}$-brane (Def. \ref{NG}) below
and prove (Corollary \ref{M5LagrangianIsRelativeTrivialization} below)
that the sum of super-exceptional Nambu-Goto Lagrangian
and the super-exceptional Perry-Schwarz Lagrangian
arise as relative trivialization of
the super-exceptional M5-brane cocycle
along the super-exceptional embedding
of the $\tfrac{1}{2}\mathrm{M5}$-brane.
We go further and and consider
the $\Omega S^2_{\mathrm{HW}}$-equivariant enhancement of this
statement, corresponding to KK-compactification to the D4-brane,
and prove
(Theorem \ref{TheTrivialization} below)
that this corrects the relative trivialization by
a summand proportional to the
super-exceptional topological Yang-Mills term
that constitutes the D4-brane WZ term
(Remark \ref{D4WZ} below).
While a priori this further summand is exact only after
compactification on $S^1_{\mathrm{HW}}$,
as befits the nature of the D4 arising form the M5,
we observe (Remark \ref{EquivariantRelativeTrivialization} below)
that this D4 term, too, does become genuinely exact
after a natural completion of
the $\Omega S^2_{\mathrm{HW}}$-action on the
super-exceptional $\tfrac{1}{2}\mathrm{M5}$ spacetime.

\medskip

In direct generalization of the super Nambu-Goto Lagrangian
\eqref{SuperVolumeFormIn3d}, we set:
\begin{defn}
  \label{NG}
  The \emph{super-exceptional Nambu-Goto Lagrangian} for the
  $\tfrac{1}{2}\mathrm{M5}$-brane (Remark \ref{TheHalfM5Spacetime})
  is the super volume form of the $\tfrac{1}{2}\mathrm{M5}$-locus,
  hence, with \eqref{LeftInvariantVielbeingInTermsOfCanonicalCoordinates},
  the left-invariant completion of the bosonic volume form
  under translational supersymmetry, hence
  is the element
  \begin{equation}
    \label{M5SuperVolumeForm}
    \begin{aligned}
    \mathbf{L}^{\!\!\mathrm{NG}}_{\mathrm{ex}_s}
    \;\coloneqq\;
    \mathrm{svol}^{5+1}_{\mathrm{ex}_s}
    & \coloneqq
    e_0
      \wedge
    e_1
      \wedge
    e_2
      \wedge
    e_3
      \wedge
    e_4
      \wedge
    e_{5'}
    \\
    & =
    (\pi_{\mathrm{ex}_s})^\ast
    \big(
    e_0
      \wedge
    e_1
      \wedge
    e_2
      \wedge
    e_3
      \wedge
    e_4
      \wedge
    e_{5'}
    \big)
    \;\in\;
    \mathrm{CE}
    \big(
      \mathbb{R}^{5,1\vert \mathbf{8}}
      \times
      \mathbb{R}^1
    \big)_{\mathrm{ex}_s}
    \end{aligned}
  \end{equation}
  in the CE-algebra of the super-exceptional
  $\tfrac{1}{2}\mathrm{M5}$-spacetime (Def. \ref{HalfM5LocusAndItsExceptionalTangentBundle}.)
\end{defn}
\begin{remark}[$S^1_{\rm B}$ vs. $S^1_{\rm HW}$ directions]
Beware that the last wedge factor in \eqref{M5SuperVolumeForm}
is $e_{5'}$ and that $e_5$ does not appear
(see Remark \ref{TheHalfM5Spacetime} for discussion of the $5$-$5'$-plane).
Mathematically, this comes out from the MO9-projection in the last step
\eqref{abcd} in the proof of Theorem \ref{TheTrivialization} below.
\end{remark}

\begin{theorem}[Full M5-Lagrangian from equivariant super-embedding]
\label{TheTrivialization}
The super-exceptional $\Omega S^2_{\mathrm{HW}}$-equivariant
M5-brane cocycle
from Theorem \ref{SuperExceptionalPSEquivariantEnhancement}
is equal to
\begin{equation}
  \label{HorizontalizedEquivariant7Cocycle}
  \Big(
    (i_{\mathrm{ex}_s})^\ast
    \mathbf{dL}^{\!\!\mathrm{WZ}}_{\mathrm{ex}_s}
  \Big)_{\!\!\sslash \Omega S^2_{\mathrm{HW}}}
  \;=\;
  d
  \big(
    \mathbf{L}^{\!\!\mathrm{NG}}_{\mathrm{ex}_s}
    +
    \mathbf{L}^{\!\!\mathrm{PS}}_{\mathrm{ex}_s}
    \wedge
    e^5
  \big)
    -
  \underset{
    =
    e^5
      \wedge
    d
    \big(
      C_1 \wedge \mathbf{L}^{\!\!\mathrm{tYM}}_{\mathrm{ex}_s}
    \big)
    +
    \omega_3 \wedge \mathbf{L}^{\!\!\mathrm{tYM}}_{\mathrm{ex}_s}
  }{
    \mathclap{\phantom{A \atop A}}
    \underbrace{
     (e^5 \wedge \omega_2 + \omega_3)
     \wedge
     \mathbf{L}^{\!\!\mathrm{tYM}}_{\mathrm{ex}_s}
    }
  }
  \,,
\end{equation}
where under the brace we show an equivalent re-formulation
in terms of the element $C_1$ from \eqref{RRPotential}.
\end{theorem}
\begin{proof}
  We may rewrite
  \eqref{EquivariantSuperExceptionalM5BraneCocycle}
  as follows:
  \begin{equation}
    \label{aaa}
    \begin{aligned}
      \Big(
        (i_{\mathrm{ex}_s})^\ast
        \mathbf{dL}^{\!\!\mathrm{WZ}}_{\mathrm{ex}_s}
      \Big)_{\!\sslash \Omega S^2_{\mathrm{HW}}}
      & =
      (i_{\mathrm{ex}_s})^\ast
      \mathbf{dL}^{\!\!\mathrm{WZ}}_{\mathrm{ex}_s}
      -
      \underset{\mathclap {\tiny
        = \omega_2
     } }{
        \underbrace{
          (d e^5)
        }
      }
      \wedge
      \mathbf{L}^{\!\!\mathrm{PS}}_{\mathrm{ex}_s}
      -
      \omega_3 \wedge \mathbf{L}^{\!\!\mathrm{tYM}}_{\mathrm{ex}_s}
      \\
      & =
      (i_{\mathrm{ex}_s})^\ast
      \mathbf{dL}^{\!\!\mathrm{WZ}}_{\mathrm{ex}_s}
      -
      d
      \big(
        e^5
        \wedge
        \mathbf{L}^{\!\!\mathrm{PS}}_{\mathrm{ex}_s}
      \big)
      -
      e^5
      \wedge
      \underset{
        \mathclap{\footnotesize
          \begin{aligned}
          & =
          (
            d_{\sfrac{1}{2}\mathrm{M5}}
            +
            \omega_2 \wedge \iota_{v_5}^{\mathrm{ex}_s}
          )
          \mathbf{L}^{\!\!\mathrm{PS}}_{\mathrm{ex}_s}
          \\
          & =
          \iota_{v_5^{\mathrm{ex}_s}}
            (i_{\mathrm{ex}_s})^\ast
            \mathbf{dL}^{\!\!\mathrm{WZ}}_{\mathrm{ex}_s}
           +
           \omega^2 \wedge \mathbf{L}^{\!\!\mathrm{tYM}}_{\mathrm{ex}_s}
          \end{aligned}
        }
      }{
        \underbrace{
          d \mathbf{L}^{\!\!\mathrm{PS}}_{\mathrm{ex}_s}
        }
      }
      -
      \omega_3 \wedge \mathbf{L}^{\!\!\mathrm{tYM}}_{\mathrm{ex}_s}
      \\
      & =
      \underset{
        = (-)^{\mathrm{hor}_{\mathrm{ex}_s}}
      }{
      \underbrace{
        \big(
          \mathrm{id}
          -
          e^5
          \wedge
          \iota_{v_5^{\mathrm{ex}_s}}
        \big)
      }
      }
      \big(
        (i_{\mathrm{ex}_s})^\ast
        \mathbf{dL}^{\!\!\mathrm{WZ}}_{\mathrm{ex}_s}
      \big)
      \;+\;
      d
      \big(
        \mathbf{L}^{\!\!\mathrm{PS}}_{\mathrm{ex}_s}
        \wedge
        e^5
      \big)
      -
      (\omega_3 + e^5 \wedge \omega_2)
       \wedge
      \mathbf{L}^{\!\!\mathrm{tYM}}_{\mathrm{ex}_s}.
    \end{aligned}
  \end{equation}
  Here the first line is the definition \eqref{EquivariantSuperExceptionalM5BraneCocycle}
  with the observation \eqref{omega2IsExactOnHalfM5HomotopyQuotent}
  inserted, as shown under the brace. Then, in the first step,
  we use that the differential is a derivation of bi-degree $(1,\mathrm{even})$
  and under the brace we unwind the definition of the equivariant differential
  \eqref{SuperExceptionalHalfM5QuotientedByOmegaS2} and then used
  Prop. \ref{TrivializationOf7CocycleOnExceptionalHalfM5}
  and
  Lemma \ref{SuperExceptionaltYMIsContractionalOfPS}.
  In the last step we collect terms and identify
  under the brace the
  super-exceptional horizontal projection
  from Def. \ref{SuperExceptionalHorizontalProjection}.
  This means we are now reduced to showing that the first summand
  in the last line of \eqref{aaa} is
  \begin{equation}
    \big(
      (i_{\mathrm{ex}_s})^\ast
      \mathbf{dL}^{\!\!\mathrm{WZ}}_{\mathrm{ex}_s}
    \big)^{\mathrm{hor}_{\mathrm{ex}_s}}
    \;=\;
    d \mathbf{L}^{\!\!\mathrm{NG}}_{\mathrm{ex}_s}.
  \end{equation}
  We compute as follows:
  \begin{equation}
    \label{abcd}
    \begin{aligned}
      & \big(
        (i_{\mathrm{ex}_s})^\ast
        \mathbf{dL}^{\!\!\mathrm{WZ}}_{\mathrm{ex}_s}
      \big)^{\mathrm{hor}_{\mathrm{ex}_s}}
      \\
      & =
      \big(
        \mathrm{id}
        -
        e^5
        \wedge
        \iota_{v_5^{\mathrm{ex}_s}}
      \big)
      (i_{\mathrm{ex}_s})^\ast
      \big(
        (\pi_{\mathrm{ex}_s})^\ast
        \mu_{{}_{\rm M5}}
        +
        \tfrac{1}{2}
        H_{\mathrm{ex}_s}
        \wedge
        d H_{\mathrm{ex}_s}
      \big)
      \\
      & =
      \big(
        \mathrm{id}
        -
        e^5
        \wedge
        \iota_{v_5^{\mathrm{ex}_s}}
      \big)
      (i_{\mathrm{ex}_s})^\ast
      (\pi_{\mathrm{ex}_s})^\ast
      \mu_{{}_{\rm M5}}
      +
        \tfrac{1}{2}
      \big(
        \mathrm{id}
        -
        e^5
        \wedge
        \iota_{v_5^{\mathrm{ex}_s}}
      \big)
        \big(
          (i_{\mathrm{ex}_s})^\ast
          H_{\mathrm{ex}_s}
        \big)
        \wedge
        \big(
          (i_{\mathrm{ex}_s})^\ast
          d H_{\mathrm{ex}_s}
        \big)
       \\
       & =
       \tfrac{1}{5!} \;\;\;\;
       \underset{\tiny
         \mathclap{
           a_i
           \in
           \{0,1,2,3,4,5'\}
         }
       }{\sum}
       \;\;\;\;
       \big(
       \overline{(P \psi)}
         \Gamma_{a_1 \cdots a_5}
       (P \psi)
       \big)
       \wedge
       e_{a_0}
       \wedge
       \cdots
       \wedge
       e_{a_5}
      +
      \underset{
        = 0
      }{
      \underbrace{
        \tfrac{1}{2}
      \big(
        \mathrm{id}
        -
        e^5
        \wedge
        \iota_{v_5^{\mathrm{ex}_s}}
      \big)
        \big(
          (i_{\mathrm{ex}_s})^\ast
          H_{\mathrm{ex}_s}
        \big)
        \wedge
        \big(
          e^5 \wedge
          \iota_{v_5}^{\mathrm{ex}_s}
          (i_{\mathrm{ex}_s})^\ast
          d H_{\mathrm{ex}_s}
        \big)
      }
      }
      \\
      & =
         d( e_0 \wedge e_1 \wedge e_2 \wedge e_3 \wedge e_4 \wedge e_{5'} )\;.
    \end{aligned}
  \end{equation}
  Here the first step is unwinding the definitions. The second step is multiplying
  out and using, in the second summand,
  the fact that pullback is an algebra homomorphism.
  The third step observes that the first term
  is just those summands of $\mu_{{}_{\rm M5}}$ \eqref{TheMBraneCocycles}
  whose indices are along the $\tfrac{1}{2}\mathrm{M5}$-locus,
  hence in $\{0,1,2,3,4,5'\}$, while in the second term
  we realize the presence of the projection
$e^5 \wedge \iota_{v_5}^{\mathrm{ex}_s}$ according to
Lemma \ref{PullbackOfM2CocycleToHalfM5}, in view of
$d H_{\mathrm{ex}_s} = (\pi_{\mathrm{ex}_s})^\ast \mu_{{}_{\rm M2}}$
from \eqref{Hexs}.
 With the projection operator up front, this makes the
 second term vanish, as shown under the brace.
 The last step is \cite[Lemma 6.9]{HSS18}.
This establishes the first line in
 \eqref{HorizontalizedEquivariant7Cocycle}.

 Finally, to see equality to the expression shown in
 \eqref{HorizontalizedEquivariant7Cocycle} under the brace,
 observe that
 \begin{equation}
   \label{DifferentialOfe5}
   \begin{aligned}
     d e^5
       =
     d(d x^5 + C_1)
       =
     d C_1
     \;,
   \end{aligned}
 \end{equation}
 by \eqref{RRPotential}.
  Using this, the first summand over the brace in
 \eqref{HorizontalizedEquivariant7Cocycle} becomes
 $$
   \begin{aligned}
     e^5
      \wedge
     \omega_2
       \wedge
     \mathbf{L}^{\!\!\mathrm{tYM}}_{\mathrm{ex}_s}
     & =
     e^5
      \wedge
     (d e^5)
       \wedge
     \mathbf{L}^{\!\!\mathrm{tYM}}_{\mathrm{ex}_s}
     \\
     & =
     e^5
      \wedge
     (d C_1)
       \wedge
     \mathbf{L}^{\!\!\mathrm{tYM}}_{\mathrm{ex}_s}
     \\
     & =
     e^5
      \wedge\
      d
     \big(
       C_1
         \wedge
       \mathbf{L}^{\!\!\mathrm{tYM}}_{\mathrm{ex}_s}
     \big)\,,
   \end{aligned}
 $$
 where the first step is \eqref{omega2IsExactOnHalfM5HomotopyQuotent},
 the second step is \eqref{DifferentialOfe5} and the third step is
 \eqref{topologicalYMTermIsEquivariantlyClosed}. This, therefore,
 establishes also the identification under the brace in
 \eqref{HorizontalizedEquivariant7Cocycle}.
\end{proof}

\begin{cor}[M5-Lagrangian is relative trivialization
    along super-exceptional embedding]
  \label{M5LagrangianIsRelativeTrivialization}
  The super-exceptional M5-brane cocycle
  $\mathbf{dL}^{\!\!\mathrm{WZ}}_{\mathrm{ex}_s}$
  (Def. \ref{ExceptionalM5SuperCocycle})
  becomes exact when restricted
  along the super-exceptional embedding $(i_{\mathrm{ex}_s})$
  of the $\tfrac{1}{2}\mathrm{M5}$-brane
  (Lemma \ref{SuperExceptionalEmbeddings}),
  trivialized there by the sum of the
  super-exceptional Nambu-Goto Lagrangian
  $\mathbf{L}^{\!\!\mathrm{NG}}_{\mathrm{ex}_s}$
  (Def. \ref{NG}) and the
  super-exceptional Perry-Schwarz Lagrangian
  $\mathbf{L}^{\!\!\mathrm{PS}}_{\mathrm{ex}_s}$
  (Def. \ref{SuperExceptionalPSLagrangian}):
  \begin{equation}
    \label{RelativeTrivialization}
    (i_{\mathrm{ex}_s})^\ast
    \mathbf{dL}^{\!\!\mathrm{WZ}}_{\mathrm{ex}_s}
    \;=\;
    d
    \,
    \big(
      \mathbf{L}^{\!\!\mathrm{WZ}}_{\mathrm{ex}_s}
      +
      \mathbf{L}^{\!\!\mathrm{PS}}_{\mathrm{ex}_s}
      \wedge e^5
    \big)
    \,.
  \end{equation}
\end{cor}
\begin{proof}
  The statement is the first component
  (the one independent of the equivariance generators
  $\omega_2$ and $\omega_3$)
  of the equivariant statement in Theorem \ref{TheTrivialization}.
  More formally, equation \eqref{RelativeTrivialization}
  is the pullback of \eqref{LieDerivativeAlongv5}
  along the homotopy quotient homomorphism
  $q_{{}_{\Omega S^2_{\mathrm{HW}}}}$
  \eqref{SuperExceptionalHalfM5QuotientedByOmegaS2}.
\end{proof}

\begin{remark}[Dimensional reduction to WZ-term of D4-brane]
  \label{D4WZ}
  Upon compactification on $S^1_{\mathrm{HW}}$,
  the last summand of \eqref{HorizontalizedEquivariant7Cocycle}
  manifestly gives the WZ-term
  $$
    \mathbf{L}^{\!\!\mathrm{WZ}}_{\mathrm{D4}}
    \;=\;
    C_1 \wedge F \wedge F
  $$
  of the D4-brane
  (\cite[(7.4)]{CvGNSW97}\cite[(51)]{APPS97b}\cite[6.1]{CAIB00};
   see \cite[4.3]{FSS13}\cite[4]{FSS16b}).
  From the point of view of the Yang-Mills theory on the brane,
  this identifies $C_1$ with the theta-angle (e.g. \cite[(3.1)]{L19}),
  matching the last summand in \eqref{EquivariantTrivialization}
  below.
\end{remark}

\begin{remark}[Exact $\Omega S^2_{\mathrm{HW}}$-equivariant super-exceptional M5-Lagrangians]
\label{EquivariantRelativeTrivialization}
Recall that
the first summand in \eqref{TheTrivialization} is
exact, implying the super-exceptional embedding construction
before compactification (Corollary \ref{M5LagrangianIsRelativeTrivialization}), while the second
summand in \eqref{TheTrivialization}
\begin{equation}
  \label{SecondSummandInEquivariuantSuperExceptionalM5Cocycle}
  \big( e^5 \wedge \omega_2 + \omega_3\big)
  \wedge
  \big(
    F_{\mathrm{ex}_s}
    \wedge
    F_{\mathrm{ex}_s}
  \big)
  \;\;\;\;\;\;\in\;
  \mathrm{CE}
  \Big(
  \Big(
  \big(
    \mathbb{R}^{5,1\vert \mathbf{8}}
    \times
    \mathbb{R}^1
  \big)_{\mathrm{ex}_s}
  \Big)_{\sslash \Omega S^2_{\mathrm{HW}}}
  \Big)
\end{equation}
becomes exact only
after dimensional reduction, implying the super-exceptional
embedding construction of the D4 WZ-term
(Remark \ref{D4WZ}).
It is therefore natural to ask for a pullback of the
situation to a richer extended super-spacetime on which
also the second summand \eqref{SecondSummandInEquivariuantSuperExceptionalM5Cocycle},
and hence the full
$\Omega S^2_{\mathrm{HW}}$-equivariant
super-embedded super-exceptional M5-brane cocycle
from \eqref{TheTrivialization}, become exact,
$\Omega S^2_{\mathrm{HW}}$-equivariantly.
Since \eqref{SecondSummandInEquivariuantSuperExceptionalM5Cocycle}
is the wedge product of two equivariantly closed
terms, by Prop. \ref{EquivariantClosedtYM} and Prop. \ref{Closede5omega2Plusomega3},
there are two canonical possibilities here,
by enforcing trivialization of the first or the second factor.
We will now briefly comment on both of these:
$$
  \underset{
    \mbox{
      \tiny
      \begin{tabular}{c}
        discussed in
        Remark \ref{Theta2}
      \end{tabular}
    }
  }{
    d\,\theta_2
    \;=\;
    e^5 \wedge \omega_2 + \omega_3
  }
    \phantom{AAA}
  \mbox{or}
  \phantom{AAA}
  \underset{
    \mbox{
      \tiny
      \begin{tabular}{c}
        discussed in
        Remark \ref{Heterotic5Brane}
      \end{tabular}
    }
  }{
    d\,H^{\mathrm{NS}}_{\mathrm{ex}_s}
    \;=\;
    F_{\mathrm{ex}_s} \wedge F_{\mathrm{ex}_s}
  }.
$$
\end{remark}

\begin{remark}[The generator $\theta_2$ trivializing the little-string super-cocycle]
\label{Theta2}
In the formula \eqref{OtherltlstringHomotopy} for the
$D = 5$, $\mathcal{N} =2 $, super-spacetime regarded
as an $\Omega S^2$-quotient of $D = 6$, $\mathcal{N} = (1,1)$, super-spacetime,
there appears a further
generator $\theta_2$ in degree $(2,\mathrm{even})$, whose
differential trivializes the little-string super-cocycle
\begin{equation}
  \label{DifferentialRelationForLittleStringCPField}
  d\,\theta_2
  \;=\;
  \omega_3 + e^5 \wedge \omega_2
  \;=\;
  \mu_{{}_{\rm L1}}
  \,.
\end{equation}
From comparison with the super-cocycles for the
ordinary critical string in $D = 10$ and its D-branes
(see \cite[4.3]{FSS13}\cite[4]{FSS16b}), we note that $\theta_2$
is the universal gauge field flux 2-form on
D-branes for the little-string.
However, there is no analog of the generator $\theta_2$ in the formula
\eqref{SuperExceptionalHalfM5QuotientedByOmegaS2}
for the super-exceptional $\Omega S^2_{\mathrm{HW}}$-quotient
of the super-exceptional $\tfrac{1}{2}\mathrm{M5}$-locus
from Def. \ref{HomotopyQuotientOfSuperExceptionalHalfM5Spacetime}.
But if we consider the Cartesian product of the
super-exceptional $\tfrac{1}{2}\mathrm{M5}$ spacetime
with the classifying space $B S^1$ of a gauge field,
then an $\Omega S^2$-action on this larger space generally contains,
on top of the part left-induced from an $S^1$-action
(Example \ref{LeftInducedOmegaS2Action}), precisely the
extra structure of \eqref{DifferentialRelationForLittleStringCPField}:
\begin{equation}
  \label{CEAlgebraForSuperExceptionalM5QuotientWithGaugeFieldAdded}
  \mathrm{CE}
  \Big(
    \mathfrak{l}
  \Big(
    \big(
      \mathbb{T}^{5,1} \times \mathbb{R}^1
    \big)_{\mathrm{ex}_s}
    \times
    {\color{blue} B S^1}
  \Big)_{\!\!\sslash \Omega S^2_{\mathrm{HW}}}
  \Big)
  \;=\;
  \mathrm{CE}
  \Big(
  \Big(
    \big(
      \mathbb{R}^{5,1} \times \mathbb{R}^1
    \big)_{\mathrm{ex}_s}
  \Big)_{\!\!\sslash \Omega S^2_{\mathrm{HW}}}
  \Big)
  \big[
    {\color{blue}\theta_2}
  \big]
  \Big/
  \big(
    d\,\theta_2 = e^5 \wedge \omega_2 + \omega_3
  \big).
\end{equation}
On this larger space,
equation \eqref{HorizontalizedEquivariant7Cocycle}
for the trivialization of the M5-brane cocycle
completes to an $\Omega S^2_{\mathrm{HW}}$-equivariant trivialization:
\begin{equation}
  \label{EquivariantTrivialization}
  \Big(
    (i_{\mathrm{ex}_s})^\ast
    \mathbf{dL}^{\!\!\mathrm{WZ}}_{\mathrm{ex}_s}
  \Big)_{\!\sslash \Omega S^2_{\mathrm{HW}}}
  \;=\;
  d
  \big(
    \mathbf{L}^{\!\!\mathrm{NG}}_{\mathrm{ex}_s}
    +
    \mathbf{L}^{\!\!\mathrm{PS}}_{\mathrm{ex}_s}
    \wedge
    e^5
    -
    {\color{blue}\theta_2}
    \wedge
    \mathbf{L}^{\!\!\mathrm{tYM}}_{\mathrm{ex}_s}
  \big)
  \phantom{AAAA}
  \mbox{on
  $
  \Big(
    \big(
      \mathbb{T}^{5,1} \times \mathbb{R}^1
    \big)_{\mathrm{ex}_s}
    \times
    {\color{blue} B S^1}
  \Big)_{\!\!\sslash \Omega S^2_{\mathrm{HW}}}
  $\!.
  }
  \end{equation}
\end{remark}

\begin{remark}[The heterotic 5-brane]
  \label{Heterotic5Brane}
  The
  \emph{heterotic super-exceptional $\sfrac{1}{2}\mathrm{M5}$-spacetime}
  $
      \big(
        \mathbb{T}^{5,1\vert\mathbf{16}}
        \times
        \mathbb{T}^1
      \big)_{\mathrm{ex}_s}^{\mathrm{het}}
  $
  is the homotopy fiber of the
  wedge square of the super-exceptional 2-flux \eqref{ExceptionalF}
  on the
  super-exceptional $\sfrac{1}{2}\mathrm{M5}$-spacetime
  (Def. \ref{HalfM5LocusAndItsExceptionalTangentBundle})
  $$
    \xymatrix{
      \big(
        \mathbb{T}^{5,1\vert\mathbf{16}}
        \times
        \mathbb{T}^1
      \big)_{\mathrm{ex}_s}^{\mathrm{het}}
      \ar[dd]_{ \mathrm{hofib}( F_{\mathrm{ex}_s} \wedge F_{\mathrm{ex}_s}) }
      &&
      &&
      \mathclap{
      \mathrm{CE}
      \Big(
        \big(
        \mathbb{R}^{5,1\vert \mathbf{8}}
        \times
        \mathbb{R}^1
        \big)_{\mathrm{ex}_s}
      \Big)
      \big[
        {\color{blue} H^{\mathrm{NS}}_{\mathrm{ex}_s} }
      \big]
      \Big/
      \big(
        d\, H^{\mathrm{NS}}_{\mathrm{ex}_s}
        =
        F_{\mathrm{ex}_s} \wedge F_{\mathrm{ex}_s}
      \big)
      }
      \ar@{<-^{)}}[dd]
      \\
      \\
      \big(
        \mathbb{T}^{5,1\vert\mathbf{16}}
        \times
        \mathbb{T}^1
      \big)_{\mathrm{ex}_s}
      \ar[rr]^-{ F_{\mathrm{ex}_s} \wedge F_{\mathrm{ex}_s} }
      &&
      K(\mathbb{R},4)\;,
      \;\;\;\;\;\;\;\;\;\;\;\;\;\;\;\; &&
      \mathrm{CE}
      \Big(
        \big(
        \mathbb{R}^{5,1\vert \mathbf{8}}
        \times
        \mathbb{R}^1
        \big)_{\mathrm{ex}_s}
      \Big)
      \ar[rr]^-{
        F_{\mathrm{ex}_s}
        \wedge
        F_{\mathrm{ex}_s}
        \mapsfrom
        \; c_4
      }
      &&
      (
        d\,c_4 = 0
      )\;.
    }
  $$
  Then on its super-exceptional
  $\Omega S^2_{\mathrm{HW}}$-compactification as in
  Def. \ref{HomotopyQuotientOfSuperExceptionalHalfM5Spacetime}
  the second summand \eqref{SecondSummandInEquivariuantSuperExceptionalM5Cocycle}
  and with it the full
  $\Omega S^2_{\mathrm{HW}}$-equivariant
  super-exceptional M5-brane cocycle \eqref{HorizontalizedEquivariant7Cocycle}
  becomes exact
  $$
  \Big(
    (i_{\mathrm{ex}_s})^\ast
    \mathbf{dL}^{\!\!\mathrm{WZ}}_{\mathrm{ex}_s}
  \Big)_{\!\!\sslash \Omega S^2_{\mathrm{HW}}}
  \;=\;
  d
  \Big(
    \mathbf{L}^{\!\!\mathrm{NG}}_{\mathrm{ex}_s}
    +
    \mathbf{L}^{\!\!\mathrm{PS}}_{\mathrm{ex}_s}
    \wedge
    e^5
    +
    \underset{
      \eqqcolon
      \mathbf{L}^{\!\!\mathrm{WZ}}_{\mathrm{NS5}}
    }{
      \underbrace{
        \tfrac{1}{2}
        (e^5 \wedge \omega_2 + \omega_3)
        \wedge
        {\color{blue} H^{\mathrm{NS}}_{\mathrm{ex_s} } }
      }
    }
  \Big)
  \,.
  $$
  The new term $\mathbf{L}^{\!\!\mathrm{WZ}}_{\mathrm{NS5}}$
  in the Lagrangian
  that appears this way has the form of the WZ-term for the
  heterotic NS5-brane \cite[(1.5)]{Lechner10}.
\end{remark}

\medskip

In summary, we thus arrive at the picture shown in \eqref{M5LagrangianAsAHomotopy}.

\medskip

\section{Outlook}

In closing, we briefly comment on a few interconnections, issues to be addressed in
the future, and some loose ends.

\medskip

\noindent
{\bf Combining local super-exceptional geometry with global topology.}
The discussion in this article focuses
on the situation of vanishing bosonic 4-flux, keeping
only the super-components of the 4-flux, being the M2-brane cocycle
$\mu_{{}_{\rm M2}}$ \eqref{TheMBraneCocycles}.
We had discussed the opposite case of pure bosonic flux
 in \cite{FSS19c}. In that case, the subtlety is all in the
global topological structure of the Hopf-Wess-Zumino term
controlled by \emph{Cohomotopy} cohomology theory \cite{FSS19c} \cite{FSS19b},
while here the subtlety is all in
the local differential structure of the
Perry-Schwarz Lagrangian,
hence of 5d super Yang-Mills plus KK-modes:

\begin{equation}
\label{FullCohomotopy}
\mbox{
\begin{tabular}{|c||c|c|}
  \hline
  \multicolumn{3}{|c|}{
   {\bf Single M5-brane sigma-model}
  }
  \\
  \hline
  \emph{Aspect}
  &
  \begin{tabular}{c}
   \emph{Global topological}
  \end{tabular}
  &
  \begin{tabular}{c}
    \emph{Local differential}
  \end{tabular}
  \\
  \hline
  \hline
  \emph{Lagrangian term}
  &
  Hopf-Wess-Zumino
  &
  \begin{tabular}{rcl}
    Perry-Schwarz &=& 5d Yang-Mills + KK
    \\
    \& & & 5d top. Yang-Mills
  \end{tabular}
  \\
  \hline
  \emph{Controlled by}
  &
  \begin{tabular}{c}
    twisted
    \\
    Cohomotopy
  \end{tabular}
  &
  \begin{tabular}{c}
    $\Omega S^2$-equivariant
    \\
    super-exceptional geometry
  \end{tabular}
  \\
  \hline
  \emph{Discussed in}
  &
  \cite{FSS19c}
   &
   \cref{SuperExceptionalReductionOfMTheoryCircle},
   \cref{SuperExceptionalM5Lagrangian}
  \\
  \hline
\end{tabular}
}
\end{equation}

In a full picture
of the M5-brane sigma-model, both the super-exceptional local
geometry and the cohomotopical global structure are to
be combined. This will be discussed elsewhere.

\newpage

\noindent {\bf Non-abelian gauge enhancement.}
While we have only discussed abelian gauge fields here,
their appearance (by the construction in \cref{SuperExceptionalLagrangian},
\cref{SuperExceptionalM5Lagrangian}) via the M2/M5 super-cocycle
$\mu_{{}_{\rm M2/M5}}$ \eqref{TheMBraneCocycles} means that the
super-cohomotopical gauge enhancement mechanism
found in \cite{BSS18} applies. Together with the constraints of
 half-integral flux quantization and
  tadpole cancellation,
which we demonstrate in \cite{FSS19b} and \cite{SS19}, respectively,
to follow from
C-field charge quantization in full Cohomotopy \eqref{FullCohomotopy},
this should
generate the expected types of non-abelian gauge fields, both
for heterotic M-theory as well as for coincident 5-branes,
in the form discussed in \cite{tcu}\cite{Sati10}\cite{FSS12a}\cite{FSS12b}.
While the details remain to be worked out, this opens up the
possibility of a concrete strategy for systematically deriving the
non-abelian D=6 M5-brane theory from first principles.

A proposal for the gauge sector of a non-abelian M5-brane Lagrangian based
on the higher structures of \cite{Sati10}\cite{FSS12a}\cite{FSS12b}
has recently been made in \cite{SL18}\cite{SL19}, but the details
are notoriously subtle
and hard to get right by an educated Ansatz,
see \cite{Saemann19}.
With the above strategy this Ansatz might eventually
be connected to a systematic derivation.

\medskip

\noindent {\bf Covariant enhancement.}
The constructions in this article relate to the non-covariant M5-brane Lagrangian
of \cite{PerrySchwarz97}\cite{Schwarz97}\cite{APPS97},
(following \cite{HenneauxTeitelboim88}), but not manifestly to the covariant
formulation of \cite{PastiSorokinTonin96}\cite{PastiSorokinTonin97} \cite{BLNPST97}.

\vspace{2mm}
\noindent \begin{tabular}{ll}
\begin{minipage}[l]{9cm}
However, as explained in \cite[pp. 6-7]{HSS18}, in the vein of understanding
supergravity as super Cartan geometry \cite{Lott90}\cite{EgeilehChami12},
we are to think of the super-exceptional Minkowski spacetimes
considered here, such as the super-exceptional
$\tfrac{1}{2}\mathrm{M5}$-spacetime illustrated in
\eqref{HalfM5BraneSetup}, as being the
infinitesimal moving frames on a super-exceptional curved
Cartan geometry. Furthermore, the covariant form of this
geometry is to be obtained systematically by actually moving the frames,
subject only to the condition of vanishing super-torsion,
which, by \cite{CandielloLechner94}\cite{Howe97}
(see \cite[Sec. 2.4]{CGNT04}),
is equivalently the condition that the
equations of motion of 11d supergravity hold.
\end{minipage}
&
\raisebox{-60pt}{
\includegraphics[width=.5\textwidth]{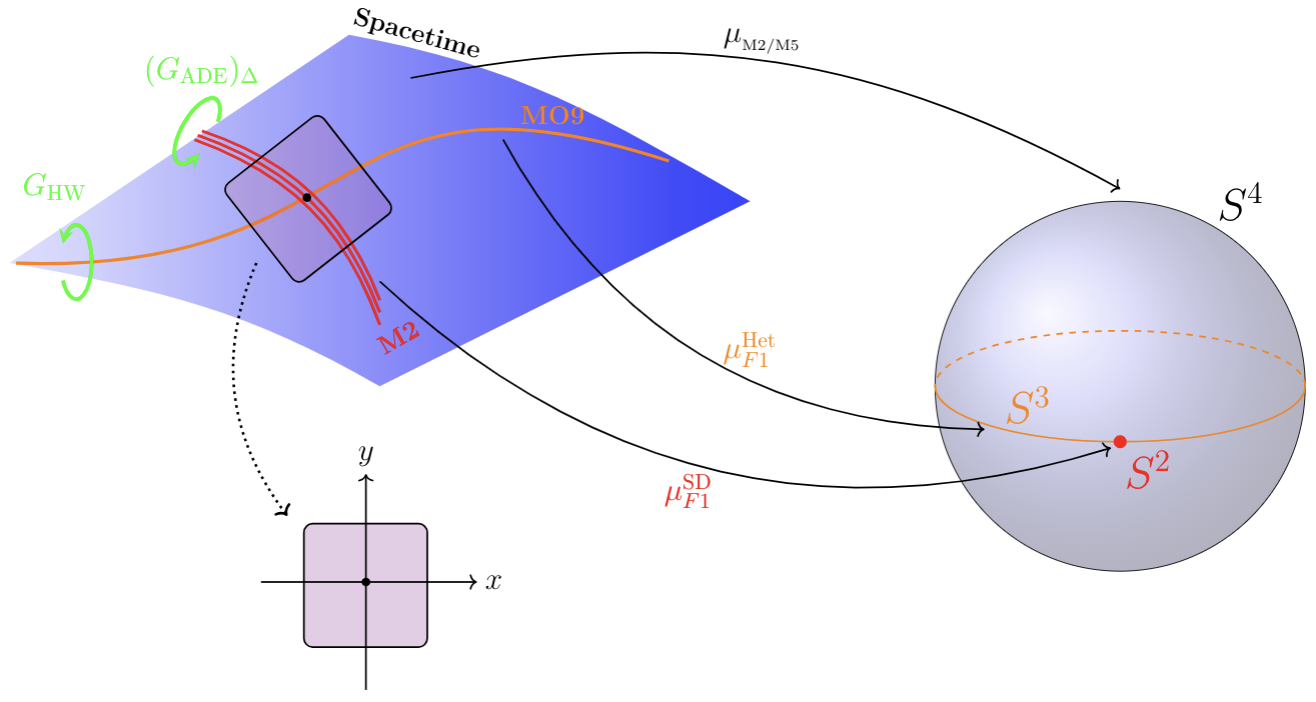}
}
\end{tabular}

\vspace{2mm}
Thus super-Cartan geometry provides
a systematic and precise
mechanism for globalizing/covariantizing all the
local brane constructions given here,
as well as the related constructions in \cite{HSS18}\cite{BSS18},
reviewed in \cite{FSS19a}.
While this provides the algorithm, it still has to be
run, the details still have to be worked out elsewhere.

\medskip

\noindent {\bf Necessity of the $\tfrac{1}{2}\mathrm{M5}$ configuration?}
The string theory folklore claims (e.g. \cite{DHTV14}) that there are two
ways of geometrically engineering $D = 6$ $\mathcal{N} = (1,0)$ theories,
one of them being the $\tfrac{1}{2}\mathrm{M5}$-brane
configuration of Remark \ref{TheHalfM5Spacetime}.
Since our rigorous construction confirms this expectation,
it would be interesting to precisely classify this -- on a mathematical basis --
and possibly determine further available choices for the construction.
This needs further thinking.
At this point, we may at least highlight precisely where the
spinor projections of Def. \ref{SpinorProjection}, which
define the $\tfrac{1}{2}\mathrm{M5}$-locus, enter our proofs:

\begin{enumerate}
\vspace{-2mm}
\item
The assumption that the spinor 1-forms are fixed under $\pmb{\Gamma}_{6789}$
is what makes the super-exceptional lift $v_{a}^{\mathrm{ex}_s}$
of the isometric flow along the compactification circle exist:
it is used in the last line of
\eqref{MakeSuperExceptionalIsometryWork} in the proof of
Prop. \ref{LiftOfVectorField}.
There is at least one other brane configuration where
an analogous construction works, namely the M-wave inside the MO9
(Remark \ref{MWaveInMO9}). But there could be more possible variants.

\vspace{-2mm}
\item
The assumptions that spinor 1-forms are fixed under
$\pmb{\Gamma}_{5}$ is what implies
the technical Lemma \ref{PullbackOfM2CocycleToHalfM5}.
This Lemma appears crucially in various of the following proofs
leading up to and including the main Theorem \ref{TheTrivialization},
notably it is used in both \eqref{FirstSummandForContractedCocycle}
and \eqref{SummandForContractedCocycleSecond} proving
Prop. \ref{TrivializationOf7CocycleOnExceptionalHalfM5}
and then again in \eqref{abcd} proving Theorem \ref{TheTrivialization}.

This makes it seem at least unlikely that this assumption could be
removed
while still retaining the result of a super-exceptional embedding
construction.
\end{enumerate}

\medskip

\noindent {\bf DBI corrections in $\alpha'$ and Super-exceptional 5-form.}
In our analysis of the super-exceptional Perry-Schwarz Lagrangian
in \cref{SuperExceptionalLagrangian}, we have
restricted to the special case that the value of the
5-index tensor $e_{a_1 \cdots a_5}$ on super-exceptional
spacetime vanishes.  The formulas \eqref{Hexs}
and \eqref{DecomposedCFieldAsSumOfBosonicAndFermionicContribution}
show that, without this assumption, the
super-exceptional Perry-Schwarz Lagrangian picks up further
correction terms. At the same time, we have studied here the expected
DBI corrections to the brane Lagrangian starting at quartic order in the field strength $F$
\cite{FradkinTseytlin85} (see \cite{Tseytlin00}).
It is natural to conjecture that these two effects are related,
but this needs more investigation.

\medskip

\noindent {\bf AKSZ sigma-model description.}
The article \cite{Arvanitakis18} discusses
the local differential structure of the M5 Hopf-Wess-Zumino term
in view of exceptional geometry from the point of
view of AKSZ sigma-models. The development seems complementary
to ours here, but it would be interesting to see if there is a connection.

\medskip

\medskip

\noindent {\large \bf Acknowledgements.}

\noindent  The authors would like to thank
Alex Arvanitakis, Igor Bandos, and Dmitri Sorokin  for useful
comments on the first draft. For the final stages of this project, H. S. gratefully
acknowledges:
``This work was performed in part at Aspen Center for Physics, which is
supported by National Science Foundation grant PHY-1607611. This
work was partially supported by a grant from the Simons Foundation."

\medskip


\vspace{1cm}
\noindent Domenico Fiorenza,  {\it Dipartimento di Matematica, La Sapienza Universita di Roma, Piazzale Aldo Moro 2, 00185 Rome, Italy.}

 \medskip
\noindent Hisham Sati, {\it Mathematics, Division of Science, New York University Abu Dhabi, UAE.}

 \medskip
\noindent Urs Schreiber,  {\it Mathematics, Division of Science, New York University Abu Dhabi, UAE, on leave from Czech Academy of Science, Prague.}

\end{document}